\documentclass[a4paper,10pt,twoside]{article}

\usepackage{a4wide}
\usepackage{lmodern}
\usepackage{amsmath,amssymb,amsfonts,amsthm}
\usepackage{moreverb,rotating,graphics}
\usepackage{hyperref}
\usepackage[T1]{fontenc}
\usepackage[utf8]{inputenc}
\usepackage{epsfig}
\usepackage[francais,english]{babel} 
\usepackage{framed,fancybox}
\usepackage{tikz,pgfplots}
\usepackage{graphicx}
\usepackage{caption}
\usepackage{subcaption}
\usepackage{mathrsfs}
\usepackage{esint}
\usepackage{dsfont}
\usepackage{enumerate}

\newtheorem{theorem}{Theorem}
\newtheorem{proposition}[theorem]{Proposition}
\newtheorem{remark}[theorem]{Remark}
\newtheorem{lemma}[theorem]{Lemma}
\newtheorem{corollary}[theorem]{Corollary}
\newtheorem{definition}[theorem]{Definition}

\newcommand\1{{\ensuremath {\mathds 1} }} 

\def\RR{{\mathbb R}}

\def\LL{{\mathbb L}}
\def\UU{{\mathbb U}}
\def\FF{{\mathbb F}}
\def\HH{{\mathbb H}}
\def\N{{\mathbb N}}

\def\R{{\mathbb R}}
\def\C{{\mathbb C}}

\def\bA{{\bold A}}

\def\bB{{\bold B}}

\def\bR{{\bold R}}

\def\br{{\bold r}}

\def\bk{{\bold k}}

\def\br{{\bold X}}
\def\bx{{\bold x}}
\def\br{{\bold r}}

\def\Im{\mathrm{Im }\, }
\def\ri{{\mathrm{i}}}
\def\rc{{\mathrm{c}}}
\def\rd{{\mathrm{d}}}
\def\re{{\mathrm{e}}}
\def\rp{{\rm p}}
\def\rh{{\rm h}}
\def\Re{{\mathrm{Re }\, }}

\def\Tr{{\rm Tr}}
\def\div{{\rm div \,}}
\def\cA{{\mathcal A}}
\def\cB{{\mathcal B}}
\def\cC{{\mathcal C}}

\def\cF{{\mathcal F}}
\def\cG{{\mathcal G}}
\def\cH{{\mathcal H}}
\def\cK{{\mathcal K}}

\def\cS{{\mathcal S}}

\def\fS{{\mathfrak S}}

\def\fH{{\mathfrak H}}

\def\fB{{\mathfrak B}}

\def\ffg{{\mathfrak g}}
\def\fs{{\mathfrak s}}

\newcommand{\dps}{\displaystyle}
\newcommand{\bra}{\langle}
\newcommand{\ket}{\rangle}

\newcommand{\ess}{\mathrm{ess}}
\newcommand{\xc}{\mathrm{xc}}

\newcommand{\loc}{\mathrm{loc}}
\newcommand{\pv}{\mathrm{p.v.} }
\newcommand{\sgn}{\mathrm{sgn}}

\newcommand{\ext}{\mathrm{ext}}
\newcommand{\sym}{\mathrm{sym}}
\newcommand{\Supp}{\mathrm{Supp}}
\newcommand{\GW}{\mathrm{GW}}
\newcommand{\PPM}{\mathrm{PPM}}
\newcommand{\app}{\mathrm{app}}

\newcommand{\sS}{\mathscr{S}}
\newcommand{\sC}{\mathscr{C}}
\newcommand{\sD}{\mathscr{D}}
\newcommand{\vc}{v_{c}}

\renewcommand{\leq}{\leqslant}
\renewcommand{\le} {\leqslant}
\renewcommand{\geq}{\geqslant}
\renewcommand{\ge} {\geqslant}

\def\sqw{\hbox{\rlap{\leavevmode\raise.3ex\hbox{$\sqcap$}}$%
\sqcup$}}
\def\cqfd{\ifmmode\sqw\else{\ifhmode\unskip\fi\nobreak\hfil
\penalty50\hskip1em\null\nobreak\hfil\sqw
\parfillskip=0pt\finalhyphendemerits=0\endgraf}\fi}

\hypersetup{
    colorlinks,
    linkcolor={blue!50!blue},
    citecolor={red}
}
\renewcommand{\eqref}[1]{(\ref{#1})}

\usepackage[normalem]{ulem}
\usepackage{color}

\begin{document}

\title{A mathematical analysis of the $\GW^0$ method for computing electronic excited energies of molecules}

\author{
  Eric Canc\`es, David Gontier and Gabriel Stoltz \\
  Universit\'e Paris-Est, \'Ecole des Ponts and INRIA, F-77455 Marne-la-Vall\'ee
}

\maketitle

\begin{abstract}
This paper analyses the GW method for finite electronic systems. In a first step, we provide a mathematical framework for the usual one-body operators that appear naturally in many-body perturbation theory. We then discuss the GW equations which construct an approximation of the one-body Green's function, and give a rigorous mathematical formulation of these equations. Finally, we study the well-posedness of the $\GW^0$ equations, proving the existence of a unique solution to these equations in a perturbative regime.
\end{abstract}

\tableofcontents

\newpage

\section{Introduction}

Computational quantum chemistry is nowadays a standard tool to numerically determine the properties of molecules. The Density Functional Theory (DFT) first developed by Hohenberg and Kohn~\cite{HK64} and by Kohn and Sham~\cite{KS65}, is a very powerful method to obtain ground state properties of molecular systems. However, it does not allow one to compute optical properties and electronic excited energies. In order to calculate such quantities, several approaches have been considered in the last decades~\cite{Onida2002}. Among them are the time-dependent DFT (TDDFT)~\cite{TDDFT2, TDDFT}, wave-function methods~\cite{Helgaker2014} such as Coupled-Cluster, full-CI and Green's function methods. In this article, we study the GW method, which is based on Hedin's equations for the one-body Green's function~\cite{Hedin1965}. The formal derivation of the latter equations relies on many-body perturbation techniques. While the GW method has been proven very successful in practice to predict electronic-excited energies, no rigorous mathematical framework has yet been developed to understand its mathematical properties. The aim of this work is to present such a framework. \\

In non-relativistic first-principle molecular simulation, the electrons of a molecular system are described by an $N$-body Hamiltonian operator $H_N$, which is a bounded below self-adjoint operator on the fermionic space $\bigwedge^N L^2(\R^3)$ (see Equation~\eqref{eq:H_N} below). Whenever $N \le Z$, where $Z$ is the total nuclear charge of the molecular system, $H_N$ has an infinity of discrete eigenvalues $E_N^0 \le E_N^1 \le E_N^2 \le \cdots$ below the bottom of the essential spectrum, where $E_N^0$ is its ground state energy. The quantities we would like to evaluate are the \textit{electronic-excitation energies}
\[
E_N^0 - E_{N+1}^k \quad \text{(gain of an electron)},
\quad \text{and} \quad
E_{N}^0 - E_{N-1}^k \quad \text{(loss of an electron)}.
\]
These energy differences are not to be confused with the optical-excitation energies, which are energy differences of the form $E_N^{k} - E_N^0$, between two states with the same number of electrons. More generally, it is interesting to compute the particle electronic-excitation set $S_\rp := \sigma \left( H_{N+1} - E_N^0 \right)$ and the hole electronic-excitation set $S_\rh := \sigma \left( E_{N}^0 - H_{N-1} \right)$. As will be made clear in Section~\ref{subsec:GF}, these sets are closely linked to the one-body Green's function: the time-Fourier transform of the Green's function becomes singular on these sets. In order to study the electronic-excitation sets, we therefore study the one-body Green's function. Also, the one-body Green's function is a fundamental object which contains a lot of useful information, and allows one to easily compute the ground state electronic density, the ground state one-body density matrix, and even the ground state energy thanks to the Galitskii-Migdal formula~\cite{Galiskii1958}. \\

Calculating the one-body Green's function is however a difficult task. In his pioneering work in 1965, Hedin proved that the Green's function satisfies a set of (self-consistent) equations, now called the Hedin's equations~\cite{Hedin1965}. These equations link many operator-valued distributions, namely the reducible and irreducible polarizability operators, the dynamically screened interaction operator, the self-energy operator, the vertex operator, and of course the one-body Green's function. The state-of-the-art method to compute the one-body Green's function consists in solving Hedin's equations.\\

Immediately, two difficulties arise. The first one is related to the lack of regularity of the Green's function (we expect its time-Fourier transform $\widehat{G}$ to be singular on the electronic-excitation sets). One way to get around this problem is to consider the analytical extension of $\widehat{G}$ into the complex plane, which we denote by $\widetilde{G}$. This is possible whenever the following classical stability condition holds true\footnote{The question ``Is the stability condition always true for Coulomb systems'' is still an open problem~\cite[Part VII]{Bach2014}.}:
\[
\boxed{\textbf{Stability assumption:} \text{ It holds that } 2E_N^0 < E_{N+1}^0 + E_{N-1}^0.}
\]
The physical relevance of this inequality is discussed for instance in~\cite[Section~4.2]{Farid99}. It allows one to define the chemical potential $\mu$, chosen such that
\[
E_N^0 - E_{N-1}^0 < \mu < E_{N+1}^0 - E_N^0.
\]
Instead of studying the Green's function $G(\tau)$ in the time domain, or its Fourier transform $\widehat{G}(\omega)$ in the frequency domain, we rather study its analytical continuation $\widetilde{G}$ on the imaginary axis $\mu + \ri \R$. The function $\omega \mapsto \widetilde{G}(\mu + \ri \omega)$ enjoys very nice properties, both in terms of regularity and integrability, which makes it a privileged tool for numerical calculations.\\

The second difficulty comes from the fact that Hedin's equations cannot be exactly solved and, even more importantly, that the mathematical definition of some terms in these equations are unclear. It however opens the way to some approximate resolutions. The most widely used approximation nowadays is the so-called GW-approximation, also introduced by Hedin~\cite{Hedin1965}. These equations are traditionally set on the time-axis, or on the energy-axis~\cite{RJT10,KFS14}. However, as previously mentioned, the various operators under consideration are singular on these axes, which makes the traditional GW equations cumbersome to implement numerically, and difficult to analyze mathematically. In order to manipulate better-behaved equations, it is more convenient to replace every operator-valued distribution involved in the GW equations by its analytic continuation on an appropriate imaginary axis, thanks to the ``contour deformation'' technique introduced in~\cite{Rojas1995, RSWRG99}. The resulting GW equations, which give an approximation of the map $\omega \mapsto \widetilde{G}(\mu + \ri \omega)$, turn out to give simulation results in very good agreement with experimental data~\cite{SDL06,SDL09,CRRRS12,CRRRS13}. \\

From the GW equations set on the imaginary axis, several further approximation may be performed. The GW equations are solved self-consistently, and the Green's function is updated at each iteration until convergence. When only one iteration is performed, we obtain the one-shot GW approximation, also called the G$_0$W$^0$ approximation of the Green's function. For molecules, self-consistent GW approaches give results of similar quality as G$_0$W$^0$, sometimes almost identical~\cite{SDL09,KFS14}, sometimes slightly worse~\cite{RJT10}, sometimes slightly better~\cite{CRRRS12,CRRRS13}. When several iterations are performed, while keeping the screening operator $W$ fixed, equal to a reference screening operator $W^0$, we obtain the $\GW^0$ approximation of the Green's function~\cite{SDL09, VonBarth1996}. Since the update of the screening operator $W$ in a self-consistent GW scheme seems difficult to analyze mathematically, we prefer to study in this article the equations resulting from the $\GW^0$ approximation. \\

The purpose of this article is threefold. First, we clarify the mathematical definitions and properties of the usual one-body operators involved in many-body perturbation theory. Then, we embed the $\GW^0$ equations in a mathematical framework. Finally, we prove that, in a perturbative regime, the $\GW^0$ equations admit a unique solution close to a reference Green's function. \\

From a physical viewpoint, the analysis we perform in this work is more relevant for atoms and molecules. Indeed, as discussed in \cite[Section~4.1]{BG14} for instance, fully self-consistent GW approaches are questionable for solid-state systems, for which quasiparticle methods are preferred~\cite{AG98,AJW99}. \\

The paper is organized as follows. In Section~\ref{sec:setting_stage}, we provide the mathematical tools that will be used throughout the article. We recall the Titchmarsh's theorem, and introduce the kernel-product of two operators, which can be seen as an infinite dimensional version of the Hadamard product for matrices. We also explain the underlying structure that makes the ``contour deformation'' possible. In Section~\ref{sec:operators}, we recall the standard definitions of the usual one-body operators that appear in many-body perturbation theory. A consistent functional setting is given for each of these operators, and their basic properties are recalled and proved. Section~\ref{sec:GW_finite} is concerned with the $\GW$ approximation. We explain why some of the GW equations are not well-understood mathematically, and prove that the $\GW^0$ equations are well-posed in a perturbative regime. Most of the proofs are postponed until Section~\ref{sec:proofs}.

\section{Setting the stage}
\label{sec:setting_stage}

\subsection{Some notation}
\label{sec:notation}

The GW method is based on time-dependent perturbation theory and therefore involves space-time operators. Following the common notation in physics, we denote by $t$ the time coordinate, by $\br$ the space coordinates, and by $\bx$ or $\br t$ the space-time coordinates. The functional spaces considered in this work are by default composed of complex-valued functions, unless we explicitly mention that the functions are real-valued.

Most of the space-time operators appearing in the GW formalism are time-translation invariant. A time-translation invariant operator $\cC$ can be characterized by the family of operators $(C(\tau))_{\tau \in \R}$ such that, formally, the kernel of $\cC$ is of the form $\cC(\br_1 t_1,\br_2 t_2)=C(\br_1,\br_2,t_1-t_2)$, where $C(\br,\br',\tau)$ is the kernel of the operator $C(\tau)$. For clarity, we will systematically use the letter $\tau$ to denote a time variable which is in fact a time difference. 

Let $\cH$ be a separable complex Hilbert space, whose associated scalar product is simply denoted by~$\langle \cdot,\cdot \rangle$ and the associated norm $\| \cdot \|$. We denote by $\cB(\cH)$ the space of bounded linear operators on~$\cH$, by $\cS(\cH)$ the space of bounded self-adjoint operators on $\cH$, by $\fS_p(\cH)$ ($1 \le p < \infty$) the Schatten class 
\[
\fS_p(\cH) = \left\{ A \in \cB(\cH) \; \left| \; \|A\|_{\fS_p(\cH)}:= \Tr(|A|^p)^{1/p}<\infty \right.\right\},
\]
and by $A^\ast$ the adjoint of a linear operator $A$ on $\cH$ with dense domain. The real and imaginary parts of an operator $A \in \cB(\cH)$ are defined as
\[
\Re A = \frac{A+A^\ast}2, \quad \Im A = \frac{A-A^\ast}{2\ri}.
\]
Note that, when $A$ is closed (which implies $A^{**} = A$), the operators $\Re A$ and $\Im A$ are self-adjoint. For $f,g \in \cH$ and given operators~$A,B$ on $\cH$, we will often use the notation
\[
	\langle f | A | g\rangle_\cH := \langle f, Ag\rangle_\cH, 
\qquad 
	\langle f | AB | g\rangle_\cH := \langle f, ABg\rangle_\cH ,
\]
even in cases when the operators are not self-adjoint. Operators are always understood to act on the function on the right in this notation.

We will sometimes need to manipulate the adjoints of operators between two different Hilbert spaces $\cH_a$ and $\cH_b$.
The adjoint of a bounded operator~$A \in \cB(\cH_a,\cH_b)$ is the bounded operator $A^* \in \cB(\cH_b,\cH_a)$ defined by
\[
\forall (x,y) \in \cH_a \times \cH_b, 
\qquad 
\left(A^* y, x \right)_{\cH_a} = \left(y, Ax \right)_{\cH_b}.
\]

Let $E$ be a Banach space. We denote by $\mathscr{S}'(\R, E)$ the space of $E$-valued tempered-distributions on $\R$, \textit{i.e.} the set of continuous linear maps from the Schwartz's functional space $\mathscr{S}(\R)$ into $E$. Recall that, by definition, a family $(T_\eta)_{\eta > 0}$ of elements of $\mathscr{S}'(\R, E)$ converges in $\mathscr{S}'(\R, E)$ to some $T \in \mathscr{S}'(\R, E)$ when $\eta$ goes to $0$ if
\[
\forall \phi \in \sS(\R), \qquad \left\| \langle T_\eta,\phi \rangle_{\sS',\sS} -  \langle T,\phi \rangle_{\sS',\sS} \right\|_E \mathop{\longrightarrow}_{\eta \to 0^+} 0.
\]

Let $f \in L^1(\R, E)$ be a time-dependent $E$-valued integrable function. The time-Fourier transform of $f$ is defined, using the standard convention in physics, as
\begin{equation}
  \label{eq:def_F_T}
\forall \omega \in \R, \quad   \widehat{f} (\omega) := \left( \cF_T f \right)(\omega) := \int_\R f(\tau) \re^{\ri \omega \tau} \,  \rd \tau.
\end{equation}
For the sake of clarity, we will sometimes denote by $\R_t$ or $\R_\tau$ the time-domain, by $\R_\omega$ the frequency-domain, by $\sS'(\R_\tau,E)$ (resp. $\sS'(\R_\omega,E)$) the space of time-dependent (resp. frequency-dependent) $E$-valued distributions, etc. We will also denote with a hat the functions defined on the frequency domain. Using this notation, $\cF_T$ can be extended to a bicontinuous isomorphism from $\mathscr{S}'(\R_\tau, E)$ into $\mathscr{S}'(\R_\omega, E)$. When $\widehat f \in L^1(\R_\omega,E)$, we have
$$
\forall \tau \in \R, \quad  \left( \cF_T^{-1} \widehat f \right)(\tau) =  \frac{1}{2\pi} \int_\R \widehat f(\omega) \re^{-\ri \omega \tau} \,  \rd \omega.
$$

The Dirac distribution at $a \in \R^d$ is denoted by $\delta_{a}$, and the Heaviside function on $\R$ by $\Theta$: 
\begin{equation}
  \label{eq:Heaviside}
  \Theta(\tau) = 1 \, \mathrm{for} \ \tau>0, 
  \qquad \Theta(\tau) = 0 \ \mathrm{for} \ \tau < 0, 
  \qquad 
  \Theta(0) = 1/2.
\end{equation}
Recall that the time-Fourier transform of $\Theta$ is, in the tempered distributional sense,
\begin{equation} \label{eq:fourier_heaviside}
\widehat{\Theta}(\omega) = \pi \delta_0(\omega) + \ri \pv \left( \dfrac{1}{\omega} \right),
\end{equation}
where $\pv$ is the Cauchy principal value.
We will also make use of the notation $\tau^+$ for a number strictly above $\tau$, but infinitesimally close to $\tau$, and of the convention
$$
\Theta(\tau) \delta_0(\tau^+) :=  \delta_0(\tau), \quad  \Theta(-\tau) \delta_0(\tau^+) :=  0.
$$

 
\subsection{Hilbert transform of functions and distributions}

The Hilbert transform, which amounts to a convolution by $\pi^{-1} \pv(\frac1 \cdot)$, plays a crucial role in the GW formalism. We first recall some well-known results on the standard Hilbert transform on $L^p(\R_\omega)$, and extend the results to the Sobolev spaces $H^s(\R_\omega)$ for $s \in \R$. Usually, the name ``Hilbert transform'' is only used on functional spaces $E \subset L^1_\mathrm{loc}(\R_\omega)$ such that, for any function $\widehat f \in E$, the limit
\[
\left[\widehat f \ast \pv\left( \dfrac{1}{\cdot} \right)\right](\omega) = \pv \int_{-\infty}^{+\infty} \dfrac{\widehat f(\omega')}{\omega - \omega'} \, \rd\omega'  := \lim_{\eta \to 0^+} \int_{\R \setminus [\omega-\eta,\omega+\eta]} \dfrac{\widehat f(\omega')}{\omega - \omega'} \, \rd \omega'
\]
exists for almost all $\omega \in \R_\omega$. However, in the sequel, we will also use the name ``Hilbert transform'' in functional spaces where the above integral representation is not always valid (for instance when $\widehat f$ is not a locally integrable function). Note that we define the Hilbert transform on Fourier transforms of functions (\textit{i.e.} on functions on the frequency domain) since this is the typical setting in the GW formalism.

\subsubsection{Hilbert transform in $L^p$ spaces}

We first begin with the following classical definition (see for instance~\cite[Section~4.1]{Grafakos2004}).

\begin{definition}[Hilbert transform on $\sS(\R_\omega)$] The Hilbert transform of a function $\widehat \phi \in \sS(\R_\omega)$ is defined by
\begin{equation}\label{eq:defHT1}
 \fH \widehat\phi  := \frac 1 \pi \pv \left( \frac{1}{\cdot} \right) \ast \widehat\phi, 
 \end{equation}
 or equivalently by
 \begin{equation}\label{eq:defHT2}
  \fH \widehat\phi  :=   \left(\cF_T(-\ri \, \mbox{\rm sgn}(\cdot)) \,  \cF^{-1}_T \right) \widehat\phi,
\end{equation}
where $\pv \left( \frac{1}{\cdot} \right)$ is the Cauchy principal value of the function $\omega \mapsto \frac{1}{\omega}$, $\ast$ the convolution product, $\cF_T$ the Fourier transform defined in~\eqref{eq:def_F_T} and $-\ri \, \mbox{\rm sgn}(\cdot)$ the multiplication operator by the $L^\infty$ function $t \mapsto -\ri \, \mathrm{sgn}(t)$ (where $\mathrm{sgn}(t) = \Theta(t)-\Theta(-t)$ is the sign function). 
\end{definition} 

The Hilbert transform can be extended by continuity to a large class of tempered distributions. We refer to \cite{Grafakos2004, Riesz1928} for a proof of the following theorem. 

\begin{theorem} 
  \label{th:Riesz}
  For all $\widehat f \in L^p(\R_\omega)$ with $1 < p < \infty$, the Hilbert transform 
  \[
  \fH \widehat{f}(\omega) = \pv \int_{-\infty}^{\infty} \dfrac{\widehat f(\omega')}{\omega - \omega'} \, \rd \omega'
  \]
  is well-defined for almost all $\omega \in \R$. It holds $\fH \in \cB(L^p(\R_\omega))$ with
  \[
  \|\fH\|_{\cB(L^p(\R_\omega))} = \left| \begin{array}{lll} \tan(\pi/(2p)) & \mbox{ if } 1 < p \le 2, \\ \mbox{\rm cotan}(\pi/(2p)) & \mbox{ if } 2 \le  p < \infty. \end{array}\right.
  \]
  Moreover, the Hilbert transform commutes with the translations and the positive dilations, and anticommutes with the reflexions. Finally, it is a unitary operator on $L^2(\R_\omega)$.
\end{theorem}

\subsubsection{Hilbert transform in Sobolev spaces}

Recall that for any $s \in \R$, the Sobolev space $H^s(\R_\omega)$ is the Hilbert space defined as
\[
H^s(\R_\omega) := \left\{ \widehat f \in \mathscr{S}'(\R_\omega)) \; \Big| \;  (1+|\cdot|^2)^{s/2}  \mathcal{F}_T^{-1} \widehat{f} \in L^2(\R_\tau) \right\},
\]
and endowed with the scalar product 
$$
 \left\bra \widehat f,\widehat g \right\ket_{H^s} = 2 \pi \int_{-\infty}^{+\infty} (1 + \tau^2)^s  \overline{(\mathcal{F}_T^{-1}\widehat f)(\tau)} \, (\mathcal{F}_T^{-1}\widehat g)(\tau) \,  \rd \tau,
$$
and that $H^{-s}(\R_\omega)$ can be identified with the dual of $H^{s}(\R_\omega)$ when the space $L^2(\R_\omega)= H^{0}(\R_\omega)$ is used as a pivoting space. One of the reasons to introduce these spaces is that the image of $L^\infty(\R_\tau)$ by the Fourier transform~$\mathcal{F}_T$ is contained in the Sobolev spaces of indices strictly lower than $-1/2$.

\begin{lemma}[Fourier transform in $L^\infty(\R_\tau)$]
  \label{lem:image_Fourier_L_infty}
  Let $s > 1/2$. Then $\dps \mathcal{F}_T\left(L^\infty(\R_\tau)\right) \subset H^{-s}(\R_\omega)$  and
  \begin{equation}
    \label{eq:def_Cs}
  \left\| \mathcal{F}_T \right\|_{\mathcal{B}(L^\infty,H^{-s})} = C_s \quad \mbox{with} \quad 
    C_s =  \left(2\pi \int_\R \frac{\rd\tau}{(1+\tau^2)^s}  \right)^{1/2}.
  \end{equation}
\end{lemma}
For completeness, we recall the proof of Lemma~\ref{lem:image_Fourier_L_infty} in Section~\ref{sec:image_Fourier_L_infty}. 

\medskip

Since the Hilbert transform in $\mathscr{S}(\R_\omega)$ amounts to a multiplication by the bounded function $-\ri \, \sgn(\cdot)$ in the time domain (see~\eqref{eq:defHT2}), it can be directly extended to the Sobolev spaces $H^s(\R_\omega)$.

\begin{lemma} 
  \label{lem:hilbert_hm1} For any $s \in \R$, the Hilbert transform~$\fH$ is a unitary operator on the Sobolev spaces~$H^s(\R_\omega)$ satisyfing $\fH^{-1}=-\fH$ (and therefore $\fH^2 = -\mathrm{Id}$).
\end{lemma}

\begin{remark}[Hilbert transform of distributions]
  \label{rem:DLp'}
  Extending the Hilbert transform to Sobolev spaces is straightforward using (\ref{eq:defHT2}). Extensions of the Hilbert transform to other subspaces of $\sD'(\R_\omega)$, such as the $\sD'_{L^p}(\R_\omega)$ spaces defined in~\cite[Section~VI.8]{Schwartz}, can be obtained from (\ref{eq:defHT1}).  
\end{remark}

\subsubsection{Hilbert transforms of operator-valued distributions}

We now need to properly define the Hilbert transform of operator-valued distributions on the frequency domain, as such objects naturally appear in the GW formalism. We first introduce, for $s \in \R$, the Banach space
\[
H^s(\R_\omega,\cB(\cH)) := \Big\{ \widehat A \in \mathscr{S}'(\R_\omega,\cB(\cH))) \; \Big| \;  (1+|\cdot|^2)^{s/2}  \mathcal{F}_T^{-1} \widehat{A} \in L^2(\R_\tau,\cB(\cH)) \Big\},
\]
endowed with the norm
\[
	\left\| \widehat A \right\|_{H^s(\R_\omega,\cB(\cH))} = \sqrt{2 \pi} \left( \int_{-\infty}^{+\infty} (1 + \tau^2)^s  \left\| \left( \mathcal{F}_T^{-1}\widehat A \right)(\tau) \right\|_{\cB(\cH)}^2 \,  \rd\tau \right)^{1/2}.
\]
The following definition makes sense in view of Lemma~\ref{lem:hilbert_hm1}.

\begin{definition}[Hilbert transforms of frequency-dependent operators] 
\label{def:HT-FDO}
Let $\cH$ be a Hilbert space, and consider $s \in \R$ and $\widehat A \in H^s(\R_\omega, \cB(\cH))$. The Hilbert transform of $\widehat A$ is the element of $H^s(\R_\omega, \cB(\cH))$, denoted by $\fH(\widehat A)$, and defined by 
\begin{equation}\label{eq:defHTop}
\forall  (f,g) \in \cH \times \cH, \qquad  \left\langle f \left| \fH(\widehat A) \right|g \right\rangle = \fH  \left( \left\langle f \left| \widehat A \right|g \right\rangle \right).
\end{equation}
\end{definition}
 
In particular, it is possible to define the Hilbert transform of the Fourier transform of a uniformly bounded field of time-dependent operators, using the following result, which is a straightforward extension of Lemma~\ref{lem:image_Fourier_L_infty}. 

\begin{lemma} \label{lem:FourierA} Let $\cH$ be a Hilbert space, and let $s > 1/2$. Then for all $A \in L^\infty(\R_\tau, \cB(\cH))$, we have $\widehat A \in H^{-s}(\R_\omega, \cB(\cH))$, with
 \[
  \left\| \widehat{A} \right\|_{H^{-s}(\R_\omega, \cB(\cH))} = \left( 2 \pi \int_\R \left(1+\tau^{2}\right)^{-s} \left\| {A}(\tau) \right\|_{\cB(\cH)}^2 \, \rd\tau\right)^{1/2} \leq C_s\, \| A\|_{L^\infty(\R_\tau, \cB(\cH))},
  \]
  where $C_s$ is defined in~\eqref{eq:def_Cs}.
\end{lemma}
 
Let $\mathscr{B}(\R)$ be the set of Borel subsets of $\R$, $b \in \mathscr{B}(\R)$ a Borelian set, and $H$ a self-adjoint operator on a Hilbert space~$\cH$. We denote by $P_b^H:=\1_b(H)$ the spectral projection on $b$ of~$H$ (here, $\1_b$ is the characteristic function of the set $b$, and $\1_b(H) \in \cB(\cH)$ is defined by the spectral theorem for self-adjoint operators; see for instance~\cite[Theorem VII.2]{RS4}). 

\begin{definition}[Principal value of the resolvent of a self-adjoint operator] 
  \label{def:HT-PV}
  Let $H$ be a self-adjoint operator on a Hilbert space $\cH$. We define the $\cB(\cH)$-valued distribution $ \pv \left( \frac{1}{\cdot - H} \right)$ on the frequency domain $\R_\omega$ by
  \[
  \forall (f,g) \in \cH \times \cH, \qquad \left\langle f \left| \pv \left( \frac{1}{\cdot - H} \right) \right| g \right\rangle :=  \pi \, \fH(\mu^H_{f,g}),
\]
where $\mu^H_{f,g}$ is the finite complex Borel measure on $\R_\omega$ defined by
\[
\forall b \in \mathscr{B}(\R_\omega), \qquad \mu^H_{f,g}(b)= \langle f | P^H_b | g \rangle.
\]
\end{definition}

As any complex-valued bounded Borel measure on $\R_\omega$ is an element of $H^{-s}(\R_\omega)$ for any $s > 1/2$ (this is a consequence of the continuous embedding $H^s(\R_\omega) \hookrightarrow C^0(\R_\omega) \cap L^\infty(\R)$ for $s > 1/2$), it follows from Definitions~\ref{def:HT-FDO} and~\ref{def:HT-PV} that
\[
\pv \left( \frac{1}{\cdot - H} \right) = \pi \, \fH(P^H) \quad \mbox{in } H^{-s}(\R_\omega,\cB(\cH)), \quad s > 1/2,
\]
which is the operator analog of the well-known formula
\begin{equation}
\label{eq:Fh_Dirac}
\pv\left( \frac 1 \cdot \right) = \pi \, \fH(\delta_0)\quad \mbox{in } H^{-s}(\R_\omega), \quad s > 1/2,
\end{equation}
which is itself a simple reformulation of the equality 
\[
\mathcal{F}_T^{-1} \left[ \pv\left(\frac{1}{\cdot}\right) \right] = -\frac{\ri}2  \, \sgn(\cdot) \quad \mbox{in } L^\infty(\R_\tau).
\]
 
\subsection{Causal and anti-causal operators} 
\label{subsection:Fourier}

The GW formalism makes use of families of time-dependent operators $(T_{\rm c}(\tau))_{\tau \in \R}$ and $(T_{\rm a}(\tau))_{\tau \in \R}$ of the form
\[
T_{\rm c}(\tau) = \Theta(\tau) A_{\rm c}(\tau)  \quad \mbox{and} \quad T_{\rm a}(\tau) = \Theta(-\tau) A_{\rm a}(\tau),
\]
where $\Theta \, : \, \R \rightarrow \R$ is the Heaviside function~\eqref{eq:Heaviside}, and $A_{\rm c}$ and $A_{\rm a}$ belong to $L^\infty(\R,\cB(\cH))$ for a given Hilbert space~$\cH$. The family of operators $(T_{\rm c}(\tau))_{\tau \in \R}$ is called a \textit{causal operator}, as $T_{\rm c}(\tau) = 0$ for all $\tau < 0$. Likewise, the family of operators $(T_{\rm a}(\tau))_{\tau \in \R}$ is called an \textit{anti-causal operator}, as $T_{\rm a}(\tau) = 0$ for all $\tau > 0$. We recall in this section the basic properties of causal and anti-causal operators.

\subsubsection{Causal operators}
\label{sec:causal_operators}

Causal functions have very nice properties, because their Fourier transforms have analytic extensions in the upper half-plane
$$
\UU:= \left\{ z \in \C \, | \;  \Im z  > 0 \right\}.
$$
This comes from the fact that, if $f \in L^1(\R_\tau)+ L^\infty(\R_\tau)$ is such that $f(\tau) = 0$ for $\tau < 0$, the Laplace transform $\widetilde f$ of~$f$, defined on $\UU$ by\footnote{The Laplace transform is usually defined as
\[
	F(p) = \int_0^\infty f(\tau) \re^{-p\tau} \rd \tau.
\]
Our definition, which is better adapted to the GW framework, simply amounts to rotating the axis, or, in other words, to setting $z = \ri p$.
}  
\begin{equation} \label{eq:laplace}
\forall z \in \UU, \qquad \widetilde{f} (z) := \int_{\R} f(\tau) \re^{\ri z \tau} \,  \rd \tau = \int_0^{+\infty} f(\tau) \re^{\ri z \tau} \,  \rd \tau,
\end{equation}
is a natural analytic lifting onto $\UU$ of the time-Fourier transform $\widehat f$ of $f$ defined on $\R_\omega = \partial\UU$. Note that the Laplace transform can be extended to appropriate classes of tempered distributions, see~\cite[Chapter~VIII]{Schwartz}.

\medskip

Let us first recall the Titchmarsh's theorem~\cite{Titchmarsh_book} (see for instance~\cite[Section~1.6]{Nussenzveig1972}).

\begin{theorem}[Titchmarsh's theorem in $L^2$~\cite{Titchmarsh_book}] \label{th:Titchmarsh_Linfty}
  Let $f \in L^2(\R_\tau)$ and $\widehat{f} \in L^2(\R_\omega)$ be its time-Fourier transform. The following assertions are equivalent:
  \begin{itemize}
  \item[(i)] $f$ is causal  (\textit{i.e.}  $f(\tau) = 0$ for almost all $\tau < 0$);
  \item[(ii)] there exists an analytic function $F$ in the upper half-plane $\UU$ satisfying
    \[
    \sup_{\eta > 0} \left( \int_{-\infty}^{+\infty} \left| F(\omega + \ri \eta) \right|^2 \rd \omega \right) < \infty
    \]
    and such that, $F(\cdot + \ri \eta) \to \widehat f$ strongly in $L^2(\R_\omega)$, as $\eta \to 0^+$;
  \item[(iii)] $\Re \widehat{f}$ and $\Im \widehat{f}$ satisfy the first Plemelj formula
    \begin{equation} \label{eq:Plemelj1}
      \Re \widehat{f}  =  - \fH \left( \Im \widehat{f}  \right)  \quad \mathrm{in} \ L^2(\R_\omega);
    \end{equation}
  \item[(iv)] $\Re \widehat{f}$ and $\Im \widehat{f}$ satisfy the second Plemelj formula
    \begin{equation} \label{eq:Plemelj2}
      \Im \widehat{f} =  \fH \left( \Re \widehat{f}  \right)   \quad \mathrm{in} \ L^2(\R_\omega).
    \end{equation}
  \end{itemize}
  If these four assertions are satisfied, then the function $F$ in (ii) is unique, and coincides with the Laplace transform $\widetilde f$ of $f$.
\end{theorem}

We refer to \cite{Titchmarsh_book} for a proof of this theorem. Formulae (\ref{eq:Plemelj1})-(\ref{eq:Plemelj2}) are sometimes referred to as the Kramers-Kr\"onig formulae or the dispersion relations in the physics literature. Titchmarsh's theorem implies in particular that square integrable causal functions, which can be very easily characterized in the time domain (they vanish for negative times), can also be easily characterized in the frequency domain (the imaginary parts of their Fourier transforms are the Hilbert transforms of their real parts).

We emphasize that the above version of Titchmarsh's theorem is only valid in $L^2$, while the GW setting mostly involves $L^\infty$ causal functions (see Section \ref{subsec:GF} for instance). Weaker versions of Titchmarsh's theorem are available for wider classes of tempered distributions (see \cite{Nussenzveig1972} and references therein), but the $L^\infty$ setting turns out to be sufficient for our purposes and has the advantage of allowing short, self-contained proofs of all statements. Note that the assertions are no longer equivalent.

\begin{theorem} [Titchmarsh's theorem in $L^\infty(\R)$] 
  \label{th:Tit_Linfty}
  Let $g \in L^\infty(\R_\tau)$ be a causal function (\textit{i.e.}  $g(\tau) = 0$ for $\tau < 0$) and let $\widehat g \in H^{-s}(\R_\omega)$ for all $s > 1/2$ be its time-Fourier transform, and $\widetilde{g}$ be its Laplace transform defined on $\UU$. Then,
  \begin{enumerate}[(i)]
  \item $\widetilde{g}$ is analytic on~$\UU$;
  \item the function $\eta \mapsto \widetilde{g}(\cdot + \ri \eta)$ is continuous from $(0, +\infty)$ to $H^{s}(\R_\omega)$ for all $s \in \R$, and is uniformly continuous from $[0,+\infty)$ to $H^{-s}(\R_\omega)$ for all $s > 1/2$. Moreover, $\widetilde{g}(\cdot + \ri \eta) \to \widehat{g}$ strongly in $H^{-s}(\R_\omega)$ for all $s > 1/2$, as $\eta \to 0^+$; 
  \item for all $z \in \UU$,
      \begin{equation}\label{eq:gtildez}
    \widetilde{g}(z) = \frac{1}{2\ri\pi} \left\langle \widehat{g}, (\cdot-z)^{-1} \right\rangle_{H^{-1},H^1}. 
    \end{equation}
  \item $\Re \widehat{g}$ and $\Im \widehat{g}$ satisfy the Plemelj formulae:
    \begin{equation}\label{eq:dispersion_rel}
    \Re \widehat{g} = -  \fH \left( \Im \widehat{g}  \right)  
    \quad \mathrm{and} \quad  
    \Im \widehat{g} = \fH \left( \Re \widehat{g} \right) \quad \mathrm{in} \ H^{-1}(\R_\omega).
    \end{equation}
  \end{enumerate}
\end{theorem}

The proof of Theorem~\ref{th:Tit_Linfty}, which is a simplified version of the proof of the more general result given by~\cite[Lemma 1]{Taylor58} (see also~\cite[Section~1.7]{Nussenzveig1972}), is given in Section~\ref{ssec:proof_Tit_Linfty}. For simplicity, we stated~\eqref{eq:gtildez} and~\eqref{eq:dispersion_rel} in $H^{-s}$ for the value $s=1$, but similar results hold for any value $s >1/2$. 


Let us now extend these results to operator-valued functions. We recall that a map $\widetilde A(z)$ from an open set $U \subset \C$ to a Banach space $E$ is said to be strongly analytic on $U$ if $U \ni z \mapsto \widetilde{A}(z) \in E$ is $\C$-differentiable on $U$, \textit{i.e.} $\rd \widetilde{A}(z) / \rd z \in E$ for all $z \in U$.

\medskip

\begin{definition}[bounded causal operator] \label{def:particle_type_operator}
  Let $\cH$ be a Hilbert space and $T_{\rm c} \in L^\infty(\R_\tau, \cB(\cH))$. We say that $T_{\rm c}$ is a bounded causal operator on $\cH$ if $T_{\rm c}(\tau) = 0$ for almost all $\tau< 0$.
\end{definition}

Lemma~\ref{lem:image_Fourier_L_infty} and Theorem~\ref{th:Tit_Linfty} can be straightforwardly extended to operator-valued maps (see Section~\ref{sec:proof_lem:particle-type} for the proof).

\begin{proposition} \label{lem:particle-type}
  Let $\cH$ be a Hilbert space and $T_{\rm c} \in L^\infty(\R_\tau, \cB(\cH))$ a bounded causal operator on~$\cH$. Then its time-Fourier transform $\widehat{T_{\rm c}}$ belongs to $H^{-s}(\R_\omega,\cB(\cH))$ for any $s > 1/2$, and its Laplace transform
\[
\widetilde{T_{\rm c}}(z) := \int_\RR T_{\rm c}(\tau) \, \re^{\ri z \tau}\, \rd\tau = \int_0^{+\infty} T_{\rm c}(\tau) \, \re^{\ri z \tau}\, \rd\tau
\]
is well defined on the upper-half plane $\UU$. Moreover,
  \begin{enumerate}[(i)] 
  \item $\widetilde T_{\rm c}$ is a strongly analytic function from $\UU$ to $\cB(\cH)$; 
  \item the function $\eta \mapsto \widetilde T_{\rm c} (\cdot + \ri \eta)$ is continuous from $(0,+\infty)$ to~$H^{s}(\R_\omega,\cB(\cH))$ for all $s \in \mathbb{R}$, and uniformly continuous from $[0,+\infty)$ to~$H^{-s}(\R_\omega,\cB(\cH))$ for $s>1/2$. Moreover, for any $s > 1/2$, $\widetilde T_{\rm c} (\cdot + \ri \eta) \to \widehat T_{\rm c}$ strongly in $H^{-s}(\R_\omega, \cB(\cH))$ as $\eta \to 0^+$;
  \item for all $z \in \UU$, it holds
    \[
    \widetilde{T_{\rm c}}(z) = \dfrac{1}{2 \ri \pi} \left\langle \widehat{T_{\rm c}}, (\cdot - z)^{-1} \right\rangle_{H^{-1},H^1} ;
    \]
  \item the operators $\Re \widehat{T_{\rm c}}$ and $\Im \widehat{T_{\rm c}}$ satisfy the Plemelj formulae:
    \begin{equation} 
      \label{eq:Plemelj_formulae}
      \Re \widehat{T_{\rm c}}  =  - \fH \left( \Im \widehat{T_{\rm c}}  \right)  \quad \mathrm{and} \quad 
      \Im \widehat{T_{\rm c}}  = \fH \left( \Re \widehat{T_{\rm c}}  \right) \quad \mathrm{in} \ H^{-1}(\R_\omega,\cB(\cH)).
    \end{equation}
  \end{enumerate}
\end{proposition}

Besides the general case covered by Proposition~\ref{lem:particle-type}, the particular case of causal time-propagators is often encountered. Explicit formulae can be provided for the Laplace and Fourier transforms in this case, as made precise in the following result (see Section~\ref{sec:proof:prop:analytic_extension_causal_time_evolution} for the proof).

\begin{proposition}[Analytic extension of causal time propagators] \label{prop:resolvent}
\label{prop:analytic_extension_causal_time_evolution}
Let $H$ be a self-adjoint operator on a Hilbert space $\cH$ and $A_{\rm c}(\tau) := -\ri \Theta(\tau) \re^{-\ri \tau H}$. The Laplace transform $(\widetilde A_{\rm c}(z))_{z \in \UU}$ coincides with the resolvent of $H$ in $\UU$: 
\[
\widetilde{A}_{\rm c}(z) = (z-H)^{-1}.
\]
Moreover, $\widetilde A_{\rm c}(\cdot +\ri \eta)$ converge to $\widehat A_{\rm c}$ in $H^{-1}(\R_\omega, \cB(\cH))$ as $\eta \to 0^+$, and
\[
\Re \widehat A_{\rm c} = \pv \left( \frac{1}{\cdot - H} \right) \quad \mathrm{and} \quad \Im \widehat A_{\rm c} =-\pi P^H \quad \mathrm{in} \ H^{-1}(\R_\omega, \cB(\cH)).
\]
\end{proposition}

Let us conclude this section with a useful result (see Section~\ref{sec:proof:lemma:Im_positive_particle} for the proof).

\begin{lemma} \label{lemma:Im_positive_particle}
  Consider a bounded causal operator $T_{\rm c} \in L^\infty(\R_\tau, \cB(\cH))$ such that $\mathrm{Supp}(\Im \widehat{T_{\rm c}}) \subset [\omega_0,\infty)$ for some $\omega_0 \in \R$. Then $\Im \widehat{T_{\rm c}} \ge 0$ on $\R_\omega$ if and only if $\Re \widehat{T_{\rm c}} \ge 0$ on~$(-\infty,\omega_0]$.
\end{lemma}

\subsubsection{Anti-causal operators}

\begin{definition}[bounded anti-causal operator] \label{def:hole_type_operator}
  Let $\cH$ be a Hilbert space and $T_{\rm a} \in L^\infty(\R_\tau,\cB(\cH))$. We say that $T_{\rm a}$ is  a bounded anti-causal operator if $T_{\rm a}(\tau) = 0$ for almost all $\tau>0$.
\end{definition}
 
All the results for causal operators stated in the previous section can be straightforwardly transposed to  anti-causal operators, by remarking that if $(T_{\rm a}(\tau))_{\tau \in \R}$ is an anti-causal operator, then $(T_{\rm a}(-\tau))_{t\in \R}$ is a causal operator. We will use in particular the following results, which are the counterparts of Proposition~\ref{lem:particle-type}, Proposition~\ref{prop:resolvent} and Lemma~\ref{lemma:Im_positive_particle}.

\begin{proposition} 
  \label{lem:hole-type}
  Let $\cH$ be a Hilbert space and $T_{\rm a} \in L^\infty(\R_\tau,\cB(\cH))$ a bounded anti-causal operator on $\cH$. Then its time-Fourier transform $\widehat{T_{\rm a}}$ belongs to $H^{-s}(\R_\omega,\cB(\cH))$ for any $s > 1/2$, and its Laplace transform $\widetilde{T_{\rm a}}$ is well defined on the lower half-plane
  \[
  \LL=\left\{z \in \C \; | \; \Im(z) < 0 \right\}.
  \]
  Moreover,
  \begin{enumerate}[(i)] 
  \item $\widetilde T_{\rm a}$ is a strongly analytic function from $\LL$ to $\cB(\cH)$; 
  \item the function $\eta \mapsto \widetilde T_{\rm a} (\cdot - \ri \eta)$ is continuous from $(0,+\infty)$ to~$H^{s}(\R_\omega,\cB(\cH))$ for all $s \in \mathbb{R}$, and uniformly continuous from $[0,+\infty)$ to~$H^{-s}(\R_\omega,\cB(\cH))$ for $s>1/2$. Moreover, for any $s > 1/2$, $\widetilde T_{\rm a} (\cdot - \ri \eta) \to \widehat T_{\rm a}$ strongly in $H^{-s}(\R_\omega, \cB(\cH))$ as $\eta \to 0^+$;
  \item for all $z \in \LL$, it holds
    \[
    \widetilde{T_{\rm a}}(z) = -\dfrac{1}{2 \ri \pi} \left\langle \widehat{T_{\rm a}}, (\cdot - z)^{-1} \right\rangle_{H^{-1},H^1} ;
    \]
  \item the operators $\Re \widehat{T_{\rm a}}$ and $\Im \widehat{T_{\rm a}}$ satisfy the Plemelj formulae:
    \begin{equation}\label{eq:Plemelj_formulae_anti}
      \Re \widehat{T_{\rm a}}  =\fH \left(  \Im \widehat{T_{\rm a}}   \right) \quad \mathrm{and} \quad 
      \Im \widehat{T_{\rm a}}  = - \fH \left(   \Re \widehat{T_{\rm a}}   \right) 
      \quad \mathrm{in} \ H^{-1}(\R_\omega,\cB(\cH)).
    \end{equation} 
      \end{enumerate}
\end{proposition}
Note that the signs in the Plemelj formulae are different for causal and anti-causal operators (compare~\eqref{eq:Plemelj_formulae} and~\eqref{eq:Plemelj_formulae_anti}). Also, the Laplace transform is defined in the lower-half plane $\LL$ for anti-causal operators, while it is defined in the upper-half plane $\UU$ for causal operators. The counterpart of Proposition~\ref{prop:resolvent} is the following proposition.
\begin{proposition}[Analytic extension of anti-causal time propagators] \label{prop:resolvent_anti-causal}
Let $H$ be a self-adjoint operator on a Hilbert space $\cH$ and $A_{\rm a}(\tau) := \ri \Theta(-\tau) \re^{\ri \tau H}$. The Laplace transform $(\widetilde A_{\rm a}(z))_{z \in \LL}$ is
\[
\widetilde{A}_{\rm a}(z) = (z+H)^{-1}.
\]
Moreover, $\widetilde A_{\rm a}(\cdot - \ri \eta)$ converge to $\widehat A_{\rm a}$ in $H^{-1}(\R_\omega, \cB(\cH))$ as $\eta \to 0^+$, and
\[
\Re \widehat A_{\rm a} = \pv \left( \frac{1}{\cdot + H} \right) \quad \mathrm{and} \quad \Im \widehat A_{\rm a} =\pi P^{-H} \quad \mathrm{in} \ H^{-1}(\R_\omega, \cB(\cH)).
\]
\end{proposition}

Finally, a result similar to Lemma~\ref{lemma:Im_positive_particle} can also be stated.

\begin{lemma} 
  \label{lemma:Im_positive_hole}
  Consider a bounded anti-causal operator $T_{\rm a} \in L^\infty(\R_\tau,\cB(\cH))$ such that $\mathrm{Supp}(\Im \widehat{T_{\rm a}})\subset (-\infty,\omega_0]$ for some $\omega_0 \in \R_\omega$. Then, $\Im \widehat{T_{\rm a}} \ge 0$ if and only if $\Re \widehat{T_{\rm a}}(\omega) \ge 0$ on $[\omega_0,+\infty)$.
\end{lemma}

\subsection{Operators defined by kernel products}
\label{ssection:kernel_multiplication}

Two of the fundamental equations in the GW method (see Sections~\ref{sec:irreducible_polarizability} and~\ref{ssec:self_energy}) are of the form
\begin{equation} \label{eq:ex_C2}
  \cC(\bx_1,\bx_2) = \ri \cA(\bx_1, \bx_2) \cB(\bx_2, \bx_1),
\end{equation}
where $\cA(\bx,\bx')$ and $\cB(\bx,\bx')$ are the kernels of space-time operators invariant by time translations. As the product of the kernels of two operators is not, in general, the kernel of a well-defined operator, we have to clarify the meaning of~\eqref{eq:ex_C2}. We first treat the case of time-independent operators in Section~\ref{sec:kernel_products}, and consider time-dependent operators and their Laplace transforms in a second step (see Section~\ref{sec:Fourier_transform_kernel}). 

\subsubsection{Definition of the kernel product}
\label{sec:kernel_products}

We first consider the special case when the operators in~\eqref{eq:ex_C2} are time-independent. Our aim is to give a meaning to equalities such as 
\begin{equation} 
  \label{eq:ex_C3}
  C(\br_1,\br_2) := A(\br_1, \br_2) B(\br_2, \br_1),
\end{equation}
where $A(\br,\br')$ and $B(\br,\br')$ are the kernels of two integral operators~$A$ and~$B$ on $L^2(\R^3)$. For this purpose, we replace (\ref{eq:ex_C3}) by the formally equivalent definition
\begin{align}
\forall (f,g)  \in L^2(\R^3) \times L^2(\R^3), \qquad \bra f | C | g \ket & := \int_{\R^3} \int_{\R^3} \overline{f}(\br_1) C(\br_1, \br_2) g(\br_2) \, \rd \br_1 \, \rd \br_2 \nonumber \\
&= \int_{\R^3} \int_{\R^3} A(\br_1, \br_2) g(\br_2) B(\br_2, \br_1) \overline{f}(\br_1)\, \rd \br_1 \, \rd \br_2 \nonumber \\
& = \Tr_{L^2(\R^3)} \left( A g B \overline{f} \right), \label{eq:defAodotB}
\end{align}
where the last line involves the operators $A$ and $B$ themselves, and not their kernels ($\overline{f}$ and $g$ are there seen as multiplication operators by the functions $\overline{f}$ and $g$ respectively). 

The formal equalities leading to~\eqref{eq:defAodotB} suggest to define the kernel product of two operators $A$ and $B$ (defined on dense subspaces of $L^2(\R^3)$), as the operator on $L^2(\R^3)$ with domain $D \subset L^2(\R^3)$, denoted by $A \odot B$ and characterized by
\begin{equation}
  \label{eq:defAodotB2}
  \forall (f,g)  \in L^2(\R^3) \times D, \qquad \bra f | (A \odot B) | g \ket := \Tr_{L^2(\R^3)}\left( A g B \overline{f} \right).
\end{equation}
In particular, the product $A \odot B$ is a well-defined bounded operator on $L^2(\R^3)$ as soon as $A g B \overline{f}$ is trace-class for all $(f,g) \in  L^2(\R^3) \times L^2(\R^3)$ and $(f,g) \mapsto \Tr_{L^2(\R^3)}(A g B \overline{f})$ is a continuous sesquilinear form on $L^2(\R^3) \times L^2(\R^3)$. It follows from the above considerations that if $A$ and $B$ are operators with well-behaved (for instance smooth and compactly supported) kernels $A(\br_1, \br_2)$ and $B(\br_1, \br_2)$, then $A \odot B$ is a bounded integral operator with kernel $(A \odot B)(\br_1, \br_2)= A(\br_1, \br_2) B(\br_2, \br_1)$. 

\begin{remark}\label{rem:odot}
	It is also possible to rely on the formal equality
	\[
		\forall (f, g) \in L^2(\R^3) \times L^2(\R^3), \quad \bra f | C | g \ket = \Tr_{L^2(\R^3)} \left( \overline{f} A g B \right),
	\]
	and define another kernel product $\widetilde{\odot}$ by
	\[
		\forall (f,g)  \in L^2(\R^3) \times D, \qquad \left\bra f \left| A \, \widetilde{\odot} \, B \right| g \right\ket := \Tr_{L^2(\R^3)}\left( \overline{f} A g B \right).
	\]
	It may hold that $A \odot B$ is a well-defined bounded operator, while $A \, \widetilde{\odot} \, B$ is an unbounded operator.\footnote{
As an example of such a situation, take $\phi \in L^2(\R^3) \cap L^\infty(\R^3)$, $\psi \in L^2(\R^3) \setminus L^\infty(\R^3)$, and set $A = | \psi \ket \bra \phi |$ and $B = | \phi \ket \bra \phi |$. Then, for all $f, g \in L^2(\R^3)$, the operator $A g B \overline{f} = | \psi \ket \bra \phi | g | \phi \ket  \bra \phi f |$ is a well-defined rank-$1$ bounded operator since $\phi f \in L^2(\R^3)$, hence is trace class. Moreover,
\[
	\Tr_{L^2(\R^3)} \left( A g B \overline{f} \right) \le \left( \| \phi \|_{L^\infty}^2 \| \phi \|_{L^2} \| \psi \|_{L^2} \right) \| f \|_{L^2} \| g \|_{L^2},
\]
so that $A \odot B$ is a well-defined bounded operator on $L^2(\R^3)$. On the other hand, it formally holds $\overline{f} A g B = | f \psi \ket  \bra \phi | g | \phi \ket \bra \phi |$. If $f$ is such that $f \psi \notin L^2(\R^3)$, then this operator is not bounded. \\
We are grateful to Yanqi Qiu for pointing out this counter-example to our attention.
} In the sequel, we will mostly state the results for the $\odot$ kernel product.
\end{remark}

\begin{remark}
The product $A \odot B$ can be seen as an infinite-dimensional extension of the Hadamard product $\bA \circ \bB^T$ defined for two matrices $\bA \in \C^{m \times n}$ and $\bB \in \C^{n \times m}$ by
$$
\forall 1 \le i \le m, \quad \forall 1 \le j \le n, \quad \left(\bA \circ \bB^T\right)_{ij} = \bA_{ij} \left(\bB^T\right)_{ij}= \bA_{ij} \bB_{ji}.
$$
\end{remark}

Let us specify possible sufficient conditions for the operator $A \odot B$ to be well-defined. The typical situation we will encounter in the GW setting (see Sections~\ref{sec:irreducible_polarizability} and~\ref{ssec:self_energy}) is the case when $A \in \cB(L^2(\R^3))$, while $B$ is an operator on~$L^2(\R^3)$ satisfying
\begin{equation}
  \label{eq:condition_trace_A_Hadamard}
  \forall f, g \in L^2(\R^3), \quad  \Tr \left( \left| g B \overline{f} \right| \right)  \le C_B \| f \|_{L^2} \| g \|_{L^2}.
\end{equation}
In this case, the operator $A \odot B$ defined in~\eqref{eq:defAodotB2} is a well-defined bounded linear operator on $L^2(\R^3)$, and
\[
\| A \odot B \|_{\cB(L^2(\R^3))} \le C_B \| A \|_{\cB(L^2(\R^3))}.
\]
The operators~$B$ arising in the GW formalism are usually of the form $B = B_1^* B_2 B_1$, where $B_1$ is an operator from $L^2(\R^3)$ to some Hilbert space~$\cH$, and $B_2 \in \cB(\cH)$. In fact, assume that the operator $B_1$ is such that $B_1f \in \fS_2(L^2(\R^3),\cH)$ for any $f \in L^2(\R^3)$, with 
\begin{equation}
  \label{eq:condition_A1_S2}
  \| B_1 f \|_{\fS_2(L^2(\R^3),\cH)} \leq K \| f\|_{L^2},
\end{equation}
for a constant $K \in \R^+$ independent of $f$. In the left-hand side of~\eqref{eq:condition_A1_S2}, $f$ denotes the multiplication operator by the function $f$. In this case, \eqref{eq:condition_trace_A_Hadamard} holds with
\[
C_B = K^2 \| B_2 \|_{\cB(\cH)}.
\]

Let us conclude by giving a simple example when~\eqref{eq:condition_A1_S2} is satisfied, in the situation when $\cH = L^2(\R^3)$. 

\begin{lemma} 
  \label{lem:L2Linfty_kernel}
  Let $B_1$ be a linear operator with integral kernel $B_1(\br,\br') \in L_{\rm loc}^2(\R^3 \times \R^3)$, such that $\br \mapsto \| B_1(\br,\cdot) \|_{L^\infty} \in L^2(\R^3)$. Then $B_1 \in \cB(L^1(\R^3),L^2(\R^3))$, so that $B_1$ defines an operator on $L^2(\R^3)$ with domain $L^1(\R^3) \cap L^2(\R^3)$. Moreover, for any $f \in L^2(\R^3)$, the operator $B_1f$ is Hilbert-Schmidt on $L^2(\R^3)$, with  
  \[
  \| B_1f \|_{\fS_2(L^2(\R^3))} \le \left( \int_{\RR^3} \| B_1(\br, \cdot) \|^2_{L^\infty(\R^3)} \, \rd\br \right)^{1/2} \| f \|_{L^2(\R^3)}.
  \]
\end{lemma}

The proof of this result can be read in Section~\ref{ssec:proof_L2Linfty_kernel}. In the GW setting, a technical result similar to Lemma~\ref{lem:L2Linfty_kernel} is provided by Lemma~\ref{lem:Avcf}.

\subsubsection{Properties of the kernel product}

\begin{lemma}
\label{lem:positivity_kernel_product}
Consider two bounded operators $A,B \in \cB(L^2(\R^3))$ such that $A,B \geq 0$ and~\eqref{eq:condition_trace_A_Hadamard} holds. Then, $A\odot B$ is a bounded, positive operator on~$L^2(\R^3)$. 
\end{lemma}

The proof of this result is very simple: it relies on the observation that, for any $f \in L^2(\R^3)$,
\[
\left\bra f \left| A\odot B\right| f\right\ket = \Tr_{L^2(\R^3)}\left( A f B \overline{f} \right) = \Tr_{L^2(\R^3)}\left( A^{1/2} f B \overline{f} A^{1/2}\right) \geq 0,
\]
since $f B \overline{f}$ is a positive, trace class operator and $A^{1/2} \geq 0$ is a bounded operator.

\begin{lemma}
\label{lem:adjoint_kernel_product}
Consider two bounded operators $A,B \in \cB(L^2(\R^3))$ such that \eqref{eq:condition_trace_A_Hadamard} holds. Then, $A\odot B$ is a bounded operator with adjoint $(A\odot B)^* = A^* \odot B^*$.
\end{lemma}

The proof of this result is also elementary: for any $f,g \in L^2(\R^3)$,
\[
\begin{aligned}
\left\bra f \big| \left( A\odot B \right) g\right\ket & = \Tr_{L^2(\R^3)}\left( A g B \overline{f} \right) = \overline{\Tr_{L^2(\R^3)}\left( \left(A g B \overline{f}\right)^* \right)} = \overline{\Tr_{L^2(\R^3)}\left( f B^* \overline{g} A^*  \right)} \\
& = \overline{\Tr_{L^2(\R^3)}\left( A^* f B^* \overline{g} \right)} = \overline{ \left\bra g, (A^* \odot B ^*) f\right\ket} = \left\bra (A^* \odot B ^*) f, g\right\ket.
\end{aligned}
\]
In particular, $A \odot B$ is self-adjoint whenever $A$ and $B$ are self-adjoint.

\subsubsection{Laplace transforms of kernel products}
\label{sec:Fourier_transform_kernel}

We finally combine the results on causal operators with those on the kernel product $\odot$ defined in Section~\ref{sec:kernel_products} in order to give a meaning to~\eqref{eq:ex_C2}. Note first that the space-time operator with kernel $\cC(\bx,\bx')$ is also time-translation invariant and that the family of operators $(A(\tau))_{\tau \in \R}$, $(B(\tau))_{\tau \in \R}$ and $(C(\tau))_{\tau \in \R}$ such that, formally, $\cA(\bx_1, \bx_2) =A(\br_1,\br_2,t_1-t_2)$, $\cB(\bx_1, \bx_2) =B(\br_1,\br_2,t_1-t_2)$, and $\cC(\bx_1, \bx_2) = C(\br_1,\br_2,t_1-t_2)$, are related by
\begin{equation}\label{eq:defCtau}
C(\tau) = \ri A(\tau) \odot B(-\tau).
\end{equation}
We assume here that $A$ and $B$ are such that~\eqref{eq:defCtau} is well-defined. When all the operator-valued functions have sufficient regularity in time, their Fourier transforms decay sufficiently fast at infinity and it is possible to Fourier transform~\eqref{eq:defCtau}. This is however not the typical case in the GW setting since we work with causal and anti-causal operators, whose Fourier transforms are in $H^{-s}(\R_\omega)$ for some $s>1/2$. 

We therefore rather consider Laplace transforms. More precisely, for two fields of uniformly bounded operators $(A(\tau))_{\tau \in \R}$ and $(B(\tau))_{\tau \in \R}$, and provided $C(\tau) := \ri A(\tau) \odot B(-\tau)$ is well defined, we can decompose $A$, $B$ and $C$ as the sums of their causal and anti-causal parts as
\[
A(\tau) = A^+(\tau) + A^-(\tau) \quad \text{with} \quad A^+(\tau) := \Theta(\tau) A(\tau) \quad \text{and} \quad A^-(\tau) := \Theta(-\tau) A(\tau),
\]
and similarly for $B$ and $C$. Then,
\begin{equation} \label{eq:Fourier_kernel}
C^+(\tau) = \ri A^+(\tau) \odot B^-(-\tau) \quad \text{and} \quad C^-(\tau) = \ri A^-(\tau) \odot B^+(-\tau).
\end{equation}
We next consider $\omega > 0$ and $0 < \eta < \omega$. From the equality
\[
C^+(\tau) \, \re^{-\omega \tau} = \ri \left[ A^+(\tau) \, \re^{-(\omega-\eta)\tau} \right] \odot \left[ B^-(-\tau)\, \re^{-\eta\tau} \right],
\]
we deduce, by Fourier transform, that
\begin{equation}
  \label{eq:laplace_transform_C+}
  \widetilde{C^+}(\nu + \ri\omega) = \frac{\ri}{2\pi}\int_{-\infty}^{+\infty} \widetilde{A^+}\big(\nu-\omega'+\ri(\omega-\eta)\big)\odot \widetilde{B^-}(-\omega'-\ri\eta) \, \rd \omega'.
\end{equation}
The convolution on the right-hand side is well defined in view of Propositions~\ref{lem:particle-type} and~\ref{lem:hole-type}. It however becomes ill-defined as $\omega,\eta \to 0$. In the case when the causal and anti-causal operators $A^+$ and $B^-$ under consideration are time-propagators, it is possible to remove this singularity by rewriting the convolution on appropriately shifted imaginary axes.

\begin{theorem}
  \label{thm:analytic_continuation_kernel_product}
  Consider three Hilbert spaces $\cH,\cH_a,\cH_b$, and assume that 
  \[
  \begin{aligned}
  A^+(\tau) = -\ri \Theta(\tau) A_1^* \re^{-\ri \tau A_2} A_1, & 
  \qquad
  A^-(\tau) = \ri \Theta(-\tau) A_1^* \re^{\ri \tau A_2} A_1, \\
  B^+(\tau) = -\ri \Theta(\tau) B_1^* \re^{-\ri \tau B_2} B_1, & 
  \qquad 
  B^-(\tau) = \ri \Theta(-\tau) B_1^* \re^{\ri \tau B_2} B_1,
  \end{aligned}
  \]
  where $A_1 \in \mathcal{B}(\cH,\cH_a)$, $B_1 \in \mathcal{B}(\cH,\cH_b)$ and $A_2,B_2$ are possibly unbounded, self-adjoint operators on~$\cH_a$ and~$\cH_b$ respectively, for which there exist real numbers $a,b$ such that $A_2 \geq a$ and $B_2 \geq b$. 
  We assume in addition that, for any $f \in \cH$, $B_1 f \in \fS_2(\cH,\cH_b)$ with $\| B_1 f \|_{\fS_2(\cH,\cH_b)} \leq K \| f\|_\cH$, for a constant $K \in \R^+$ independent of $f$. Then, the operators $C$, $C^+$ and $C^-$ in~\eqref{eq:defCtau}-\eqref{eq:Fourier_kernel} are well-defined, the Laplace transforms of $C^+$ and $C^-$ admit analytical continuations on $\UU \cup \LL \cup (-\infty, a+b)$ and $\UU \cup \LL \cup (-(a+b), \infty)$ respectively, and it holds
  for any $\nu < a+b$ and $\nu' \in (-b,a-\nu)$,
  \begin{equation}
    \label{eq:general_formula_analytic_continuation_C+}
    \forall \omega \in \R, 
    \qquad 
    \widetilde{C^+}(\nu + \ri\omega) = -\frac{1}{2\pi} \int_{-\infty}^{+\infty}  \widetilde{A^+}\big(\nu+\nu'+\ri(\omega+\omega')\big)\odot \widetilde{B^-}(\nu'+\ri\omega') \, \rd\omega',
  \end{equation}
  while, for any $\nu > -(a+b)$ and $\nu' \in (-a-\nu,b)$,
  \begin{equation}
    \label{eq:general_formula_analytic_continuation_C-}
    \forall \omega \in \R, 
    \qquad 
    \widetilde{C^-}(\nu + \ri\omega) = -\frac{1}{2\pi} \int_{-\infty}^{+\infty}  \widetilde{A^-}\big(\nu+\nu'+\ri(\omega+\omega')\big)\odot \widetilde{B^+}(\nu'+\ri\omega') \, \rd\omega'.
  \end{equation}
  Finally, the following equality holds provided $b>0$ and $a+b>0$: for any $\nu \in (-(a+b),a+b)$ and $\nu' \in (-b,b)$,
  \begin{equation}
    \label{eq:general_formula_analytic_continuation_C}
    \forall \omega \in \R, 
    \qquad 
    \widetilde{C}(\nu + \ri\omega) = -\frac{1}{2\pi} \int_{-\infty}^{+\infty}  \widetilde{A}\big(\nu+\nu'+\ri(\omega+\omega')\big)\odot \widetilde{B}(\nu'+\ri\omega') \, \rd\omega'.
  \end{equation}
\end{theorem}

The proof of Theorem~\ref{thm:analytic_continuation_kernel_product} can be read in Section~\ref{sec:thm:analytic_continuation_kernel_product}. The choices of $\nu,\nu'$ ensure that the function $\omega' \mapsto \widetilde{A^+}\big(\nu+\nu'+\ri(\omega+\omega')\big)$ is in $L^p(\R_\omega,\mathcal{B}(\cH))$ for any $p > 1$, while, for any $f,g\in\cH$, the function $\omega' \mapsto g\widetilde{B^-}(\nu'+\ri\omega')\overline{f}$ is in $L^p(\R_\omega,\fS_1(\cH))$ for any $p > 1$. Therefore, in view of~\eqref{eq:general_formula_analytic_continuation_C+}, the function $\omega \mapsto \widetilde{C^+}(\nu + \ri\omega)$ is in $L^p(\R_\omega,\mathcal{B}(\cH))$ for any $p > 1$. Similar results hold for $\omega \mapsto \widetilde{C^-}(\nu + \ri\omega)$ and $\omega \mapsto \widetilde{C}(\nu + \ri\omega)$.


\medskip

Let us conclude this section by deducing interesting properties from the analytic continuation results given by Theorem~\ref{thm:analytic_continuation_kernel_product} (see Section~\ref{sec:corr:analytic_C} for the proof).

\begin{corollary}
  \label{corr:analytic_C}
  Assume that the conditions of Theorem~\ref{thm:analytic_continuation_kernel_product} hold. Then, 
  \begin{equation}
    \label{eq:imaginary_widehat_C}
  \begin{aligned}
    \Supp \left(\Im \widehat{C^+} \right) \subset \left[a+b,+\infty\right), \qquad & \Im \widehat{C^+} \geq 0, \\
    \Supp \left(\Im \widehat{C^-} \right) \subset \left(-\infty, -(a+b)\right], \qquad & \Im \widehat{C^-} \geq 0,
  \end{aligned}
  \end{equation}
  so that
  \[
  \Supp \left(\Im \widehat{C} \right) \subset \R \backslash \left( -(a+b),a+b \right), \qquad \Im \widehat{C} \geq 0.
  \]
  Moreover,
  \begin{equation}
   \label{eq:Sigma_limit_real}
  \begin{aligned}
  \widehat{C^+} & = \Re \widehat{C^+} \ge 0 \quad \text{on} \quad \big(-\infty,a+b\big), \\
  \widehat{C^-} & = \Re \widehat{C^-} \ge 0 \quad \text{on} \quad \big(-(a+b),+\infty\big).
  \end{aligned}
  \end{equation}
  In particular, $\widehat{C}  = \Re \widehat{C} \geq 0$ on $\big( -(a+b),a+b \big)$.
\end{corollary}

\subsection{Second quantization formalism}
\label{sec:2ndQ}

We recall here the definitions of the main mathematical objects used in the second quantization formalism, which are used to define -- at least formally -- the kernels of the operators arising in the GW method. More details about the second quantization formalism can be found \textit{e.g.} in~\cite{Derezinski1997}.

\medskip

We consider a system of $N$ electrons in Coulomb interaction subjected to a time-independent real-valued external potential $v_{\ext} \in L^2(\R^3, \R) + L^\infty(\R^3, \R)$. In order to study the response of the system when electrons are added or removed, we embed this $N$-body problem in a more general framework where the number of electrons is not prescribed. We denote by $\cH_1=L^2(\R^3,\C)$ the one-electron state space (the spin variable is omitted for simplicity), by $\cH_N=\bigwedge^N \cH_1$ the $N$-electron state space, and by $\dps \FF=\oplus_{N=0}^{+\infty} \cH_N$ the Fock space, with the convention that $\cH_0=\C$. The Hamiltonian of the $N$-particle system reads
\begin{equation} \label{eq:H_N}
H_N=-\frac 12 \sum_{i=1}^N \Delta_{\br_i} + \sum_{i=1}^N v_{\rm ext}(\br_i)+\sum_{1 \le i < j \le N} \frac{1}{|\br_i-\br_j|},
\end{equation}
and the corresponding Hamiltonian acting on the Fock space is denoted by $\HH$, so that $H_N=\HH|_{\cH_N}$.\\
 
For $f \in \cH_1$, the creation and annihilation operators $a^\dagger(f)$ and $a(f)$ are the bounded operators on the Fock space $\FF$ defined by
$$
\forall N \in \N, \quad a^\dagger(f)|_{\cH_N} \in \cB(\cH_N,\cH_{N+1}), \qquad a(f)|_{\cH_{N+1}} \in \cB(\cH_{N+1},\cH_{N}), 
$$
and for all $\Phi_N \in \cH_N$,
\begin{equation}
\label{eq:annihilation_op}
\begin{aligned}[]
  [a^\dagger(f)\Phi_N](\br_1, \ldots \br_{N+1}) & 
  := \frac{1}{\sqrt{N+1}} \sum_{j=1}^{N+1} (-1)^{j+1} f(\br_{j}) \Phi_N(\br_{1}, \ldots, \br_{j-1},\br_{j+1},\ldots,\br_{N+1}), \\
  [a(f)\Phi_{N}](\br_1, \ldots \br_{N-1}) & 
  := \sqrt{N} \int_{\R^3} \overline{f}(\br) \Phi_{N}( \br, \br_1, \ldots, \br_{N-1}) \, \rd \br.
\end{aligned}
\end{equation}
The creation and annihilation operators satisfy $a^\dagger(f) = a(f)^*$ and the anticommutation relations
\begin{equation} \label{eq:anticommutation_a}
\forall (f,g) \in \cH_1 \times \cH_1, \quad [a(f),a(g)]_+=0, \quad [a^\dagger(f),a^\dagger(g)]_+=0, \quad [a(f),a^\dagger(g)]_+=\langle f|g \rangle \mathds{1}_{\FF},
\end{equation}
where $[A,B]_+=AB+BA$ is the anti-commutator of the operators $A$ and $B$, and where $\mathds{1}_{\FF}$ is the identity operator on $\FF$. In particular,
$$
a^\dagger(f)a(f)+a(f)a^\dagger(f) = \| f\|^2_{\cH_1} \, \mathds{1}_{\FF} .
$$
The mappings $\cH_1 \ni f \mapsto a^\dagger(f) \in \cB(\FF)$ and $\cH_1 \ni f \mapsto a(f) \in \cB(\FF)$ are respectively linear and antilinear. 

In most physics articles and textbooks, the GW formalism is presented in terms of the quantum field operators in the position representation $\Psi(\br)$ and $\Psi^\dagger(\br)$. We recall that, formally,
\[
\forall \br \in \R^3, \ \Psi^\dagger(\br) = \sum_{i=1}^\infty \overline{\phi_i(\br)}  a^\dagger(\phi_i) , 
\qquad 
\Psi(\br) = \sum_{i=1}^\infty \phi_i(\br)  a(\phi_i),
\]
where $\left\{ \phi_i \right\}_{i \in \N}$ is any orthonormal basis of $\cH_1$. Note that for any $f \in \cH_1$,
\[
	\int_{\R^3} \Psi^\dagger(\br) f(\br) \rd \br = a^\dagger(f)
	\quad \text{and} \quad
	\int_{\R^3} \Psi(\br) f(\br) \rd \br = a(\overline{f}).
\]
In the second-quantization formalism, $\HH$ reads,
\begin{equation*} 
	\HH  = \int_{\R^3} \Psi^\dagger( \br)  \left( -\frac 12 \Delta_\br + v_{\rm ext}(\br) \right) \Psi ( \br) \, \rd \br + \dfrac12 \int_{(\R^3)^2}  \Psi^\dagger( \br) \Psi^\dagger( \br')  |\br-\br'|^{-1}  \Psi(\br') \Psi ( \br) \, \rd \br \, \rd \br' .
\end{equation*} 

Finally, we introduce the Heisenberg representation of the annihilation and creation field operators $\Psi_{\rm H}(\br t)$ and $\Psi^{\dagger}_{\rm H}(\br t)$, formally defined by 
\[
\Psi^\dagger_{\rm H}(\br t) = \re^{\ri t \HH} \Psi^\dagger(\br) \re^{- \ri t \HH} \quad \text{and} \quad \Psi_{\rm H}(\br t) = \re^{\ri t \HH} \Psi(\br) \re^{-\ri t \HH}.
\]
Note that, still formally, $\Psi_{\rm H}(\br t)^\ast =\Psi^\dagger_{\rm H}(\br t)$, and
\begin{equation} 
  \label{eq:Psi_dagger}
  \Psi_{\rm H}^\dagger(\br t) \big|_{\cH_N} = \re^{\ri t H_{N+1}} \Psi^\dagger(\br) \, \re^{- \ri t H_N},
  \qquad 
  \Psi_{\rm H}(\br t) \big|_{\cH_{N+1}} = \re^{\ri t H_{N}} \Psi(\br) \, \re^{-\ri t H_{N+1}} .
\end{equation}

\section{Operators arising in the GW method for finite systems}
\label{sec:operators}

This section aims at providing rigorous mathematical definitions of the operators arising in the GW method. For each one of them, we first recall the formal definition given in the physics literature, using the second quantization formalism. We then explain  how to recast this formal definition into a (formally equivalent) satisfactory mathematical definition involving only well-defined operators on the $k$-particle spaces $\cH_k$, with $k=1,N-1,N,N+1$, the Coulomb space $\cC$ (defined in Section~\ref{sec:def_rho_H}), and its dual $\cC'$. We finally establish some mathematical properties of the operator under consideration, using our definition as a starting point. Unless otherwise specified, scalar products and norms are by default considered on~$\cH_1 = L^2(\R^3, \C)$.

We first need to make some assumptions on the physical system under consideration (see Section~\ref{sec:assumptions_reference}). We can then define the one-body Green's functions in Section~\ref{subsec:GF}. Linear response operators are considered in Section~\ref{sec:linear_response}, which culminates with the definition of the dynamically screened interaction operator~$W$. We finally introduce the self-energy operator in Section~\ref{subsec:SE}.

\subsection{Assumptions on the reference $N$-electron system}
\label{sec:assumptions_reference}

Recall that the reference system with $N$ electrons is described by the Hamiltonian $H_N$ on $\cH_N$ defined by (\ref{eq:H_N}). Our first assumption concerns the ground state energy $E_N^0$ of the reference system described by $H_N$:
\begin{equation*}
  \boxed{\mbox{\textbf{Hyp. 1:}  The ground state energy $E^0_N$ is a simple discrete eigenvalue of~$H_N$.}}
\end{equation*}
In this case, the normalized ground state wave-function $\Psi_N^0$ of the reference system is unique up to a global phase. We also define the energy of the first excited state:
\[
E_N^1 = \min \Big( \sigma(H_N) \backslash \{ E_N^0 \} \Big).
\]
Together with $\Psi_N^0$, we introduce the ground state one-body reduced density-matrix
\begin{equation}\label{eq:def-DM}
\gamma_N^0(\br,\br') := N   \int_{(\R^3)^{N-1}} \Psi_N^0(\br  , \br_2,\cdots, \br_N) \overline{\Psi_N^0(\br', \br_2,\cdots, \br_N)} \, \rd \br_2 \cdots \rd \br_N,
\end{equation}
the ground state density
$$
\rho_N^0(\br) := \gamma_N^0(\br,\br) = N \int_{(\R^3)^{N-1}} |\Psi_N^0(\br, \br_2,\cdots, \br_N)|^2 \, \rd \br_2 \cdots \rd \br_N,
$$
and the ground state two-body density
\begin{equation} \label{def:2body}
\rho_{N,2}^0(\br,\br') := \frac{N (N-1)}{2}  \int_{(\R^3)^{N-2}} |\Psi_N^0(\br , \br', \br_3,\cdots, \br_N)|^2 \, \rd \br_3 \cdots \rd \br_N
\end{equation}
of the reference $N$-electron system.

We recall in the following proposition some important properties on $\Psi_N^0$, $\gamma_N^0$, $\rho_N^0$ and $\rho_{N,2}^0$ (most of the assertions below are well known; we provide elements of proof in Section~\ref{sec:proof:prop:Psi_N^0} for the less standard statements). Note that both $\gamma_N^0(\br,\br')$ and $\rho_{N,2}^0(\br,\br')$ can be seen as the kernels of bounded operators on $\cH_1 = L^2(\R^3)$ that we also denote by $\gamma_N^0$ and $\rho_{N,2}^0$.  

\begin{proposition}[Properties of the ground state] \label{prop:Psi_N^0} 
  Assume that $v_{\rm ext}$ is of the form
  \[
  	v_\ext (\br)= - \sum_{k=1}^M \dfrac{z_k}{| \br - \bR_k|},
  \]
  with $z_k \in \N^*$ and $\bR_k \in \R^3$ for all $1 \le k \le M$, and that {\bf Hyp. 1} is satisfied. Then,
  \begin{enumerate}[(1)]
  \item the ground state wave-function $\Psi_N^0$ can be chosen real-valued and $\Psi_N^0 \in H^2(\R^{3N})$;
  \item the ground state density $\rho_N^0$ is in $L^1(\R^3, \R) \cap L^\infty(\R^3, \R)$ and $\nabla \sqrt{\rho_N^0} \in \left( L^2(\R^3, \R) \right)^3$. Moreover, $\rho_N^0$ is continuous and everywhere positive on $\R^3$;
  \item the ground state one-body reduced density operator $\gamma_N^0$ is in
    \[
    \cK_N:=\left\{ \gamma_N \in \cS(\cH_1) \; \Big| \; 0 \le \gamma _N\le 1, \; \Tr_{\cH_1}(\gamma_N)=N, \; \Tr_{\cH_1}(|\nabla|\gamma_N|\nabla|) < \infty \right\},
    \]
    and satisfies 
    \begin{equation}
      \label{eq:recover_gamma_with_creation_annihilation}
      \forall (f,g) \in \cH_1 \times \cH_1, \qquad \langle f | \gamma_N^0| g \rangle = \langle \Psi_N^0 | a^\dagger(g) a(f) |\Psi_N^0\rangle_{\cH_N};
    \end{equation}
  \item the kernel $\gamma_N^0(\br,\br')$ satisfies the pointwise estimate $|\gamma_N^0(\br, \br')|^2 \leq \rho_N^0(\br) \rho_N^0(\br')$;
  \item the operator $\rho_{N,2}^0$ belongs to~$\cS(\cH_1)$, and $\| \rho_{N,2}^0 \|_{\cB(\cH_1)} \leq \dfrac{N-1}{2} \|\rho_N^0\|_{L^\infty}$.
  \end{enumerate}
\end{proposition}

Much finer regularity results on $\Psi_N^0$ are available~\cite{Fournais2002, Fournais2005, Yserentant2010}, but are not needed for our purpose. Similar results hold true if $v_\ext$ is replaced by a potential generated by smeared nuclei or pseudo-potentials.
\medskip

Our second assumption is concerned with the (discrete) convexity of $N \mapsto E_N^0$. We assume that $N \ge 1$, and that (with the convention $E_0^0 = 0$ in the case $N=1$)
\begin{equation*}
  \boxed{\mbox{\textbf{Hyp. 2:}  $ E_N^0 - E_{N-1}^0 < E_{N+1}^0 - E_N^0$}. }
\end{equation*}
In this case, any real number $\mu$ such that $E_N^0 - E_{N-1}^0 < \mu < E_{N+1}^0 - E_N^0$ is an admissible chemical potential (Fermi level) of the electrons for the ground state of the reference system. The physical relevance of this assumption is discussed for instance in~\cite[Section~4.2]{Farid99}.
 
\subsection{Green's functions} 
\label{subsec:GF} 

We begin our journey in the GW formalism with Green's functions. The GW method has been designed from the equation of motion for the time-ordered one-body Green's function $G$ \cite{Hedin1965}, which is the concatenation of two meaningful physical objects: the particle Green's function $G_{\rm p}$ and the hole Green's function $G_{\rm h}$.

\subsubsection{The particle Green's function $G_{\rm p}$}

\paragraph{Rigorous definition of the particle Green's function.}
The particle (or forward, or retarded) Green's function is formally defined by (see for instance~\cite[Section~7]{FetterWalecka})
\begin{equation} \label{eq:forwardGF_physics}
  \cG_{\rm p}(\br t,\br' t') := -\ri \Theta(t-t') \ \bra \Psi_N^0 | \Psi_{\rm H}(\br t) \Psi^\dagger_{\rm H}(\br't') |\Psi_N^0 \ket,
\end{equation}
where $\Theta$ is the Heaviside function~\eqref{eq:Heaviside}, and $\Psi_{\rm H}(\br t)$ and $\Psi^{\dagger}_{\rm H}(\br t)$ are the Heisenberg representations of the annihilation and creation field operators introduced in Section~\ref{sec:2ndQ}. As $\Psi_N^0 \in \cH_N$, we can replace $\Psi(\br t)$ and $\Psi^\dagger(\br t)$ by their expressions (\ref{eq:Psi_dagger}):
\begin{align*}
  \cG_{\rm p}(\br t,\br' t') & = -\ri \Theta(t-t') \ \bra \Psi_N^0 | \re^{\ri t H_{N}} \Psi(\br) \re^{- \ri (t-t') H_{N+1}} \Psi^\dagger(\br') \re^{- \ri t' H_{N}} | \Psi_{N}^0 \ket \\
  & =  -\ri \Theta(t-t') \ \bra \Psi_N^0 | \Psi(\br) \re^{- \ri (t-t') (H_{N+1} - E_N^0)} \Psi^\dagger(\br')  | \Psi_{N}^0 \ket.
\end{align*}
As $\cG_{\rm p}$ only depends on the time difference $t - t'$, it is sufficient to study the function $G_{\rm p} (\br, \br', \tau) := \cG_{\rm p} (\br \tau, \br' 0)$. We then notice that, for all $f \in \cH_1$,
\[
\int_{\R^3}   \Psi^\dagger(\br')   | \Psi_{N}^0\rangle  f(\br') \, \rd \br' =   a^\dagger(f) |\Psi_N^0\rangle.
\]
Introducing
\[
\begin{array}{llll}
  A_+^\ast : & \cH_1 & \to & \cH_{N+1} \\
  & f & \mapsto & a^\dagger(f) |\Psi_N^0\rangle
\end{array}
\]
and $A_+ = (A_+^\ast)^*$, we observe that $G_{\rm p}(\br,\br',\tau)$ is formally the kernel of the following one-body operator.

\begin{definition}[Particle Green's function]
  The particle Green's function is defined as
  \begin{equation}
    \label{eq:defGp}
    G_{\rm p}(\tau) := - \ri \Theta(\tau) \ A_+ \re^{- \ri \tau (H_{N+1} - E_N^0)} A_+^\ast.
\end{equation}
\end{definition}

\paragraph{First properties of the particle Green's function.}
The study of $G_{\rm p}$ can be decomposed into the study of the operators $A_+$ and $\re^{- \ri \tau (H_{N+1} - E_N^0)}$. The latter is clearly bounded on~$\cH_{N+1}$. As for the operator $A_+^\ast$, we deduce from (\ref{eq:anticommutation_a}) and~\eqref{eq:recover_gamma_with_creation_annihilation} that
\begin{equation*}
  \bra a^\dagger(f) \Psi_N^0 | a^\dagger(g) \Psi_N^0 \ket = \bra  \Psi_N^0 | a(f) a^\dagger(g) | \Psi_N^0 \ket = \bra f | g\ket  - \bra  \Psi_N^0 | a^\dagger(g) a(f) | \Psi_N^0 \ket =  \bra f | 1 - \gamma_N^0 | g \ket,
\end{equation*}
or equivalently,
\begin{equation} \label{eq:A+bounded}
  A_+ A_+^\ast = \mathds{1}_{\cH_1} - \gamma_N^0. 
\end{equation} 
Hence, $A_+^\ast$ is a bounded operator from $\cH_1$ to $\cH_{N+1}$, and $A_+$ is a bounded operator from $\cH_{N+1}$ to $\cH_1$. In fact, since
\[
\| A_+^\ast f\|^2_{\cH_{N+1}} = \left\langle f| (\mathds{1}_{\cH_1}-\gamma_N^0)|f \right\rangle = \left\| (\mathds{1}_{\cH_1}-\gamma_N^0)f \right\|_{\cH_1}^2,
\]
it holds $\| A_+^\ast \|_{\cB(\cH_1,\cH_{N+1})} = 1$. The following properties are obtained as a direct corollary of Proposition~\ref{prop:analytic_extension_causal_time_evolution}.

\begin{proposition}[Properties of the particle Green's function] 
  \label{lemma:forward_GF}
  The family $(G_{\rm p}(\tau))_{\tau\in \R}$ defines a bounded causal operator on $\cH_1$. The real and imaginary parts of its time-Fourier transform are in $H^{-s}(\R_\omega, \cB(\cH_1))$ for all $s > 1/2$, and are given by
  \begin{equation} 
    \label{eq:hat_Gp}
    \Re \widehat{G_\rp}  = A_+  \pv \left( \dfrac{1}{\cdot - (H_{N+1} - E_{N}^0)} \right) A^\ast_+ 
    \quad \text{and} \quad 
    \Im \widehat{G_\rp}  = - \pi A_+ P^{H_{N+1} - E_N^0} A_+^\ast.
  \end{equation}
  The analytic operator-valued function $\widetilde G_{\rm p}$ defined in the upper half-plane by
  \begin{equation}
    \label{eq:G_p_tilde}
    \forall z \in \UU, \quad \widetilde{G_{\rm p}}(z)  := A_+  \dfrac{1}{z - (H_{N+1} - E_{N}^0)} A_+^\ast
  \end{equation}
  is the Laplace transform of $G_{\rm p}$ and satisfies 
  \[
  \widehat{G_{\rm p}} = \lim\limits_{\eta \to 0^+} \widetilde{G_{\rm p}}(. + \ri \eta) \quad \mbox{in } H^{-s}(\R_\omega,\cB(\cH_1)) \quad \text{for all} \quad s > 1/2.
  \] 
\end{proposition}

The imaginary part of $\widehat{G_\rp}$ is related to the so-called spectral function $\cA_\rp$ (see Section \ref{sec:spectral_functions}).
 
\paragraph{Analytic continuation to the complex plane.}  
Let us introduce the particle optical excitation set
\begin{equation} 
  \label{eq:particle_opt_excitations}
  S_\rp := \sigma (H_{N+1} - E_N^0).
\end{equation}
We recall that the operator $H_{N+1} - E_N^0$ with domain  $\cH_{N+1} \cap H^2(\R^{3(N+1)})$ is self-adjoint on $\cH_{N+1}$. Its essential spectrum is of the form $\sigma_\ess(H_{N+1} - E_N^0) = [\Sigma_{N+1}, \infty)$, and there are possibly infinitely many eigenvalues below $\Sigma_{N+1}$ that can only accumulate at $\Sigma_{N+1}$. According to the HVZ theorem \cite{Hunziker1966, vanWinter1964, Zhislin1960}, $\Sigma_{N+1} = E_N^0 - E_{N}^0 = 0$. In particular, $S_\rp$ is the union of a discrete negative part, and the half-line $[0, + \infty)$.

We next infer from~\eqref{eq:G_p_tilde} that $\widetilde{G_{\rm p}}(z)$ can be extended to an analytic function from $\C \setminus S_\rp$ to $\cB(\cH_1)$. This is of particular interest for the following reason. The operator-valued distribution $\widehat{G_\rp}(\omega)$ is highly peaked and irregular (for instance, its imaginary part is a sum of Dirac measures on the discrete part of $S_\rp$). Instead of studying $\widehat{G_\rp}(\omega)$ on the real axis, we will study its analytic continuation $\widetilde{G_{\rm p}}(z)$ (defined \textit{a priori} only in the upper-half plane, but actually on $\C \setminus S_\rp$) \emph{on the imaginary axis} $\mu + \ri \R$, where $\mu < E_{N+1}^0 - E_N^0 \leq 0$ is an admissible chemical potential (see \textbf{Hyp. 2}). The set $S_\rp$ can be recovered from $\omega \mapsto \widetilde{G_{\rm p}}(\mu + \ri \omega)$ by locating the singularities of $\widehat{G_{\rm p}}$, obtained from $\widetilde{G_{\rm p}}$ either by analytic continuation, or by fitting some parameters~\cite{Rojas1995}. We do not address this interesting numerical reconstruction problem in the present article.
\begin{center}
  \begin{figure}[h]
    \begin{center}
      \begin{tikzpicture}[scale =1]
	\draw (-7,0)--(5,0);
	\draw(0, -3) -- (-0,3);
	\node at (0.3, -0.3) {$0$};
	\draw[red, very thick] (0, 0) -- (5, 0);
	\draw[red] (0, -0.2) -- (0, 0.2);
	\node[red] at (3, -0.5) {$\sigma_\ess(H_{N+1} - E_N^0)$};
	\draw[red] (-1.2, -0.2) -- (-0.8, 0.2);
	\draw[red] (-0.8, -0.2) -- (-1.2, 0.2);
	\node[red] at (-1, -0.5) {$E_{N+1}^0 - E_N^0$};
        \draw[red] (-0.7, -0.2) -- (-0.3, 0.2);
	\draw[red] (-0.3, -0.2) -- (-0.7, 0.2);
	\draw[blue, very thick] (-2.5, -3) -- (-2.5, 3);
	\node[blue] at (-2.7, -0.3) {$\mu$};
	\node at (2.5, 0.5) {$\omega \mapsto \widehat{G_\rp}(\omega)$};
	\draw[->] (1, 0.2) -> (4, 0.2);
	\node at (-4, 1) {$\omega \mapsto \widetilde{G_\rp}(\mu + \ri \omega)$};
	\draw[->] (-2.7, 0.5) -> (-2.7, 2.2);
	\draw[->, very thick] (1, 0.5) arc(20:90:3);
	\node at (2, 2) {analytic continuation};
      \end{tikzpicture}
    \end{center}
    \caption{Illustration of the analytic continuation: from $\omega \mapsto \widehat{G_\rp}(\omega)$ to $\omega \mapsto \widetilde{G_\rp}(\mu + \ri \omega)$.}
    \label{fig:forward_GF}
  \end{figure}
\end{center}

The following lemma makes precise the behavior of the Green's function on the vertical axis $\mu + \ri \RR$. It is a direct consequence of the representation~\eqref{eq:G_p_tilde}.
\begin{lemma}
  \label{lem:ReGp}
  Consider $\mu < E_{N+1}^0 - E_N^0$. Then the function $\omega \mapsto \widetilde{G_\rp}(\mu + \ri \omega)$ is real analytic from $\R_\omega$ to $\cB(\cH_1)$ and is in $L^p(\R_\omega, \cB(\cH_1))$ for all $p > 1$. Moreover, for all $\omega \in \R$,
  \[
  \Re \widetilde{G_\rp}(\mu + \ri \omega) = - A_+ \dfrac{H_{N+1} - E_N^0 - \mu}{\omega^2 + (H_{N+1} - E_N^0 - \mu)^2} A_+^\ast
  \]
  is a negative, bounded, self-adjoint operator on~$\cH_1$ which enjoys the following symmetry property:
  \[
  	\forall \omega \in \R_\omega, \quad  \Re \widetilde{G_\rp}(\mu + \ri \omega) = \Re \widetilde{G_\rp}(\mu - \ri \omega).
  \]
  For any $f\in \cH_1$, the function $\omega \mapsto \bra f | \Re \widetilde{G_\rp}(\mu + \ri \omega) | f \ket$ is non-positive, in $L^1(\R_\omega)$, and
  \begin{equation} \label{eq:fGpf}
  \int_{-\infty}^{+\infty} \left\bra f \left|  \Re \widetilde{G_\rp}(\mu + \ri \omega) \right| f \right\ket \, \rd \omega = -\pi \bra f | (\mathds{1}_{\cH_1} - \gamma_N^0) | f \ket.
  \end{equation}
\end{lemma}
The last assertion comes from the spectral theorem,~\eqref{eq:A+bounded}, and the equality
\[
\forall E > 0, \qquad \int_{-\infty}^{+\infty} \dfrac{E}{\omega^2 + E^2} \, \rd \omega = \pi.
\]

\begin{remark}
  Unfortunately, although $\Re \widetilde{G_\rp}(\mu + \ri \cdot)$ has a sign and~\eqref{eq:fGpf} is satisfied for all $f \in \cH_1$, the function
  $
  	\omega \mapsto \left\| \Re \widetilde{G_\rp}(\mu + \ri \cdot) \right\|_{\cB(\cH_1)}
$
does not belong to $L^1(\R_\omega)$.  This is essentially due to the fact that
  \[
  \sup_{E \ge 0} \left( \dfrac{E}{\omega^2 + E^2}\right) = \dfrac{1}{2 \omega} \notin L^1(\R_\omega).
  \]
\end{remark}

Note that the imaginary part of $\widetilde{G_\rp}(\mu + \ri \omega)$,
\[
\Im \widetilde{G_\rp}(\mu + \ri \omega) = -A_+ \dfrac{\omega}{\omega^2 + (H_{N+1} - E_N^0 - \mu)^2} A_+^\ast,
\]
has no definite sign on~$\R_\omega$, and that, for a generic $f \in \cH_1$, the function $\omega \mapsto \left\bra f \left| \Im \widetilde{G_\rp} (\mu + \ri \omega) \right| f \right\ket$ does not belong to $L^1(\R_\omega)$. It will therefore be more convenient in general to work with the real part of $\widetilde{G_\rp}(\ri \omega)$ only, especially since the imaginary part can be recovered from the real part (see Lemma~\ref{lem:plemelj_iR} below). Indeed, the operator-valued functions $\widetilde{g_{\rp,\eta}}: \omega \mapsto \widetilde{G_\rp}(\mu-\eta + \ri \omega)$ are in $L^2(\R, \cB(\cH_1))$ for any $\eta > 0$, and converge to $\widetilde{g_\rp}: \omega \mapsto \widetilde{G_\rp} (\mu + \ri \omega)$ in $L^2(\R, \cB(\cH_1))$ as $\eta \to 0^+$. We can therefore apply Titchmarsh's theorem (see Theorem~\ref{th:Titchmarsh_Linfty}), which gives the following result.

\begin{lemma}
  \label{lem:plemelj_iR}
  Let $\mu < E_{N+1}^0 - E_N^0$. The function $\widehat{g_\rp} (\omega) := \widetilde{G_\rp}(\mu + \ri \omega)$ is the Fourier transform of the causal function 
  \begin{equation}
    \label{eq:def_g_p}
    g_\rp(\tau) = - \Theta(\tau) A_+ \re^{-\tau(H_{N+1}-E_N^0 - \mu)} A_+^\ast, 
  \end{equation}
  which belongs to $L^2(\R_\tau, \cS(\cH_1))$. In particular, the Plemelj formulae hold true:
  \[
  \Re \widehat{g_\rp} = - \fH \left( \Im \widehat{g_\rp} \right) \quad \mathrm{and} \quad \Im \widehat{g_\rp} = \fH \left( \Re \widehat{g_\rp} \right) \quad \mathrm{in} \ L^2(\R_\omega, \cB(\cH_1)).
  \]
  Moreover, the function $\tau \mapsto \| g_\rp(\tau) \|_{\cB(\cH_1)}$ is exponentially decreasing as $|\tau| \to +\infty$.
\end{lemma}

\begin{remark}
  The exponential decay of $g_\rp$ is consistent with the analyticity of its Fourier transform. This property is of interest when calculating numerically convolutions on the imaginary axis $\mu+\ri\R$, since convolutions can be replaced, up to a Fourier transform, with point-wise multiplications of causal functions which are exponentially decreasing. This approach was advocated in~\cite{RSWRG99}, and is now routinely used in GW computations.
\end{remark}

\subsubsection{The hole (backward) Green's function $G_{\rm h}$}
\label{sec:hole_GF}

\paragraph{Definition and first properties of the hole Green's function.}
Together with the particle Green's function, we introduce the hole (or backward, or advanced) Green's function, formally defined within the second quantization formalism by
\[
\cG_{\rm h}(\br t, \br' t') := \ri \Theta(t' - t) \ \bra \Psi_N^0 | \Psi_{\rm H}^\dagger (\br' t') \Psi_{\rm H}(\br t) | \Psi_N^0 \ket.
\]
Observing that
\[
\int_{\R^3} \Psi(\br) |\Psi_N^0 \rangle f(\br)  \, d\br = a(\overline{f}) |\Psi_N^0\rangle,
\]
we introduce
\[
\begin{array}{llll}
  A_- : & \cH_1 & \to & \cH_{N-1} \\
  & f & \mapsto & a(\overline{f}) |\Psi_N^0\rangle.
\end{array}
\]
Similarly as before, we note that $\cG_{\rm h}(\br t, \br' t')$ only depends on the time difference $t - t'$. Introducing $G_{\rm h}(\br, \br', \tau) := \cG_{\rm h}(\br \tau, \br' 0)$, we see that $G_{\rm h}(\br, \br', \tau)$ is formally the kernel of the following one-body operator.

\begin{definition}
  The hole Green's function is defined as
  \begin{equation}\label{eq:defGh}
    G_{\rm h}(\tau) := \ri \Theta(-\tau) A_-^\ast \re^{ \ri \tau (H_{N-1} - E_N^0)} A_-.
  \end{equation}
\end{definition}

Similarly as in (\ref{eq:A+bounded}), it holds that
\[
A_-^\ast A_- = \gamma_N^0. 
\]
Hence, $A_-$ is a bounded operator from $\cH_1$ to $\cH_{N-1}$, $A_-^\ast$ is a bounded operator from $\cH_{N-1}$ to $\cH_1$, and it holds $\| A_- \|_{\cB(\cH_1,\cH_N)}= \| A_-^\ast \|_{\cB(\cH_{N-1}, \cH_1)} \leq 1$. The properties of the hole Green's function are quite similar to the properties of the particle Green's function (compare with Proposition~\ref{lemma:forward_GF}). 

\begin{proposition}[Properties of the hole Green's function] 
  \label{lemma:backward_GF}
  The family $(G_{\rm h}(\tau))_{\tau \in \R}$ defines a bounded anti-causal operator on $\cH_1$.  The real and imaginary parts of its time-Fourier transform are in $H^{-s}(\R_\omega, \cB(\cH_1))$ for all $s > 1/2$, and are given by
  \begin{equation} \label{eq:hat_Gh}
    \Re \widehat{G_\rh}  = A_-^\ast  \pv \left( \dfrac{1}{\cdot - (E_N^0 - H_{N-1})} \right) A_- \quad \text{and} \quad 
    \Im \widehat{G_\rh}  = \pi A_-^\ast P^{E_N^0 - H_{N-1}}  A_-.
  \end{equation}
  The analytic operator-valued function $\widetilde {G_{\rh}}$ defined in the lower half-plane by
  \begin{equation}
    \label{eq:G_h_tilde}
    \forall z \in \LL, \qquad \widetilde{G_{\rh}}(z)  := A_-^\ast  \dfrac{1}{z - (E_{N}^0-H_{N-1})} A_-  
  \end{equation}
  is the Laplace transform of $G_{\rh}$ and satisfies
  \[ 
  \widehat{G_{\rm h}} = \lim_{\eta \to 0^+} \widetilde{G_{\rm h}}(. - \ri \eta) \quad \mbox{in } H^{-s}(\R_\omega,\cB(\cH)) \quad \text{for all} \quad s > 1/2.
  \]
\end{proposition}

\paragraph{Analytic continuation into the complex plane.}
The hole optical excitation set is defined as
\begin{equation} \label{eq:hole_opt_excitations}
  S_\rh := \sigma(E_N^0 - H_{N-1}).
\end{equation}
It is clear from~\eqref{eq:G_h_tilde} that the operator-valued function $\widetilde{G_{\rh}}$ can be analytically continued to $\C \setminus S_\rh$. Instead of studying the highly irregular distribution $\omega \mapsto \widehat{G_\rh}(\omega)$, it is more convenient to study its analytical continuation $\widetilde{G_{\rh}}$ on the imaginary axis $ \mu + \ri \R$, with $\mu > E_N^0 - E_{N-1}^0$.

\begin{center}
  \begin{figure}[h]
    \begin{center}
      \begin{tikzpicture}[scale =1]
        \draw (-7,0)--(5,0);
	\draw(0, -3) -- (-0,3);
	\node at (0.3, -0.3) {$0$};
	\draw[red, very thick] (-7, 0) -- (-3, 0);
	\draw[red] (-3, -0.2) -- (-3, 0.2);
	\node[red] at (-5, 0.4) {$\sigma_\ess( E_N^0 - H_{N-1})$};
	\draw[red] (-2.1, -0.2) -- (-1.7, 0.2);
	\draw[red] (-1.7, -0.2) -- (-2.1, 0.2);
	\node[red] at (-1.9, 0.5) {$E_{N}^0 - E_{N-1}^0$};
        \draw[red] (-2.7, -0.2) -- (-2.3, 0.2);
	\draw[red] (-2.3, -0.2) -- (-2.7, 0.2);
	\draw[blue, very thick] (-1, -3) -- (-1, 3);
	\node[blue] at (-1.2, -0.3) {$\mu$};
	\node at (-5.5, -0.5) {$\omega \mapsto \widehat{G_\rh}(\omega)$};
	\draw[->] (-6.5, -0.2) -> (-4, -0.2);
	\node at (0.5, -2) {$\omega \mapsto \widetilde{G_\rh}(\mu + \ri \omega)$};
	\draw[->] (-0.8, -3) -> (-0.8, -0.5);
	\draw[->, very thick] (-4, -0.5) arc(200:270:3);
	\node at (-5, -2) {analytic continuation};
      \end{tikzpicture}
    \end{center}
    \caption{Illustration of the analytic continuation: from $\omega \mapsto \widehat{G_\rh}(\omega)$ to $\omega \mapsto \widetilde{G_\rh}(\mu + \ri \omega)$.}
    \label{fig:forward_GF_hole}
  \end{figure}
\end{center}

We can state a result similar to Lemma~\ref{lem:ReGp}.

\begin{lemma} 
  \label{lem:ReGh}
  Consider $\mu > E_N^0 - E_{N-1}^0$. Then the function $\omega \mapsto \widetilde{G_\rh}(\mu + \ri \omega)$ is real analytic from $\R_\omega$ to $\cB(\cH_1)$ and is in $L^p(\R_\omega, \cB(\cH_1))$, for all $p> 1$. Moreover, for all $\omega \in \R$,
  \[
  \Re \widetilde{G_\rh}(\mu + \ri \omega) = A_-^\ast \dfrac{H_{N-1} + \mu - E_N^0}{\omega^2 + (E_N^0 - H_{N-1}-\mu)^2} A_-
  \]
  is a positive, bounded, self-adjoint operator, which enjoys the following symmetry property:
  \[
  	\forall \omega \in \R_\omega, \quad \Re \widetilde{G_\rh}(\mu + \ri \omega) = \Re \widetilde{G_\rh}(\mu - \ri \omega).
  \]
  For any $f \in \cH_1$, the function $\omega \mapsto \left\bra f \left|| \Re \widetilde{G_\rh}(\mu + \ri \omega) \right| f \right\ket$ is non-negative, in $L^1(\R_\omega)$, and
  \[
  \int_{-\infty}^{+\infty} \left\bra f \left| \Re \widetilde{G_\rh}(\mu + \ri \omega) \right| f \right\ket \, \rd \omega = \pi \bra f | \gamma_N^0 | f \ket.
  \]
\end{lemma}
 
\paragraph{The Galitskii-Migdal formula.} The hole Green's function is of particular interest, as it contains useful information on the $N$-body ground state. For instance, from the identity $A_-^\ast A_- = \gamma_N^0$, we directly obtain $G_{\rm h}(0^-) = \ri \gamma_N^0$, so that the expectation value in the ground state of any one-body operator $\sum_{i=1}^N C_{\br_i}$ (for $C \in \cB(\cH_1)$) can be evaluated via
\[
\left\langle \Psi_N^0 \left| \sum_{i=1}^N C_{\br_i} \right| \Psi_N^0 \right\rangle = \Tr_{\cH_1} \left( C \gamma_N^0 \right) = - \ri \, \Tr_{\cH_1} (C G_{\rm h}(0^-)).
\]
This calculation is valid only for one-body operators. It is not possible to obtain the expectation value in the ground state of a generic two-body operator from the one-body Green's function. This is however the case for the ground state energy (the expectation value of the two-body Hamiltonian $H_N$ in the ground state), as was first shown by Galiskii and Migdal~\cite{Galiskii1958}. Alternative formulae for the ground state energy are provided by the Luttinger-Ward formula~\cite{Luttinger1960} and the Klein's formula~\cite{Klein1961}.

\begin{theorem}[Galitskii-Migdal formula] \label{thm:GM_formula} 
  For all $N \ge 2$, the ground state energy can be recovered as
  \begin{align}
	  E_N^0 & =   \dfrac12 \Tr_{\cH_1} \left( - A_-^\ast \left( H_{N-1} - E_N^0 \right) A_- + \left( - \frac12 \Delta + v_\ext \right) A_-^\ast A_- \right) \label{eq:GMglobal} \\
	  	& = \dfrac12 \Tr_{\cH_1} \left[  \left( \frac{d}{d\tau} - \ri  \left( - \frac12 \Delta + v_\ext \right)  \right)  G_{\rm h}(\tau) \Big|_{\tau = 0^-} \right]. \label{eq:GM1}
  \end{align}
\end{theorem}
 
The proof of this theorem can be read in Section~\ref{sec:proof:thm:GM_formula}. Formula~\eqref{eq:GM1} is one way to obtain the right-hand side of ~\eqref{eq:GMglobal}, and is the one found in the original article~\cite{Galiskii1958}. There are however other ways to obtain~\eqref{eq:GMglobal} from the hole Green's function, without the use of derivative (which are cumbersome to evaluate numerically). One can for instance use the following equality, that we do not prove for the sake of brevity,
\[
	\Tr_{\cH_1} \Big( A_-^\ast \left( H_{N-1} + \mu - E_N^0 \right) A_- \Big) = \lim_{\omega \to \infty} \omega^2 \Tr_{\cH_1} \left(\Re \widetilde{ G_\rh} (\mu + \ri \omega) \right).
\]


\subsubsection{The time-ordered Green's function $G$}

It is often claimed in the physics literature that the main object of interest is neither the particle nor the hole Green's function, but the function
\[
\cG (\br t, \br' t') = \cG_{\rm p}(\br t, \br' t') + \cG_{\rm h}(\br t, \br' t'),
\]
called the time-ordered Green's function, which can be seen as a convenient way to concatenate the information contained in the particle and hole Green's functions. Obviously, the time-ordered Green's function only depends on the time difference $\tau = t - t'$ and $\cG (\br t, \br' t') = \cG_{\rm p}(\br (t-t'), \br' 0) + \cG_{\rm h}(\br (t-t'), \br' 0)$. In view of (\ref{eq:defGp}) and (\ref{eq:defGh}), our definition of the time-ordered Green's function therefore is the following.

\begin{definition}[Green's function]
  The (time-ordered) Green's function is the family of bounded operators $(G(\tau))_{\tau \in \RR}$ defined as
  \[
  G(\tau) = G_\rp(\tau) + G_\rh(\tau) = - \ri \Theta(\tau) \, A_+ \re^{- \ri \tau (H_{N+1} - E_N^0)} A_+^\ast + \ri \Theta(-\tau) \, A_-^\ast \re^{ \ri \tau (H_{N-1} - E_N^0)} A_-.
  \]
\end{definition}  

The following results straightforwardly follow from Propositions \ref{lemma:forward_GF} and \ref{lemma:backward_GF}, as well as Lemmas~\ref{lem:ReGp} and~\ref{lem:ReGh}. We recall that $\mu$ is a chemical potential of the electrons for the ground state $\Psi_N^0$ of the reference system, and that $E_N^0 - E_{N-1}^0 < \mu < E_{N+1}^0 - E_N^0$. In the following, we introduce some $C^\infty(\R_\omega)$ cut-off functions $\phi_\pm$ satisfying $0\le \phi_\pm \le 1$, $\phi_+ + \phi_-=1$, ${\rm Supp}(\phi_+) \subset (E_N^0 - E_{N-1}^0,+\infty)$ and ${\rm Supp}(\phi_-) \subset (-\infty,E_{N+1}^0 - E_{N}^0)$ (see Figure~\ref{fig:cut_off}). These cut-off functions allow us to write properties of the Green's function in the time representation without specifying whether $\tau$ is positive or negative.

\begin{figure}[!h]
  \begin{center}
    \begin{tikzpicture}[scale =1]
      \draw (-7,0)--(7,0);
      \draw(0, -3) -- (0,3);
      \draw[red, very thick] (0, 0) -- (7, 0);
      \draw[red] (0, -0.2) -- (0, 0.2);
      \draw[red] (-1.2, -0.2) -- (-0.8, 0.2);
      \draw[red] (-0.8, -0.2) -- (-1.2, 0.2);
      \node[red] at (-2.4, -1.1) {$E_N^0 - E_{N-1}^0$};
      \draw[red] (-0.7, -0.2) -- (-0.3, 0.2);
      \draw[red] (-0.3, -0.2) -- (-0.7, 0.2);
      \draw[red, very thick] (-3, 0) -- (-7, 0);
      \draw[red] (-3, -0.2) -- (-3, 0.2);
      \draw[red] (-2, -0.2) -- (-2.4, 0.2);
      \draw[red] (-2.4, -0.2) -- (-2, 0.2);
      \node[red] at (-0.5, -0.5) {$E_{N+1}^0 - E_{N}^0$};
      \draw[blue, very thick] (-1.65, -0.3) -- (-1.65, 0.3);
      \node[blue] at (-1.65, -0.6) {$\mu$};
      \draw[blue] (-6,2) -- (-2.3, 2) node[above] {$\phi_-$};
      \draw[blue]  (-2.3,2)	.. controls (-1.5,2) and (-1.8,0) .. (-1,0);
      \draw[blue] (-1,2) -- (6, 2) node[above] {$\phi_+$};
      \draw[blue]  (-2.3,0)	.. controls (-1.5,0) and (-1.8,2) .. (-1,2);
    \end{tikzpicture}
    \caption{\label{fig:cut_off} The cut-off functions $\phi_{\pm}$.}
  \end{center}
\end{figure}
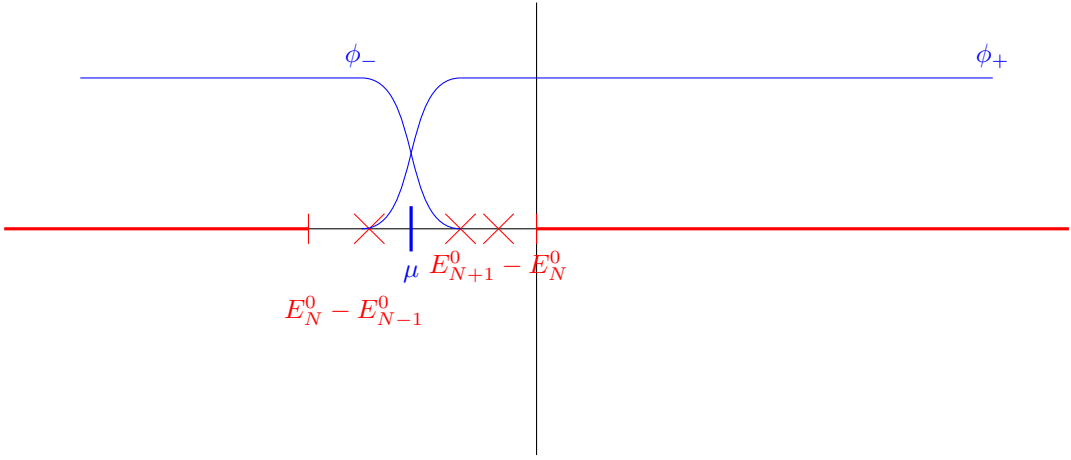

\begin{proposition}[Properties of the Green's function]
  \label{lemma:time-ordered_GF}
  The Fourier transform $\widehat{G} = \widehat{G_{\rm p}} +  \widehat{G_{\rm h}}$ is in $H^{-s}(\R_\omega,\cB(\cH_1))$ for any $s > 1/2$. The operator-valued analytic function $\widetilde G$ defined on the physical Riemann sheet  $\C \setminus \left( S_{\rm p} \cup S_{\rm h} \right)$ by  
  \begin{equation}
    \label{eq:widetilde_G}
  \forall z \in \C \setminus \left( S_{\rm p} \cup S_{\rm h} \right), 
  \qquad		
  \widetilde{G}(z)  := A_+ \dfrac{1}{ z - (H_{N+1} - E_{N}^0)} A_+^\ast  + A_-^\ast \dfrac{1}{ z - (E_N^0-H_{N-1}) } A_- 
  \end{equation}
  is such that 
  \[
  \lim\limits_{\eta \to 0+} \phi_\pm \widetilde{G}(\cdot \pm \ri \eta) =  \phi_\pm \widehat{G} \quad \mathrm{in} \ H^{-s}(\R_\omega,\cB(\cH_1)) \quad \mathrm{for}\,\mathrm{all} \ s > 1/2.
  \]
  The function $\omega \mapsto \widetilde{G}(\mu + \ri \omega)$ is real analytic from $\R_\omega$ to $\cB(\cH_1)$, and is in $L^p(\R_\omega, \cB(\cH_1))$ for all $p > 1$. Moreover, it satisfies the symmetry property
  \[
  	\forall \omega \in \R_\omega, \quad \Re \widetilde{G}(\mu + \ri \omega) = \Re \widetilde{G}(\mu - \ri \omega).
  \]
  For any $f \in \cH_1$, the function $\omega \mapsto \bra f | \Re \widetilde{G}(\mu + \ri \omega)| f \ket$ is in $L^1(\R_\omega)$, and
  \[
  \int_{-\infty}^{+\infty} \left\bra f \left| \Re \widetilde{G}(\mu + \ri \omega) \right| f \right\ket \, \rd \omega = -\pi \bra f | ( \mathds{1}_{\cH_1} - 2 \gamma_N^0) | f \ket.
  \]
 \end{proposition}
 
\subsubsection{The spectral functions $\cA_{\rp}$, $\cA_{\rh}$ and $\cA$} 
\label{sec:spectral_functions}

Spectral functions are essential tools to study many-body effects since they are concentrated on (subsets of) the particle and hole excitation sets. 

\begin{definition}[Spectral functions]
The \emph{particle spectral function} is the operator-valued Borel measure on $\R_\omega$ defined by
\begin{equation} \label{eq:def_particle_SF}
  \forall b \in \mathscr{B}(\R_\omega), \qquad \cA_\rp(b) = -\dfrac{1}{\pi} \Im \widehat{G_\rp}(b) = A_+ P_b^{H_{N+1} - E_N^0} A_+^\ast.
\end{equation}
The \emph{hole spectral function} is similarly defined:
\[
\forall b \in \mathscr{B}(\R_\omega), \qquad \cA_\rh(b)  = \dfrac{1}{\pi} \Im \widehat{G_\rh}(b) = A_-^\ast P_b^{E_N^0 - H_{N-1}} A_-.
\]
The \emph{time-ordered spectral function} is then obtained as $\cA = \cA_\rp + \cA_\rh$.
\end{definition}

With those definitions, the following lemma is straightforward, and is usually referred to as the sum-rule for spectral functions (see for instance~\cite[Section~4.5]{Farid99}).

\begin{proposition} 
  \label{prop:spectral_G}
  The spectral functions $\cA_\rp$, $\cA_\rh$ and $\cA$ are $\cS(\cH_1)$-valued Borel measures on~$\R_\omega$, with supports contained in $S_\rp$, $S_\rh$ and $S_\rp \cup S_\rh$ respectively.
  For all $b \in \mathscr{B}(\R_\omega)$, $\cA_\rp(b)$, $\cA_\rh(b)$ and $\cA(b)$ are bounded positive self-adjoint operators on $\cH_1$ with norms lower or equal to $1$. Moreover, $0 \leq \cA_\rp(b_1) \leq \cA_\rp(b_2)$ as self-adjoint operators when $b_1 \subset b_2$ (and similar inequalities for $\cA_\rh$ and $\cA$), and it holds 
  \[
  \cA_\rp(\R_\omega) = \mathds{1}_{\cH_1} - \gamma_N^0, \qquad \cA_\rh(\R_\omega) = \gamma_N^0, \qquad \cA(\R_\omega) = \mathds{1}_{\cH_1}.
  \]
\end{proposition}

Finally, the Plemelj formulae (\ref{eq:Plemelj_formulae}) allow us to recover the real part of the Green's functions from the spectral functions: $\Re \widehat{G_{\rp}} = \pi \fH(\cA_\rp)$ and $\Re \widehat{G_{\rh}} = \pi \fH(\cA_\rh)$. It therefore holds $\Re \widehat{G} = \pi \fH \cA$.

\subsection{Linear response operators} 
\label{sec:linear_response}

We study in this section the reducible polarizability operator $\chi$, which can be defined from the so-called charge-fluctuation operator introduced in Section~\ref{sec:def_rho_H}. We give a precise mathematical meaning to $\chi$ in Section~\ref{ssec:chi}, and prove Johnson's sum-rule \cite{Johnson1974} for $\chi$ in Section~\ref{sec:sum_rule}. We finally define the dynamically screened Coulomb interaction operator (see Section~\ref{subsec:W}).

\subsubsection{The charge-fluctuation operator $\rho_{\rm H}$}
\label{sec:def_rho_H}

The charge-fluctuation operator is defined, within the second quantization formalism, by (see~\cite[Equation~(97)]{Farid99})
\begin{equation*} 
  \rho_{\rm H}(\br t) := \Psi_{\rm H}^\dagger(\br t) \Psi_{\rm H}(\br t) -  \rho_N^0(\br),
\end{equation*}
so that the action of this operator on the $N$-body ground state is
\begin{equation}
  \label{eq:action_rho_H_PsiN}
  \begin{aligned}
    \rho_{\rm H}(\br t) | \Psi_N^0 \ket & =  \left( \re^{\ri t (H_N - E_N^0)} \right) \Psi^\dagger(\br) \Psi(\br) |\Psi_N^0 \ket - \rho_N^0(\br) | \Psi_N^0 \ket \\
    & = \left( \re^{\ri t (H_N - E_N^0)} \right) \left( \Psi^\dagger(\br) \Psi(\br) - \rho_N^0(\br) \right) | \Psi_N^0 \ket.
  \end{aligned}
\end{equation}
In order to define more rigorously $\rho_{\rm H}$, we need to introduce functional spaces of charge densities (the Coulomb space) and electrostatic potentials. The complex-valued Coulomb space
\begin{equation} \label{def:cC}
  \cC := \left\{ f \in \mathscr{S}'(\R^3,\C) \; \left| \; \widehat f \in L^1_\loc(\R^3,\C), \; | \cdot |^{-1} \widehat f(\cdot) \in L^2(\R^3,\C) \right.\right\},
\end{equation}
is endowed with the inner product
\[
\langle f_1 | f_2 \rangle_{\cC} = 4\pi \int_{\R^3} \frac{\overline{\widehat{f_1}(\bk)}\widehat{f_2}(\bk)}{|\bk|^2} \, \rd \bk,
\]
where the normalization condition for the space-Fourier transform is chosen such that its restriction to $L^2(\R^3, \C)$ is a unitary operator. The space~$\cC$ is a Hilbert space, and $L^{6/5}(\R^3,\C) \hookrightarrow \cC$ thanks to the Hardy-Littlewood-Sobolev inequality (upon rewriting the products in Fourier space as convolutions). The dual of $\cC$ (taking $L^2(\R^3, \C)$ as a pivoting space) is
\begin{equation} \label{def:cC'}
  \cC' := \left\{ v \in L^6(\R^3,\C) \; \left| \; \nabla v \in \left( L^2(\R^3,\C) \right)^3 \right. \right\},
\end{equation}
endowed with the inner product
\[
\langle V_1 |V_2 \rangle_{\cC'} := \frac{1}{4\pi} \int_{\R^3} \overline{\nabla V_1} \cdot \nabla V_2 = 
\frac{1}{4\pi}\int_{\R^3} |\bk|^2\overline{\widehat{V}_1(\bk)} \, \widehat{V}_2(\bk)\, \rd \bk.
\]

We also introduce the Coulomb operator $\vc$, defined as the multiplication operator by $4\pi |\bk|^{-2}$ in the Fourier representation, and its square root $\vc^{1/2}$, defined as the multiplication operator by $(4\pi)^{1/2} |\bk|^{-1}$ in the Fourier representation. The following result, whose proof is a straightforward consequence of the above definitions, will be repeatedly used throughout this article.   

\begin{lemma} 
  \label{lem:vc}
  The operator $\vc$ defines a unitary operator from $\cC$ to $\cC'$.  The operator $\vc^{1/2}$ defines a unitary operator from $\cC$ to $\cH_1$, as well as a unitary operator from $\cH_1$ to $\cC'$.
\end{lemma} 
It follows that the adjoint of the unitary operator $\vc : \cC \rightarrow \cC'$ is the unitary operator $\vc^* = \vc^{-1}: \cC' \rightarrow \cC$.

We are now able to reformulate the charge-fluctuation operator in the ground state as a well defined bounded operator. For $v \in C^\infty_c(\R^3, \C)$, it formally holds
\[
\left( \int_{\R^3} \left( \Psi^\dagger(\br) \Psi(\br) - \rho_N^0(\br) \right) | \Psi_N^0 \ket \, v(\br) \rd \br \right) (\br_1, \ldots, \br_N) = \left[ \left( \sum_{i=1}^N v(\br_i) \right) - \int_{\R^3} v \rho_N^0 \right] \Psi_N^0 (\br_1, \ldots \br_N).
\]
In order to rewrite more rigorously this equality, we introduce the operator
\begin{equation} \label{eq:B0}
  \begin{array}{rrcl}
    B : & \cC' & \to & \cH_N \\
    & v & \mapsto & \dps \left[\left( \sum_{i=1}^N v(\br_i) \right) - \bra v,  \rho_N^0 \ket_{\cC',\cC} \right] | \Psi_N^0 \ket,
  \end{array}
\end{equation}
which is well defined since $\rho_N^0 \in L^{6/5}(\R^3, \R)$ by Proposition~\ref{prop:Psi_N^0}. In fact, as made clear in Lemma~\ref{lem:B0} below, $B$ is bounded. In view of~\eqref{eq:action_rho_H_PsiN}, we can finally define the application to $\Psi_N^0$ of the charge-fluctuation operator $\rho_{\rm H}(t)$ as follows:
\begin{equation} 
  \label{eq:rhoH}
  \rho_{\rm H}(t) | \Psi_N^0 \ket = \re^{\ri t (H_N - E_N^0)} B.
\end{equation}

Let us conclude this section by giving some properties of the operators introduced above (see Section~\ref{sec:proof:lem:B0} for the proof). 

\begin{lemma} 
  \label{lem:B0}
  The operator $B$ defined by (\ref{eq:B0}) is a bounded operator from $\cC'$ to $\cH_N$. Its adjoint $B^*$ is a bounded operator from $\cH_N$ to $\cC'$ which satisfies $B^* | \Psi_N^0 \ket = 0$. As a consequence, $\rho_{\rm H}| \Psi_N^0 \ket \in L^\infty(\R_t,\cB(\cC',\cH_N))$, and $ \left( \rho_{\rm H}| \Psi_N^0 \ket \right)^* \in L^\infty(\R_t,\cB(\cH_N,\cC'))$.
\end{lemma}

\subsubsection{The (symmetrized) reducible polarizability operator $\chi$} 
\label{ssec:chi}

\paragraph{Definition of the reducible polarizability operator.}
The reducible polarizability operator $\chi(t,t')$ is the operator giving the response of the density of the system to perturbations of the external potential. It is formally defined by its kernel (see~\cite[Equation~(96)]{Farid99})
\begin{equation} \label{eq:chi}
  \chi(\br t, \br' t') := - \ri \left\bra\left. \Psi_N^0 \right| \mathcal{T} \left\{ \rho_{\rm H}(\br t) \rho_{\rm H}(\br' t') \right\} \left| \Psi_N^0 \right.\right\ket_{\cH_N}.
\end{equation}
In the above equation, $\rho_{\rm H}$ is the charge-fluctuation operator whose action on $\Psi_N^0$ is defined by~(\ref{eq:rhoH}), and $\mathcal{T}$ stands for the bosonic time-ordering operator:
\[
\mathcal{T} \left\{ A_1(t) A_2(t') \right\} = \left| 
\begin{array}{l} 
A_1(t) A_2(t') \quad \mathrm{if} \ t' < t, \\ 
A_2(t') A_1(t) \quad \mathrm{if} \ t' > t.
\end{array} \right.
\]
In view of~\eqref{eq:rhoH}, the definition~\eqref{eq:chi} of the kernel is formally equivalent to the following identity, stated for $t' < t$ (a similar equality being true for $t' > t$):
\[
\begin{aligned}
\int_{\R^3} \overline{f} \left( \chi(t, t') g \right) & = -\ri \int_{\R^3}\int_{\R^3} \overline{f}(\br) \left\bra\left. \Psi_N^0 \right| \rho_{\rm H}(\br t) \rho_{\rm H}(\br' t') \left| \Psi_N^0 \right.\right\ket_{\cH_N} g(\br') \, \rd\br \, \rd\br' \\
& = -\ri\left\bra\left. \int_{\R^3} f(\br) \rho_{\rm H}(\br t) \Psi_N^0 \, \rd\br \, \right| \int_{\R^3} g(\br') \rho_{\rm H}(\br' t') \Psi_N^0 \, \rd\br'\right\ket_{\cH_N} \\
& = -\ri \left\bra \left. \re^{\ri t(H_N - E_N^0)} Bf \, \right| \re^{\ri t'(H_N - E_N^0)} Bg \right\ket_{\cH_N} \\
& = -\ri \left\bra  f \left| B^* \re^{-\ri (t-t')(H_N - E_N^0)} Bg \right.\right\ket_{\cC'}.
\end{aligned}
\]
In order to interpret $\chi$ as giving the variation of the ground state density (an element of~$\cC$) generated by a variation of the external potential (an element of~$\cC'$), we rewrite the scalar product in~$\cC'$ as a duality braket between $\cC'$ and $\cC$:
\begin{equation}
\label{eq:from_cC'_to_duality_cC_cC'}
\left\bra f_1 \left| f_2 \right.\right\ket_{\cC'} = \left\bra \overline{f_1}, \vc^{-1} f_2 \right\ket_{\cC',\cC}.
\end{equation}
This motivates defining $\chi(t, t')$ as the bounded operator from $\cC'$ to $\cC$ given by
\[
\chi( t, t') = - \ri  \vc^{-1} B^* \re^{- \ri | t - t' | (H_N - E_N^0)} B.
\]
In particular, $\chi(t,t')$ only depends on the time difference $t - t'$, and we write in the sequel $\chi(\tau) :=  \chi(\tau, 0)$:
\begin{equation} 
  \label{def:chi2}
  \chi(\tau) = - \ri  \vc^{-1} B^* \re^{- \ri | \tau | (H_N - E_N^0)} B.
\end{equation}

It turns out to be useful to symmetrize the action of the polarizability operator using appropriate Coulomb operators. Recall that $B\vc^{1/2} \in \cB(\cH_1,\cH_N)$ while $(B\vc^{1/2})^* = \vc^{-1/2} B^* \in \cB(\cH_N,\cH_1)$.

\begin{definition}
The symmetrized reducible polarizability operator $\chi_{\rm sym} \in L^\infty(\R_\tau, \cB(\cH_1))$ is defined by
\[
	\forall \tau \in \R_\tau, \quad \chi_{\rm sym}(\tau) = \vc^{1/2}  \chi(\tau) \vc^{1/2}  = - \ri \vc^{-1/2} B^* \re^{- \ri | \tau | (H_N - E_N^0)} B\vc^{1/2}.
\]
\end{definition}

It is convenient to decompose the symmetrized reducible polarizability operator into two parts, namely its causal part and its anti-causal part:
\begin{equation}
  \label{eq:chi+-}
  \chi_\sym(\tau) = \chi_\sym^+(\tau) + \chi_\sym^-(\tau) \quad \text{with} \quad \chi_\sym^{\pm} (\tau) = \Theta(\pm \tau) \left( - \ri \vc^{-1/2}  B^* \re^{- \ri | \tau | (H_N - E_N^0)} B \vc^{1/2} \right).
\end{equation}
In the above expression, the Hamiltonian $H_N$ can be replaced by 
\begin{equation*}
	H_N^\sharp := H_N \big|_{\{\Psi_N^0\}^\perp} .
\end{equation*}
This is a consequence of Lemma~\ref{lem:B0} which shows that ${\rm Ran} \left(B\right) \subset \left\{ \Psi_N^0 \right\}^\perp$. Note that $H_N^\sharp - E_N^0 \ge E_N^1 - E_N^0$.

\paragraph{Properties of the symmetrized reducible polarizability operator.}
As rigorously stated below, the symmetrized polarizability operator has singularities at the energy differences corresponding to excitation energies for a system with a fixed number $N$ of electrons, called neutral excitations in~\cite[Section~8]{Farid99}. We therefore introduce the neutral excitation set
\[
S_0^+ := \sigma(H_N -E_N^0) \setminus \{ 0 \} = \sigma \left( H_N^\sharp - E_N^0 \right),
\]
its reflection $S_0^-:=-S_0^+$ and $S_0 := S_0^+ \cup S_0^-$. Note that $S_0^+ \subset [E_N^1-E_N^0,+\infty)$ so that $S_{0}^- \cap S_{0}^+ = \emptyset$.

As for Proposition~\ref{lemma:time-ordered_GF}, it turns out to be convenient to introduce appopriate cut-off functions. Consider $\phi^1_\pm$ such that $\phi^1_-$ and $\phi^1_+$ are in $C^\infty(\R_\omega)$ and satisfy $0\le \phi^1_\pm \le 1$, $\phi^1_+ + \phi^1_-=1$, ${\rm Supp}(\phi^1_+) \subset (-(E_N^1 - E_N^0),+\infty)$ and ${\rm Supp}(\phi^1_-) \subset (-\infty,E_{N}^1 - E_{N}^0)$ (see Figure~\ref{fig:cut_off_fct_chi}).

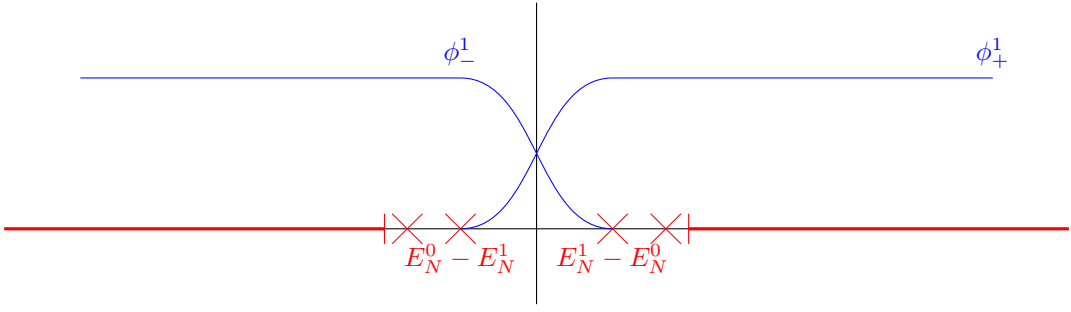
\begin{figure}[!h]
  \begin{center}
    \begin{tikzpicture}[scale =1]
      \draw (-7,0)--(7,0);
      \draw(0, -1) -- (0,3);
      \draw[red, very thick] (2, 0) -- (7, 0);
      \draw[red] (2, -0.2) -- (2, 0.2);
      \draw[red] (1.2, -0.2) -- (0.8, 0.2);
      \draw[red] (0.8, -0.2) -- (1.2, 0.2);
      \draw[red] (1.5, -0.2) -- (1.9, 0.2);
      \draw[red] (1.9, -0.2) -- (1.5, 0.2);
      \node[red] at (1, -0.4) {$E_N^1 - E_N^0$};
      \draw[red, very thick] (-2, 0) -- (-7, 0);
      \draw[red] (-2, -0.2) -- (-2, 0.2);
      \draw[red] (-1.2, -0.2) -- (-0.8, 0.2);
      \draw[red] (-0.8, -0.2) -- (-1.2, 0.2);
      \draw[red] (-1.5, -0.2) -- (-1.9, 0.2);
      \draw[red] (-1.9, -0.2) -- (-1.5, 0.2);
      \node[red] at (-1, -0.4) {$E_N^0 - E_N^1$};
      \draw[blue] (-6,2) -- (-1, 2) node[above] {$\phi^1_-$};
      \draw[blue]  (-1,2)	.. controls (0,2) and (0,0) .. (1,0);
      \draw[blue] (1,2) -- (6, 2) node[above] {$\phi^1_+$};
      \draw[blue]  (-1,0)	.. controls (0,0) and (0,2) .. (1,2);
    \end{tikzpicture}
    \caption{\label{fig:cut_off_fct_chi}The cut-off functions $\phi^1_{\pm}$.}
  \end{center}
\end{figure}

\begin{proposition} 
  \label{lem:chi^pm} 
  The symmetrized reducible polarizability operator $\chi_{\rm sym}$ satisfies the following properties:
  \begin{enumerate}[(1)]
  \item  $(\chi^+_{\rm sym}(\tau))_{\tau \in \R}$ is a bounded causal operator on $\cH_1$ while $(\chi_{\rm sym}^-(\tau))_{\tau \in \R}$ is a bounded anti-causal operator on $\cH_1$. They satisfy the following symmetry properties:
    \begin{equation} \label{chi12}
      \forall \tau \in \R, \quad \chi_\sym(-\tau) = \chi_\sym(\tau) \quad \mbox{and} \quad \chi_\sym^+(\tau) = \chi_\sym^-(-\tau);
    \end{equation} 
  \item the real and imaginary parts of the time-Fourier transforms of $\chi^+_{\rm sym}, \chi^-_{\rm sym}$ are respectively given by
    \[
    \Re \widehat{ \chi_{\rm sym}^\pm} = \pm \vc^{-1/2} B^* \pv \left(  \dfrac{1}{\cdot \mp (H_N^\sharp - E_N^0)} \right) B\vc^{1/2},
    \]
    and
    \[
    \Im \widehat{\chi_{\rm sym}^{\pm}} =  - \pi \vc^{-1/2} B^* P^{\pm(H_N^\sharp - E_N^0)} B \vc^{1/2}.
    \]
    In particular, $\mbox{\rm Supp}\left(\Im \widehat{\chi_{\rm sym}^{\pm}}\right) \subset S_0^\pm$ and $\mbox{\rm Supp}\left(\Im \widehat{\chi_{\rm sym}}\right) \subset S_0$;
  \item consider the $\cB(\cH_1)$-valued analytic functions $\widetilde{\chi_{\rm sym}^+}$, $\widetilde{\chi_{\rm sym}^-}$ and $\widetilde {\chi_{\rm sym}}$ respectively defined by
    \[
    \forall  z \in \C \setminus S_0^\pm, \qquad  \widetilde{\chi_{\rm sym}^\pm} (z)  := \pm \vc^{-1/2} B^* \dfrac{1}{z\mp(H_N^\sharp - E_N^0)} B\vc^{1/2},
    \]
    and 
    \begin{equation}
      \label{eq:chi_tilde_sym}
      \forall  z \in \C \setminus S_0, \qquad  \widetilde{\chi_{\rm sym}}(z)  := \widetilde{\chi_{\rm sym}^+} (z) + \widetilde{\chi_{\rm sym}^-} (z)=  - \vc^{-1/2} B^* \frac{2(H_N^\sharp - E_N^0)}{(H_N^\sharp - E_N^0)^2-z^2} B\vc^{1/2}. 
    \end{equation}
    It holds 
    \[
    \forall z \in \C \setminus S_0^+, 
    \qquad 
    \widetilde{\chi^+_{\rm sym}}(z) = \widetilde{\chi^-_{\rm sym}}(-z) = \left( \widetilde{\chi^+_{\rm sym}}(\overline{z})\right)^*
    \]
    and
    \[
    \forall z \in \C \setminus S_0, \quad \widetilde{\chi_{\rm sym}}(z) = \widetilde{\chi_{\rm sym}}(-z) = \left (\widetilde{\chi_{\rm sym}}(\overline{z}) \right)^\ast.
    \] 
    The functions $\widetilde{\chi_{\rm sym}^+}|_\UU$ and $\widetilde{\chi_{\rm sym}^-}|_\LL$ are respectively the Laplace transforms of $ \chi_{\rm sym}^+ $ and $ \chi_{\rm sym}^- $, and the following convergences hold in $H^{-s}(\R_\omega,\cB(\cH_1))$ for all $s > 1/2$:
    \[
    \lim\limits_{\eta \to 0^+}  \widetilde{\chi_{\rm sym}^\pm}(\cdot \pm \ri \eta) = \widehat{\chi_{\rm sym}^\pm}, 
    \qquad
    \lim_{\eta \to 0^+} \phi^1_\pm \widetilde{\chi_{\rm sym}}(\cdot \pm \ri \eta)  = \phi^1_\pm \widehat{\chi_{\rm sym}}; 
    \]
  \item for all $\omega \in \left( -(E_N^1 - E_N^0), E_N^1 - E_N^0 \right)$,  $\widetilde{\chi_{\rm sym}}(\omega)= \widehat{\chi_{\rm sym}}(\omega)$ is a negative bounded self-adjoint operator on $\cH_1$;
  \item for all $\omega \in \R$, $ \widetilde{\chi_{\rm sym}}(\ri \omega)$ is a negative bounded self-adjoint operator on $\cH_1$.
  \end{enumerate}
\end{proposition}

We omit the proof of Proposition~\ref{lem:chi^pm} since the first three assertions are similar to those of Lemma~\ref{lemma:time-ordered_GF}, while the last two ones are direct consequences of~\eqref{eq:chi_tilde_sym}. 

\paragraph{On the integrability of $\widetilde{\chi_\sym}(\ri \omega)$.}
As for the Green's function, $\omega \mapsto \widehat{\chi_\sym}(\omega)$ is difficult to study on the real-axis, and it is more convenient to study its analytical continuation $\widetilde{\chi_\sym}$ on the imaginary axis $\ri \R$. This is possible thanks to the existence of the gap $( - (E_N^1 - E_N^0), E_N^1 - E_N^0)$ around~$0$. The representation provided in Proposition~\ref{lem:chi^pm} allows one to directly deduce the integrability properties of the functions $\omega \mapsto \widetilde{\chi_\sym}(\ri \omega)$ (as in Lemma~\ref{lem:ReGp}). 

\begin{corollary} 
  \label{lem:chi_L1}
  The functions $\omega \mapsto \widetilde{\chi^\pm_\sym} (\ri \omega)$ are real-analytic from $\R_\omega$ to $\cS(\cH_1)$, and are in $L^p(\R_\omega, \cS(\cH_1))$ for all $p > 1$. For any $f \in \cH_1$, the function $\omega \mapsto \bra f | \widetilde{\chi_\sym}(\ri \omega) | f \ket$ is non-positive and in $L^1(\R_\omega)$, and it holds
    \begin{equation} \label{eq:int_chi}
      \int_{-\infty}^{+\infty} \bra f | \widetilde{\chi_\sym} (\ri \omega) | f \ket \rd \omega = - 2 \pi \left\| B \vc^{1/2} f \right\|_{\cH_N}^2.
    \end{equation}
\end{corollary}

\subsubsection{The sum-rule for the reducible polarizability operator $\chi$} 
\label{sec:sum_rule}

The behavior of the reducible polarizability operator in the high imaginary-frequency regime is well understood. This asymptotic behavior is given by the so-called Johnson's sum-rule~\cite{Johnson1974} or $f$-sum rule, the latter terminology being motivated in~\cite[Section~8.8]{Farid99} by the fact that it can formally be seen as some equality involving the {\em f}irst moment of $\mathrm{Im}\widehat{\chi_{\rm sym}}$. Knowing the large-$\omega$ behavior of $\widetilde{\chi_\sym}$ is important to design appropriate approximate operators, used in plasmon-pole models to avoid the numerical inversion of the dielectric operator (which is computationally expensive). 

The fifth point of Proposition~\ref{lem:chi^pm} implies that for all $\omega \in \R_\omega$, the operator $-\widetilde \chi(\ri \omega) :=-\vc^{-1/2}\widetilde{\chi_{\rm sym}} (\ri \omega)\vc^{-1/2}$ defines a symmetric, continuous, non-negative sesquilinear form on $\cC'$:
\[
\forall (f,g) \in \cC' \times \cC', \qquad \langle \overline{f}, -\widetilde \chi(\ri \omega) g \rangle_{\cC',\cC} = \bigg\langle Bf \bigg| \frac{2(H_N^\sharp - E_N^0)}{(H_N^\sharp - E_N^0)^2+\omega^2} \bigg| Bg \bigg\rangle_{\cH_N},
\]
so that, formally, 
\[
\lim_{\omega \to \pm \infty}  \langle \overline{f}, -\omega^2 \widetilde \chi(\ri \omega) g \rangle_{\cC',\cC}  
	= 2 \bra Bf | H_N^\sharp - E_N^0 | Bg \ket_{\cH_N}
	= 2 \left\langle \overline{f}, \vc^{-1}B^*\left(H_N^\sharp - E_N^0\right)Bg \right\rangle_{\cC',\cC}.
\]
The following theorem, whose proof is postponed until Section~\ref{sec:proof:lem:sumrule_chi}, confirms that this limit exists and allows one to identify it. 

\begin{theorem}[Johnson's sum rule] 
  \label{lem:sumrule_chi}
  The operator $2 \vc^{-1} B^* (H_N^\sharp - E_N^0) B$ is bounded from $\cC'$ to $\cC$, and $2 \vc^{-1} B^* (H_N^\sharp - E_N^0) B = - \div \left( \rho_N^0 \nabla  \cdot \right)$. Moreover, the following weak convergence holds:
  \[
  \forall (f,g) \in \cC' \times \cC', 
  \qquad 
  \lim_{\omega \to \pm \infty} \left\bra \overline{f}, - \omega^2 \widetilde{\chi}(\ri \omega)  g \right\ket_{\cC',\cC} =  \left\bra \overline{f}, -\div(\rho_N^0 \nabla g) \right\ket_{\cC',\cC} =  \int_{\R^3} \rho_N^0  \overline{\nabla f} \cdot \nabla g.
  \]
  For all $g \in \cC'$ such that $\Delta g \in L^2(\R^3)$, the following strong convergence holds:
  \[
  \lim_{\omega \to \pm \infty} \omega^2 \widetilde{\chi}(\ri \omega)  g =  \div\left(\rho_N^0 \nabla g\right) \quad \mathrm{in} \ \cC.
  \]
\end{theorem}


\subsubsection{The dynamically screened interaction operator $W$} 
\label{subsec:W}

As the name indicates, the two key operators in the GW method are on the one hand, the time-ordered Green's function $G$, and on the other hand, the so-called \textit{dynamically screened interaction} operator $W$. The latter operator is defined as
\begin{equation} 
  \label{eq:W2}
  W(\tau) = \vc \delta_0(\tau) + \vc^{1/2} \chi_\sym(\tau) \vc^{1/2},
\end{equation} 
where $\vc$ is the Coulomb operator introduced in Lemma \ref{lem:vc}. It is convenient to split $W$ into a local-in-time exchange contribution $\vc \delta_0(\tau)$ (although this is not obvious at this stage, \eqref{eq:Sigma_x} below shows that $\vc \delta_0(\tau)$ can be interpreted as an exchange term), and a nonlocal-in-time correlation contribution:
\begin{equation}
\label{eq:def_W}
W(\tau) = \vc \delta_0(\tau) + W_\rc(\tau) \quad \text{with} \quad W_\rc(\tau) := \vc \chi(\tau)  \vc = \vc^{1/2} \chi_{\rm sym}(\tau) \vc^{1/2}.
\end{equation}
The properties of the operator $W_c(\tau) \in \cB(\cC,\cC')$ therefore readily follow from the properties of the operators $\vc^{1/2}$ and $\chi_{\rm sym}(\tau)$ established in Lemma~\ref{lem:vc} and Proposition~\ref{lem:chi^pm}.

\subsection{The self-energy operator $\Sigma$} 
\label{subsec:SE} 

We give in this section the definition of the self-energy operator $\Sigma$ using the Dyson equation (see~\eqref{eq:exact_Sigma} below). Let us emphasize that, while the Dyson equation provides a definition of~$\Sigma$ in terms of Green's functions, numerical methods work the other way round: an approximation of the Green's function~$G$ is obtained from the Dyson equation (\ref{eq:exact_Sigma}), using an approximation of the self-energy operator~$\Sigma$. This approach is made precise in Section \ref{sec:GW_finite}.

\subsubsection{The non-interacting Hamiltonian $H_0$ and associated Green's function $G_{0}$} \label{subsec:G0} 

The self-energy operator is defined as the difference between the inverse of the exact Green's function~$G$ and the inverse of some reference Green's function~$G_0$. The reference Green's function is the resolvent of a mean-field non-interacting Hamiltonian. There are several possible choices for this operator, discussed in Remark~\ref{rem:KS} below. In order to remain as general as possible, we introduce a one-body operator~$h_1$ acting on~$\cH_1$, with domain~$H^2(\R^3)$, real-valued (in the sense that $h\psi$ is real-valued whenever $\psi$ is real-valued), and such that $\sigma_\ess(h_1) = [0, \infty)$. The corresponding effective non-interacting $N$-body Hamiltonian is defined on~$\cH_N$ by
\[
H_{0,N} = \sum_{i=1}^N h_1(\br_i).
\]
We define 
$$
\varepsilon_{k} := \inf_{V_k \subset {\cal V}_k}  \sup_{v \in V_k \setminus \left\{0\right\}} \frac{\langle v | h_1 | v \rangle}{\langle v|v\rangle},
$$
where ${\cal V}_k$ is the set of the subspaces of $H^1(\R^3)$ of dimension $k$. Recall that $\varepsilon_k \le 0$ and that if $\varepsilon_k < 0$, then $h_1$ has at least $k$ negative eigenvalues (counting multiplicities) and $\varepsilon_k$ is the $k^{\rm th}$ smallest eigenvalue of $h_1$ (still counting multiplicities). We make the following assumption in the sequel.
\begin{equation*}
\boxed{ \text{\textbf{Hyp. 3}: The one-body Hamiltonian $h_1$ has at least $N$ negative eigenvalues, and $\varepsilon_N < \varepsilon_{N+1}$}.}\\
\end{equation*}
This assumption implies that there is a gap between the $N^{\rm th}$ eigenvalue and the $(N+1)^{\rm st}$ eigenvalue (or the bottom of the essential spectrum if $h_1$ has only $N$ non-positive eigenvalues). 

Let us denote by $(\phi_1,\cdots,\phi_N)$ an orthonormal family of eigenvectors of $h_1$ associated with the eigenvalues $\varepsilon_1, \cdots, \varepsilon_N$. Without loss of generality, we can assume that the $\phi_k$'s are real-valued. The ground state energy of $H_{0,N}$ is $E_{0,N}^0 =  \varepsilon_1 + \ldots + \varepsilon_N$. The condition $\varepsilon_{N} < \varepsilon_{N+1}$ ensures that $E_{0,N}^0$ is a non-degenerate eigenvalue of~$H_{0,N}$ and that the normalized ground state $\Phi_N^0 = \phi_1 \wedge \cdots \wedge \phi_N$ of $H_{0,N}$ is unique up to a global phase. We introduce the one-body mean-field density matrix
\begin{equation} \label{eq:gamma0N0}
  \gamma_{0,N}^0 (\br, \br') := \sum_{k=1}^N \phi_k(\br) { \phi_k(\br')}.
\end{equation}
This function can be seen as the kernel of the spectral projector $\1_{(-\infty,\mu_0)}(h_1)$, where $\mu_0$ is  any real number in  the range $(\varepsilon_N, \varepsilon_{N+1})$ (it is an admissible Fermi level for the ground state of the non-interacting effective Hamiltonian $H_{0,N}$). The density of the non-interacting system is denoted by $\rho_{0,N}^0$. Results similar to the ones stated in Proposition~\ref{prop:Psi_N^0} for $\rho_{0,N}^0$, $\gamma_{0,N}^0$, ... hold true. Finally, similarly as in Section \ref{subsec:GF}, we introduce
\[
A_{0,+}^* (f) = a^\dagger(f) | \Phi_N^0 \ket \quad \text{and} \quad A_{0,-}(f) = a(\overline{f}) | \Phi_N^0 \ket.
\]

\begin{definition}[Reference non-interacting Green's functions]
\label{def:G0}
The reference particle, hole and time-ordered non-interacting Green's functions are respectively defined as
\[
G_{0, \rp}(\tau) = - \ri \Theta(\tau) A_{0,+} \re^{ - \ri \tau (H_{0,N+1}- E_{0,N}^0) } A_{0,+}^*,
\qquad
G_{0, \rh}(\tau) = \ri \Theta(-\tau) A_{0,-}^* \re^{ \ri \tau (H_{0,N-1} - E_{0,N}^0)} A_{0,-},
\]
and $G_0(\tau) = G_{0, \rp}(\tau) + G_{0, \rh}(\tau)$.
\end{definition}

Results similar to Propositions~\ref{lemma:forward_GF}, \ref{lemma:backward_GF} and \ref{lemma:time-ordered_GF} hold for these operators, but we do not write them explicitly for the sake of brevity. However, it should be noted that, in the non-interacting case, the Green's functions have simple explicit expressions in terms of $h_1$ (see Section~\ref{sec:proof:lemma:time-ordered_G0F} for the proof).

\begin{proposition} 
  \label{lemma:time-ordered_G0F}
  It holds
  \[
  G_{0, \rp}(\tau) = - \ri \Theta(\tau) \left( \mathds{1}_{\cH_1} - \gamma_{0,N}^0 \right) \re^{-\ri \tau h_1 } \quad \text{and} \quad
  G_{0, \rh}(\tau) =  \ri \Theta(-\tau) \gamma_{0,N}^0 \re^{-\ri \tau h_1 }.
  \]
  In particular, for any $z \in \mathbb{C} \setminus \sigma(h_1)$,
  \begin{equation} \label{eq:G0+h1}
  \widetilde{G_{0, \rp}}(z) = \left( \mathds{1}_{\cH_1} - \gamma_{0,N}^0 \right) (z - h_1)^{-1} \quad \text{and} \quad
  \widetilde{G_{0, \rh}}(z) =  \gamma_{0,N}^0 (z - h_1)^{-1}.
  \end{equation}
  Hence,
  \begin{equation} \label{eq:G0h1}
    \widetilde{G_{0}}(z) = (z - h_1)^{-1}
  \end{equation}
  is the resolvent of the one-body operator $h_1$.
\end{proposition}

\begin{remark}[On the choice of $G_0$] \label{rem:KS}
There are several possible choices for the one-body operator~$h_1$, although this choice is not really properly discussed
in the literature to our knowledge. The first option, which is used in the original derivation of the GW method~\cite{Hedin1965},
consists in choosing
\begin{equation} \label{eq:h1}
h_1 = - \frac12 \Delta + v_\ext + \rho_N^0 \ast | \cdot |^{-1},
\end{equation}
where $\rho_N^0$ is the exact ground state density. Another option (see for instance~\cite[page~112]{Farid99}) is to consider a one-body operator whose associated ground state density is (as close as possible to) the exact ground state density~$\rho_N^0$. The motivation is that, in this case, the self-energy should be smaller. The Kohn-Sham~\cite{KS65} model formally satisfies this requirement. The associated one-body operator reads
\begin{equation} \label{eq:h1KS}
h_1 = -\frac12 \Delta + v_\ext + \rho_N^0 \ast | \cdot |^{-1} + v_{\xc}\left[\rho_N^0\right],
\end{equation}
where $v_\xc$ is the (exact) exchange-correlation potential.
In practice, approximations of $\rho_N^0$ and $v_{\xc}\left[\rho_N^0\right]$ are computed by means of a Kohn-Sham LDA or GGA calculation~\cite{KS65, Perdew1996}. This is believed to provide a sufficiently accurate approximation of the exact ground state density which does not spoil the results subsequently obtained by GW calculations.
\end{remark}
    
\subsubsection{The dynamical Hamiltonian $\widetilde H(z)$}

In view of~\eqref{eq:G0h1}, it is natural to introduce the inverse of the time-ordered Green's function, which will correspond to some dynamical one-body Hamiltonian. More precisely, we would like to define, at least for each $z \in \C \setminus \R$, a one-body operator $\widetilde H(z)$ such that
\[
\widetilde G(z) := \left( z - \widetilde H(z) \right)^{-1}, \quad \text{or equivalently,} \quad  \widetilde H(z) = z - \left( \widetilde G(z) \right)^{-1}.
\]
The following proposition, proved in Section~\ref{sec:proof:lemma:Hz}, shows that such a definition makes sense.

\begin{proposition} \label{lemma:Hz}
Let $z \in \C \setminus \R$. The operator $\widetilde G(z)$ is an invertible operator from $\cH_1$ onto some vector subspace $\widetilde D(z)$ of $\cH_1$. Moreover, $\widetilde D(z)$ is dense in $\cH_1$, $\widetilde D(z) \subset H^2(\R^3)$, and $\widetilde H(z)$ is a well-defined closed operator with domain $\widetilde D(z)$.
\end{proposition}

\begin{remark}
 We do not know whether the equality $\widetilde{D}(z) = H^2(\R^3)$ is true, nor do we know whether $\widetilde{D}(z_1) = \widetilde{D}(z_2)$ for $z_1 \neq z_2$.
 \end{remark}
 
\subsubsection{Definition of the self-energy operator $\Sigma$ from the Dyson equation} \label{subsubsec:SE}

We are now able to define the exact self-energy operator $\widetilde{\Sigma}$ via the Dyson equation. Note that we do not define the self-energy in the time domain, but consider only $\widetilde{\Sigma}(z)$ (as in~\cite[Section~5.1]{Farid99}).

\begin{definition}[Self-energy] \label{def:self_energy}
  The self-energy operator is defined as
  \begin{equation} \label{eq:exact_Sigma}
    \forall z \in \C \setminus \R, \qquad \widetilde{\Sigma}(z) := \widetilde{G_0}(z)^{-1} - \widetilde{G}(z)^{-1} = (z - h_1) - \left(z- \widetilde{H}(z)\right) = \widetilde{H}(z) - h_1,
  \end{equation}
  where $h_1$ is the one-body mean-field Hamiltonian introduced in Section~\ref{subsec:G0}.
\end{definition}

The operator $ \widetilde{\Sigma}(z)$ is the difference between the one-body dynamical Hamiltonian and the reference one-body mean-field Hamiltonian $h_1$. With this writing, $\widetilde\Sigma(z)$ can be seen as the correction term to be added to the reference one-body Hamiltonian in order to obtain the dynamical mean-field one-body Hamiltonian:
\[
\widetilde{H}(z) = h_1+\widetilde{\Sigma}(z).
\]


\section{The GW approximation for finite systems}
\label{sec:GW_finite}

\subsection{G$_0$W$^0$, self-consistent $\GW^0$, self-consistent $\GW$, and all that}

\subsubsection{The GW equations}

We now turn to the GW approximation for finite systems. The purpose of the GW approximation is to estimate the time-ordered Green's function $G$ via the Dyson formula (\ref{eq:exact_Sigma}). Instead of using (\ref{eq:exact_Sigma}) to define the self-energy $\widetilde{\Sigma}(z)$, we use this equation \textit{with some approximation $\widetilde{\Sigma}^{\rm GW}(z)$ of $\widetilde{\Sigma}(z)$} to obtain an approximation $\widetilde{G}^{\rm GW}(z)$ of the time-ordered Green's function via
\begin{equation} \label{eq:Dyson}
  \left(\widetilde{G}^{\rm GW}\right)^{-1}(z)  = z - \left(h_1+\widetilde\Sigma^{\rm GW}(z)\right).
\end{equation}
Using the Dyson equation to {\em define} the time-ordered Green's function is only possible if an alternative expression of the self-energy operator is available. Such an expression was formally obtained by Hedin in 1965 (see~\cite{Hedin1965}). The GW approximation consists in replacing the so-called vertex function in Hedin's equations by a tensor product of Dirac masses.~\\

The original GW equations were derived on the time domain and on the frequency domain. However, as noticed several times in Section~\ref{sec:operators}, the operators involved in the GW equations are not smooth on these axes. It turns out that it is formally possible to recast the equations on some imaginary axis using Theorem~\ref{thm:analytic_continuation_kernel_product}. This approach, first introduced by Rojas, Godby and Needs \cite{Rojas1995} (see also \cite{RSWRG99}), is now known under the name of the ``analytic continuation method''. For reasons that we will explain throughout this section, these equations are recast as follows within our mathematical framework.\\

\begin{definition}[$\GW$ equations on the imaginary axis of the frequency domain]~\\
Find $\widetilde{G^\GW}(\mu + \ri \cdot)  \in L^2(\R_\omega, \cB(\cH))$ solution of the system
\begin{subequations}
\label{eq:newHedin}
\begin{align}
  \widetilde{P_\sym^\GW}(\ri \omega)  &= \dfrac{1}{2 \pi} \vc^{1/2} \left( \int_{-\infty}^{\infty}  \widetilde{G^\GW} \big(\mu + \ri (\omega+\omega')\big) \odot \widetilde{G^\GW}(\mu + \ri \omega') \, \rd \omega' \right)\vc^{1/2}, \label{eq:GWa} \\
  \widetilde{\chi^{\GW}_\sym}(\ri \omega) & = \left( \mathds{1}_\cH - \widetilde{P_\sym^\GW}(\ri \omega) \right)^{-1} - \mathds{1}_{\cH_1}, \label{eq:GWb} \\
  \widetilde{W_c^\GW}(\ri \omega)   &= \vc^{1/2} \widetilde{\chi^{\GW}_\sym}(\ri \omega) \vc^{1/2}, \label{eq:GWc} \\
  \widetilde{\Sigma^\GW}(\mu + \ri \omega)  &= K_x -\dfrac{1}{2 \pi} \int_{-\infty}^{\infty} \widetilde{G^\GW}\big(\mu + \ri (\omega-\omega')\big) \odot \widetilde{W_c^\GW}(\ri \omega') \, \rd \omega', \label{eq:GWd} \\
  \widetilde{G^\GW}(\mu + \ri\omega) &= \left[\mu + \ri\omega-\left(h_1+\widetilde{\Sigma^\GW}(\mu + \ri \omega)\right)\right]^{-1}, \label{eq:GWe}
\end{align}
\end{subequations}
where $h_1$ is the one-body operator defined in~\eqref{eq:h1} and where $K_x$ is the integral operator on $\cH_1$ with kernel
$$
K_x(\br,\br'):= - \dfrac{\gamma_{0,N}^0(\br, \br')}{| \br - \br'|},
$$
where $\gamma_{0,N}^0$ was defined in~\eqref{eq:gamma0N0}.
\end{definition}

\begin{remark}
	In the $\GW$ equations~\eqref{eq:newHedin}, the chemical potential $\mu$ is supposed to be known~\textit{a priori}. 
\end{remark}

The GW equations~\eqref{eq:newHedin} would be the natural equations to work with from a mathematical viewpoint (they are formally equivalent to the original Hedin's GW equations). However, we were not able to study~\eqref{eq:newHedin} for reasons detailed in Remark~\ref{pb:pb_P} below.

As one can directly see, the equations involve quite a large number of operators, which all have a physical significance. The operator $\widetilde{P_\sym^{\GW}}$ is the $\GW$ approximation of the symmetric irreducible polarizability operator, the operator $\widetilde{\chi_\sym^\GW}$ is the $\GW$ approximation of the symmetric reducible polarizability operator, the operator $\widetilde{W^\GW}$ is the $\GW$ approximation of the dynamically screened Coulomb interaction operator, and finally $\widetilde{\Sigma^\GW}$ is the GW approximation of the self-energy operator. We recognize in Equation~\eqref{eq:GWe} the Dyson equation. The name ``GW'' comes from Equation~\eqref{eq:GWd}.

\subsubsection{Different levels of GW approximation}
\label{ssec:levels}

As mentioned below (see Remark~\ref{pb:pb_P}), we were not able to study the full self-consistent problem~\eqref{eq:newHedin}. We will therefore restrict ourselves to the so-called G$_0$W$^0$ and $\GW^0$ approximations. We explain in this section how these different models are obtained. \\

(i) In the fully self-consistent GW (sc-GW) approximation, we assume that the full problem~\eqref{eq:newHedin} is well-posed, so that there exists a (unique) solution $\widetilde{G^\GW}$. It is then solved self-consistently: the idea is to start from some trial Green's function, and keep updating it with~\eqref{eq:newHedin} until convergence. This method is for instance used in~\cite{CRRRS12,CRRRS13, KFS14, RJT10, SDL06}. It was implemented only quite recently due to its high numerical cost (one needs to perform the inversion in~\eqref{eq:GWb} at each iteration).\\

(ii) In the so-called self-consistent $\GW^0$ approximation, or simply $\GW^0$ approximation, only the Green's function (and not the screened Coulomb operator) is updated in~\eqref{eq:GWd} (see for instance~\cite{SDL09, VonBarth1996}). This partial update not only speeds up the calculation (the inversion in~\eqref{eq:GWb} is only performed once), but is sometimes in better agreement with experimental results than the sc-GW approximation. This is the model that we study in Section~\ref{sec:GW0}. \\

(iii) Finally, most works simply consider the G$_0$W$^0$ approximation, where only one iteration of the sc-GW (or equivalently one iteration of $\GW^0$) is performed. This model is very popular due to its relatively low computational cost, and provides already very satisfactory results (see for instance~\cite{Blase2011}). \\

Let us also emphasize that it is unclear that a solution of the fully self-consistent GW model is a better approximation in any sense to the exact Green's function than a non self-consistent approximation such as the one obtained by the G$_0$W$^0$ approximation. This is discussed in~\cite[Section~9.8]{Farid99}, where the author also comments on the possibilities to update the effective one-body operator $h_1+K_x$ along the iterations. \\

\begin{remark} \label{pb:pb_P}
We do not know how to give a proper mathematical meaning to Equation~\eqref{eq:GWa}. More specifically, one would like to define, for a reasonable choice of Green's function $\widetilde{G^\app}$, the operator
\[
	\forall \omega \in \R_\omega, \quad \widetilde{P_\sym}[G^\app](\ri \omega) := \dfrac{1}{2 \pi} \vc^{1/2} \left( \int_{-\infty}^{\infty}  \widetilde{G^\app} \big(\mu + \ri (\omega+\omega')\big) \odot \widetilde{G^\app}(\mu + \ri \omega') \, \rd \omega' \right)\vc^{1/2},
\]
and we would like this operator to be a self-adjoint bounded negative operator on $\cH_1$. It is the case for instance when $\widetilde{G^\app}$ is the non-interacting Hamiltonian $\widetilde{G_0}$ defined in~\eqref{eq:G0h1} (see Proposition~\ref{prop:propertiesP0} and Remark~\ref{rem:transformation_contour_deformation}), or when $\widetilde{G^\app}$ is the exact Green's function defined in~\eqref{eq:widetilde_G} (this fact can be proved by adapting the arguments given in Section~\ref{ssec:analytical_continuation}). We were not able to obtain this result for a generic class of approximate Green's functions $\widetilde{G^\app}$, say $\widetilde{G^\app}$ of the form~\eqref{eq:GWe} with $\widetilde{\Sigma^\GW}(\mu + \ri \omega)$ in a small ball of $L^\infty(\R_\omega, \cB(\cH_1))$.
\end{remark}
For this reason, we will not study the self-consistent GW equation~\eqref{eq:newHedin}.

\subsection{The operator $\widetilde{W^0}$ and the random phase approximation}
\label{sec:W0}

The remainder of this section is devoted to the study of the $\GW^0$ approximation (which includes the G$_0$W$^0$ approximation), which amounts to study the two equations~\eqref{eq:GWd}-\eqref{eq:GWe} with a specific fixed choice of the screening operator $W^0$. This approximation bypasses the difficulties mentioned in Remark~\ref{pb:pb_P}. In order to present and study the $\GW^0$ approximation, one must first define the operator $W^0$.

\subsubsection{The RPA irreducible polarizability operator $\widetilde{P^0}$}
\label{sec:irreducible_polarizability}

The GW approximation of the \textit{irreducible polarizability operator} $P$ is formally defined as
\begin{equation} \label{eq:phys_P}
P^{\rm GW}(\br, \br', \tau) = - \ri G(\br, \br', \tau) G(\br', \br, -\tau).
\end{equation}
When the Green's function $G$ is the non-interacting one $G_0$ defined in~\eqref{def:G0}, this also corresponds to the so-called random phase approximation of the reducible polarizability operator (compare for instance~\eqref{eq:explicit_tildeP0} with the expression in~\cite{Cances2012}). We therefore define
\[
	P^0(\br, \br', \tau) := - \ri G_0(\br, \br', \tau) G_0(\br', \br, - \tau).
\]
This operator is expected to have properties similar to the operator $\chi$ defined in Section~\ref{ssec:chi}. In particular, $P^0(\tau)$ is expected to be a bounded operator from $\cC'$ to $\cC$. It is therefore more convenient to work with its symmetrized counterpart $P^0_\sym(\tau) := \vc^{1/2} P^0(\tau) \vc^{1/2}$, which is expected to be a bounded operator on $\cH_1$. It is possible to decompose $P^0_\sym$ as $P^0_\sym = P^{0,+}_\sym + P^{0,-}_\sym$ where, using the kernel-product $\odot$ defined in Section \ref{ssection:kernel_multiplication}, and the explicit expressions of $G_{0, \rp}$ and $G_{0, \rh}$ given in Proposition~\ref{lemma:time-ordered_G0F},
\begin{align} \label{eq:def_P+}
	P^{0,+}_\sym(\tau) 
		& = - \ri \Theta(\tau) \vc^{1/2} G_{0,\rp}(\tau) \odot G_{0, \rh}(- \tau) \vc^{1/2} \\
		& = - \ri \Theta(\tau) \vc^{1/2} \left( \left( \mathds{1}_{\cH_1} - \gamma_{0,N}^0 \right)\re^{ - \ri \tau h_1} \odot \gamma_{0,N}^0 \re^{\ri \tau h_1} \right) \vc^{1/2}
\end{align}
and
\begin{align*}
	P^{0,-}_\sym(\tau) & = - \ri \Theta(- \tau) \vc^{1/2} G_{0, \rh}(\tau) \odot G_{0, \rp}(-\tau) \vc^{1/2} \\
	& = - \ri \Theta(- \tau) \vc^{1/2} \left(  \gamma_{0,N}^0 \re^{-\ri \tau h_1} \odot\left( \mathds{1}_{\cH_1} - \gamma_{0,N}^0 \right) \re^{\ri \tau h_1} \right) \vc^{1/2}.
\end{align*}
Actually, with this definition, we were not able to give a meaning to $P_\sym^{0,-}$ (it may not be a bounded operator on $\cH_1$). We therefore prefer to use the modified kernel-product~$\widetilde{ \odot}$ defined in Remark~\ref{rem:odot}. Our correct mathematical definition for $P^{0,-}_\sym$ then is
\begin{align} \label{eq:def_P-}
	P^{0,-}_\sym(\tau) & = - \ri \Theta(- \tau) \vc^{1/2} G_{0, \rh}(\tau) \, \widetilde{\odot} \, G_{0, \rp}(-\tau) \vc^{1/2}  \\
	& = - \ri \Theta(- \tau) \vc^{1/2} \left(  \gamma_{0,N}^0 \re^{ -\ri \tau h_1} \widetilde{\odot} \left( \mathds{1}_{\cH_1} - \gamma_{0,N}^0 \right) \re^{\ri \tau h_1} \right) \vc^{1/2}.
\end{align}
As will be shown in Lemma~\ref{lem:P0}, this amounts to defining $P^{0,-}(\tau) = P^{0,+}(-\tau)$. We recall that $\gamma_{0, N}^0$ is the orthogonal projector on the vector space spanned by the eigenvectors of $h_1$ associated with the lowest $N$ eigenvalues (see~\eqref{eq:gamma0N0}), so that
\begin{equation} \label{eq:gamma0_projector}
	\gamma_{0,N}^0 = \sum_{k=1}^N | \phi_k \ket \bra \phi_k |,
\end{equation}
where $h_1 \phi_k = \varepsilon_k \phi_k$, and the eigenfunctions $\phi_k$ are real-valued and orthonormal. The following result shows that our definitions make sense, and gives explicit formulae for $P^{0, +}$ (see Section~\ref{proof:lemP0} for the proof).

\begin{lemma} \label{lem:P0}
	The family $\left(P^{0, +}_\sym(\tau) \right)_{\tau \in \R_\tau}$ defined by~\eqref{eq:def_P+} is a bounded causal operator on $\cH_1$, while $\left(P^{0, -}_\sym(\tau) \right)_{\tau \in \R_\tau}$ defined by~\eqref{eq:def_P-} is a bounded anti-causal operator on $\cH_1$. It holds $P^{0,-}_\sym(\tau) = P^{0, +}_\sym (-\tau)$ and
	\begin{equation} \label{eq:explicitP0}
		P^{0, +}_\sym(\tau) = - \ri \Theta(\tau) \sum_{k=1}^N \vc^{1/2} \phi_k \left( \mathds{1}_{\cH_1} - \gamma_{0,N}^0 \right) \re^{-\ri \tau (h_1 - \varepsilon_k)} \left( \mathds{1}_{\cH_1} - \gamma_{0,N}^0 \right) \phi_k \vc^{1/2}.
	\end{equation}
\end{lemma}

\begin{remark} \label{rem:phik}
	For $1 \le k \le N$, the notation $\phi_k$ in~\eqref{eq:explicitP0} refers to the multiplication operator by the function $\phi_k$. It is a bounded operator from $\cC'$ to $\cH_1$, and from $\cH_1$ to $\cC$ (see the proof of Lemma~\ref{lem:P0}). The operator $\phi_k v_c^{1/2}$ is bounded on $\cH_1$, and one can check that its adjoint on~$\cH_1$ is~$(\phi_k v_c^{1/2})^* := v_c^{1/2} \phi_k$.
\end{remark}

The properties of the Laplace and Fourier transforms of $P^{0, +}_\sym$ are easily deduced from~\eqref{eq:explicitP0} using Proposition~\ref{prop:resolvent} and Lemma~\ref{lemma:Im_positive_particle}.

\begin{proposition}
  \label{prop:ppties_P_GW}
  The function $z \mapsto \widetilde{P_\sym^{0,+}}(z)$ is analytic on the upper half-plane~$\mathbb{U}$, and can be analytically continued to the lower half-plane~$\mathbb{L}$ through the semi-real line $(-\infty, \varepsilon_{N+1} - \varepsilon_N)$. For all $z \in \C \setminus [\varepsilon_{N+1} - \varepsilon_N, \infty)$,
  \begin{equation} \label{eq:explicit_tildeP+}
  	\widetilde{P^{0,+}_\sym}(z) = \sum_{k=1}^N \vc^{1/2} \phi_k \left( \dfrac{ \mathds{1}_{\cH_1} - \gamma_{0,N}^0 }{z - h_1 + \varepsilon_k} \right) \phi_k \vc^{1/2}.
  \end{equation}
  Moreover $\widetilde{P^{0,+}_\sym}( \cdot + \ri \eta)$ converges to $\widehat{P^{0,+}_\sym}$ in $H^{-1}(\R_\omega, \cB(\cH_1))$ as $\eta \to 0^+$, with
  \[
  	\Re \widehat{P^{0,+}_\sym} = \pv \left( \sum_{k=1}^N \vc^{1/2} \phi_k \left( \dfrac{ \mathds{1}_{\cH_1} - \gamma_{0,N}^0 }{\cdot - h_1 + \varepsilon_k} \right) \phi_k \vc^{1/2} \right)
\]
and
\[
	\Im \widehat{P^{0,+}_\sym} = - \pi \left( \sum_{k=1}^N \vc^{1/2} \phi_k \left( \mathds{1}_{\cH_1} - \gamma_{0,N}^0\right) P^{h_1 - \varepsilon_k} \phi_k \vc^{1/2} \right).
  \]
  It also holds 
  \[
  	\forall z \in \C \setminus [\varepsilon_{N+1} - \varepsilon_N, \infty), \quad \widetilde{P^{0,-}_\sym}(z) = \widetilde{P^{0,+}_\sym}(-z),
  \]
  so that, for $z \in \UU \cup \LL \cup (- (\varepsilon_{N+1} - \varepsilon_N), \varepsilon_{N+1} - \varepsilon_N)$,
  \begin{equation} \label{eq:explicit_tildeP0}
  	\widetilde{P^{0}_\sym}(z) = - 2\sum_{k=1}^N v_c^{1/2} \phi_k (\mathds{1}_{\cH_1} - \gamma_{0,N}^0) \left( \dfrac{h_1 - \varepsilon_k }{(h_1 - \varepsilon_k)^2 - z^2} \right) (\mathds{1}_{\cH_1} - \gamma_{0,N}^0) \phi_k \vc^{1/2}.
  \end{equation}
\end{proposition}

The properties of $\widehat{P^{0,+}_\sym}$ and of $\widehat{P^{0,-}_\sym}$ can be directly read off from the previous expressions. For instance, we see that $\Im \widehat{P^{0,+}_\sym}$ and $\Im \widehat{P^{0,-}_\sym}$ are negative operator-valued measures, with support in $(\varepsilon_{N+1} - \varepsilon_N, \infty)$ and $(-\infty, -(\varepsilon_{N+1} - \varepsilon_N))$ respectively. For $\omega$ in the real gap $(- (\varepsilon_{N+1} - \varepsilon_N), \varepsilon_{N+1} - \varepsilon_N)$, we see that $\widehat{P^{0, \pm}_\sym}(\omega) = \Re \widehat{P^{0, \pm}_\sym}(\omega)$ is a negative bounded self-adjoint on $\cH_1$.\\
 For our purpose, we only need to know the behavior of $\widetilde{P^0_\sym}$ on the imaginary axis $\ri \R_\omega$. We summarize the corresponding most important results in the following proposition (see Section~\ref{proof:propertiesP0} for the proof).

\begin{proposition} \label{prop:propertiesP0}
	It holds
	\begin{equation} \label{eq:explicit_widetildeP}
		\forall \omega \in \R_\omega, \quad \widetilde{P_\sym^0}(\ri \omega) 
		= - 2 \sum_{k=1}^N \vc^{1/2} \phi_k \left( \mathds{1}_{\cH_1} - \gamma_{0,N}^0 \right) \left( \dfrac{ h_1 - \varepsilon_k}{\omega^2 + (h_1 - \varepsilon_k)^2} \right) \left( \mathds{1}_{\cH_1} - \gamma_{0,N}^0 \right) \phi_k \vc^{1/2}.
	\end{equation}
	In particular, for all $\omega \in \R_\omega$, the operator $\widetilde{P^0_\sym}(\ri \omega)$ is a negative, self-adjoint bounded operator on $\cH_1$ satisfying $\widetilde{P_\sym^0}(-\ri \omega) = \widetilde{P_\sym^0}(\ri \omega)$. In addition, the function $\omega \mapsto \widetilde{P^0_\sym}(\ri \omega)$ is analytic from $\R_\omega$ to $\cS(\cH_1)$, and is in $L^p(\R_\omega, \cS(\cH_1))$ for all $p > 1$. For any $f \in \cH_1$, the function $\omega \mapsto \left\bra f \Big| \widetilde{P^0_\sym}(\ri \omega) \Big| f \right\ket$ is non-positive, in $L^1(\R_\omega)$, and
	\begin{align} \label{eq:int_fP0f}
		\int_{-\infty}^{+\infty}  \left\bra f \Big| \widetilde{P^0_\sym}(\ri \omega) \Big| f \right\ket \rd \omega 
		& =- 2 \pi \left\bra f \Big| \vc^{1/2} \left( \left( \mathds{1}_{\cH_1} - \gamma_{0,N}^0 \right) \odot \gamma_{0,N}^0 \right) \vc^{1/2} \Big| f \right\ket \nonumber \\
		& = - 2 \pi \left\bra f \left| \sum_{k=1}^N \vc^{1/2} \phi_k \left( \mathds{1}_{\cH_1} - \gamma_{0,N}^0 \right) \phi_k \vc^{1/2} \right| f \right\ket .
	\end{align} 
	Finally, there exists a constant $C \in \R^+$ such that
	\begin{equation} \label{eq:ineq_P0}
		\forall \omega \in \R_\omega, \quad 0 \le - \widetilde{P^0_\sym}(\ri \omega) \le \dfrac{C}{(\omega^2 + 1)^{1/2}} \left( \vc^{1/2} \rho_{0,N}^0 \vc^{1/2} \right),
	\end{equation}
	where $\rho_{0,N}^0$ is the multiplication operator by the (real-valued) function $\rho_{0,N}^0$, the latter operator being bounded from $\cC'$ to $\cC$.
\end{proposition}

\paragraph{The sum-rule for the operator $\widetilde{P^0}$.}
We end this section with the sum-rule for the operator $\widetilde{P^0} = \vc^{-1/2} \widetilde{P_\sym^0} \vc^{-1/2}$, which goes from $\cC'$ to $\cC$. We postpone the proof until Section~\ref{proof:sumrule_P0}.
\begin{theorem} \label{th:sumrule_P0}
	The operator $2 \sum_{k=1}^N \phi_k (\mathds{1}_{\cH_1} - \gamma_{0,N}^0) (h_1 - \varepsilon_k) \phi_k$ is bounded from $\cC'$ to $\cC$, and it holds
	\[
		2  \sum_{k=1}^N \phi_k (\mathds{1}_{\cH_1} - \gamma_{0,N}^0) (h_1 - \varepsilon_k) \phi_k = \div( \rho_{0,N}^0 \nabla \cdot).
	\]
	Moreover, the following weak-convergence holds:
	\[
		\forall (f,g) \in \cC' \times \cC', \quad \lim_{\omega \to \pm \infty} \left\bra \overline{f}, - \omega^2 \widetilde{P^0}(\ri \omega) g \right\ket_{\cC', \cC} = \left\bra \overline{f}, - \div \left( \rho_{0,N}^0 \nabla g \right) \right\ket_{\cC', \cC} = \int_{\R^3} \rho_{0, N}^0 \overline{\nabla f} \cdot \nabla g.
	\]
	Finally, for all $g \in \cC'$ such that $\Delta g \in L^2(\R^3)$, the following strong convergence holds:
	\[
		\lim_{\omega \to \pm \infty} \omega^2 \widetilde{P^0}(\ri \omega) g = \div \left( \rho_{0,N}^0 \nabla g \right) \quad \text{in } \cC.
	\]
\end{theorem}

This sum-rule automatically leads to a sum-rule for the reducible polarizability operator in the random phase approximation $\chi^0$ (see Theorem~\ref{th:sumrule_chi0}).


\subsubsection{The analytical continuation method}
\label{ssec:analytical_continuation}

In this section, we explain why~\eqref{eq:GWa} can be thought of as a natural reformulation of the usual physical definition~\eqref{eq:phys_P}, and why problems arise with Definition~\eqref{eq:GWa} (see Problem~\ref{pb:pb_P}). This section also serves as a guideline to understand why~\eqref{eq:GWd} is a natural reformulation of the usual physical definition of $\Sigma^{\GW}$ (see~\eqref{eq:phys_Sigma} below). In the previous section, we gave the properties of~$\widetilde{P^0}$ using the explicit expression of $P^0$ given in~\eqref{eq:explicitP0}. While this approach simplifies the proofs, it somehow hides some structural properties that we highlight in this section.\\

Recall that $P_\sym^0 = P_\sym^{0,+} + P_\sym^{0,-}$ with
\[
	P_\sym^{0,+} (\tau) = - \ri \Theta(\tau) \vc^{1/2} G_{0,\rp}(\tau) \odot G_{0, \rh}(- \tau) \vc^{1/2}
\]
and
\[
	P_\sym^{0,-}(\tau) = - \ri \Theta(-\tau) \vc^{1/2} G_{0,\rh}(\tau) \ \widetilde{\odot} \ G_{0, \rp}(- \tau) \vc^{1/2},
\]
where
\[
G_{0, \rp}(\tau) = - \ri \Theta(\tau) A_{0,+} \re^{ - \ri \tau (H_{0,N+1}- E_{0,N}^0) } A_{0,+}^*,
\qquad
G_{0, \rh}(\tau) = \ri \Theta(-\tau) A_{0,-}^* \re^{ \ri \tau (H_{0,N-1} - E_{0,N}^0)} A_{0,-}.
\]
The idea is to use the results of Theorem~\ref{thm:analytic_continuation_kernel_product}. We first consider $P_\sym^{0,+}$, and prove that the hypotheses of Theorem~\ref{thm:analytic_continuation_kernel_product} are satisfied. This is given by the following lemma.
\begin{lemma}
	There exists a constant $C \in \R^+$ such that, for any $f \in \cH_1$, it holds $A_{0,-} \left( \vc^{1/2} f \right) \in \fS_2(\cH_1)$ with
	\[
		\left\| A_{0,-} \left( \vc^{1/2} f \right) \right\|_{\fS_2(\cH_1)} \le C \| f \|_{\cH_1}.
	\]
	Moreover, $H_{0,N+1}- E_{0,N}^0 \ge \varepsilon_{N+1}$ and $H_{0,N-1} - E_{0,N}^0 \ge - \varepsilon_{N}$.
\end{lemma}
\begin{proof}
The first point comes from the fact that $A_{0,-}^\ast A_{0,-} = \gamma_{0,N}^0$ and that $v_c^{1/2} f \in \cC' \hookrightarrow L^6$ whenever $f \in \cH_1$, together with Lemma~\ref{lem:Avcf}.  
\end{proof}
In particular, the hypotheses of Theorem~\ref{thm:analytic_continuation_kernel_product} are satisfied, and we deduce that for any $\nu' > \varepsilon_{N}$ and $\nu + \nu'< \varepsilon_{N+1}$,
	\begin{equation} \label{eq:Gp_odot_Gh}
		\forall \omega \in \R, \quad \widetilde{P_\sym^{0,+}}(\nu + \ri \omega) = \dfrac{1}{2 \pi} \int_{-\infty}^{+\infty} \vc^{1/2} \left( \widetilde{G_{0, \rp}} \big( \nu + \nu' + \ri (\omega + \omega') \big) \odot \widetilde{G_{0, \rh}}(\nu' + \ri \omega') \right) \vc^{1/2} \rd \omega'.
	\end{equation}
We treat $\widetilde{P_\sym^{0,-}}$ is a similar way, and find that for any $\nu' < \varepsilon_{N+1}$ and $\nu + \nu' > \varepsilon_N$,
\begin{equation} \label{eq:Gh_odot_Gp}
	\forall \omega \in \R, \quad\widetilde{P_\sym^{0,-}}(\nu + \ri \omega) 
		 = \dfrac{1}{2 \pi} \int_{-\infty}^{+\infty} \vc^{1/2} \left(\widetilde{G_{0, \rh}} \big( \nu + \nu' + \ri (\omega + \omega') \big) \ \widetilde{\odot} \ \widetilde{G_{0, \rp}}(\nu' + \ri \omega') \right) \vc^{1/2} \rd \omega'.
\end{equation}
Actually, the kernel-product $\widetilde{\odot}$ in the latter expression can be transformed into the kernel-product~$\odot$, thanks to the following lemma, whose proof is given in Section~\ref{proof:widetilde_odot_odot}.
\begin{lemma} \label{lem:widetilde_odot_odot}
		For any $\nu' < \varepsilon_{N+1}$, any $\nu + \nu' > \varepsilon_N$ and any $\omega, \omega' \in \R_\omega$, 
		\[
		\widetilde{G_{0, \rh}} \big( \nu + \nu' + \ri (\omega + \omega') \big) \ \widetilde{\odot} \ \widetilde{G_{0, \rp}}(\nu' + \ri \omega')
		=
		\widetilde{G_{0, \rh}}\big( \nu + \nu' + \ri (\omega + \omega') \big) \odot \widetilde{G_{0, \rp}}(\nu' + \ri \omega'),
		\]
		as bounded operators from $\cC'$ to $\cC$.
\end{lemma}

We can perform the same type of calculation for $G_\rh \odot G_\rh$. Following the proof of Theorem~\ref{thm:analytic_continuation_kernel_product}, we deduce from $G_\rh(\tau) \odot G_\rh(-\tau) = 0$ that, for any $\nu' > \varepsilon_N$ and $\nu + \nu' > \varepsilon_N$,
\begin{equation} \label{eq:Gh_odot_Gh}
	\forall \omega \in \R_\omega, \quad 
	\dfrac{1}{2 \pi} \int_{-\infty}^{+\infty} \vc^{1/2} \left( \widetilde{G_{0, \rh}} \big( \nu + \nu' + \ri (\omega + \omega') \big) \odot \widetilde{G_{0, \rh}}(\nu' + \ri \omega') \right) \vc^{1/2} \rd \omega' = 0.
\end{equation}
Similarly, from $G_\rp (\tau) \odot G_\rp(-\tau) = 0$, we deduce that, \textit{at least formally}, for any $\nu' < \varepsilon_{N+1}$, and any $\nu + \nu'< \varepsilon_{N+1} $,
\begin{equation} \label{eq:formal_Gp_odot_Gp}
	\forall \omega \in \R_\omega, \quad 
	\dfrac{1}{2 \pi} \int_{-\infty}^{+\infty} \vc^{1/2} \left( \widetilde{G_{0, \rp}} \big( \nu + \nu' + \ri (\omega + \omega') \big) \odot \widetilde{G_{0, \rp}}(\nu' + \ri \omega') \right) \vc^{1/2} \rd \omega' = 0.
\end{equation}
\begin{remark}
The last equality is formal, in the sense that the integrand $\widetilde{G_{0, \rp}} \odot \widetilde{G_{0, \rp}}$ is actually not well-defined: it does not define a bounded operator from $\cC'$ to $\cC$. However, we can proceed as follows. For $\omega \in \R_\omega$, let $\widetilde{P^{+,+}_\sharp}(\ri \omega)$ be the operator defined on the core $\cH_1 \cap \cC$ by
\begin{align*}
	\forall f, g \in \cH_1 \cap \cC, \quad & 
	\left\bra f \left| \widetilde{P^{+,+}_\sharp}(\ri \omega)  \right| g\right\ket  \\
	& \quad 
	:= \dfrac{1}{2 \pi}\int_{-\infty}^{+\infty} \Tr_{\cH_1} \left[ \widetilde{G_{0, \rp}}  \big( \nu + \nu' + \ri (\omega + \omega') \big) \left( \vc^{1/2} g\right)  \widetilde{G_{0, \rp}}(\nu' + \ri \omega') \left( \vc^{1/2} \overline{f} \right) \right] \rd \omega'.
\end{align*}
Noticing that $\vc^{1/2} \overline{f}$ and $\vc^{1/2} g$ are in $\cH_1$ since $f,g \in \cC$, and reasoning as in the proof of Lemma~\ref{lem:widetilde_odot_odot}, we can prove that the operator in the trace is indeed trace-class, with
\begin{align*}
	& \left| \Tr_{\cH_1} \left[ \widetilde{G_{0, \rp}}  \big( \nu + \nu' + \ri (\omega + \omega') \big) \left( \vc^{1/2} g\right)  \widetilde{G_{0, \rp}}(\nu' + \ri \omega') \left( \vc^{1/2} \overline{f} \right) \right] \right|
	 \le p_\omega(\omega') \| f \|_{\cC} \| g \|_{\cC},
\end{align*}
where $p_\omega$ is an integrable function independent of $f$ and $g$. Moreover, following the proof of Theorem~\ref{thm:analytic_continuation_kernel_product}, we can prove that, as expected,
\[
	\forall f, g \in \cH_1 \cap \cC, \quad \left\bra f \left| \widetilde{P^{+,+}_\sharp}(\ri \omega)  \right| g\right\ket = 0.
\]
The unique continuation on $\cH_1$ of $\widetilde{P^{+,+}_\sharp}(\ri \omega)$ therefore is the null operator. It is unclear to us how to extend a similar reasoning for a generic class of approximated Green's function $\widetilde{G^\app}$.

\end{remark}

By gathering~\eqref{eq:Gp_odot_Gh}, \eqref{eq:Gh_odot_Gp}, \eqref{eq:Gh_odot_Gh} and \eqref{eq:formal_Gp_odot_Gp}, we find that, for any $\nu' \in (\varepsilon_N, \varepsilon_{N+1})$ and $\nu + \nu' \in (\varepsilon_N , \varepsilon_{N+1} )$, 
\[
	\forall \omega \in \R_\omega, \quad \widetilde{P^0_\sym} (\nu + \ri \omega) = \dfrac{1}{2 \pi} \int_{-\infty}^{+\infty} \vc^{1/2} \left( \widetilde{G_0} \big( \nu + \nu' + \ri (\omega' + \omega) \big) \odot \widetilde{G_0}(\nu' + \ri \omega') \right) \vc^{1/2} \rd \omega'.
\]
In particular, this equality holds for the particular choice $\nu' = \mu_0$ and $\nu = 0$.

\begin{remark} \label{rem:transformation_contour_deformation}
To summarize the work performed in this section, we transformed the equation
\begin{equation} \label{eq:P0_before}
		P^0(\br, \br', \tau) := - \ri G_0(\br, \br', \tau) G_0(\br', \br, - \tau)
\end{equation}
into: for any $\nu' \in (\varepsilon_N, \varepsilon_{N+1})$ and $\nu \in (\varepsilon_N - \nu', \varepsilon_{N+1} - \nu')$
\begin{equation} \label{eq:contour_deformation_P0}
	\widetilde{P^0} (\nu + \ri \cdot) = \dfrac{1}{2 \pi} \int_{-\infty}^{+\infty} \left( \widetilde{G_0} \big( \nu + \nu' + \ri (\omega' + \cdot) \big) \odot \widetilde{G_0}(\nu' + \ri \omega') \right) \rd \omega'.
\end{equation}
Note that the manipulations performed in this section to transform~\eqref{eq:P0_before} into~\eqref{eq:contour_deformation_P0} are possible since the two operators involved in the kernel-product (here, both are equal to $\widetilde{G^0}(z)$) are analytic on some common domain $\UU \cup \LL \cup (a,b)$ with $a < b$ (the presence of a gap is important to deform the contour as in Theorem~\ref{thm:analytic_continuation_kernel_product}).
\end{remark}

\subsubsection{The RPA reducible polarizability operator $\chi^0$}

In order to calculate the GW approximation of the self-energy, one needs the reducible polarizability operator $\chi$, defined in Section~\ref{ssec:chi}. Unfortunatly, the expression of $\chi$ is not accessible in practice. One needs to approximate this operator. The GW approximation, which amounts to approximating the so-called \textit{vertex function}, provides a natural approximation $\chi^\GW$ of $\chi$: in Equation~\eqref{eq:GWb}, $\chi^\GW$ is defined from $G^\GW$ (see also ~\cite[Equation~(103)]{Farid99} or~\cite[Equations~(A.20) and~(A.28)]{Hedin1965}). However, in view of Remark~\ref{pb:pb_P}, the definition of $\chi^\GW$ is not well-understood mathematically. In the $\GW^0$ framework, we use the RPA reducible polarizability operator $\chi^0$, which is itself defined in terms of the RPA irreducible polarizability $P^0$. The $\GW^0$ approximation of the (symmetrized) reducible polarizability operator is usually defined in the frequency domain as
\[
\widehat{\chi_\sym^0}(\omega) := \left( \mathds{1}_{\cH_1} - \widehat{P^{0}_\sym}(\omega) \right)^{-1} - \mathds{1}_{\cH_1}.
\]
The formal analytic continuation of the above definitions is (see~\cite[Equation~(139)]{Farid99})
\begin{equation} 
  \label{eq:chi0}
  \widetilde{\chi_\sym^0}(z) := \left( \mathds{1}_{\cH_1} - \widetilde{P^{0}_\sym}(z) \right)^{-1} - \mathds{1}_{\cH_1}.
\end{equation}
Note that we use the ``tilde'' notation in $\widetilde{\chi_\sym^0}$, although it is unclear that this operator-valued function is indeed the Laplace transform of some operator-valued function in the time domain. Also, it is \textit{a priori} unclear whether the operators $\mathds{1}_{\cH_1} - \widehat{P^{\GW}_\sym}(\omega)$ or $\mathds{1}_{\cH_1} - \widetilde{P^{\GW}_\sym}(z)$ are invertible. This is however the case for appropriate values of $z$, as shown by the following lemma.

\begin{lemma}
  For $z \in ( - (\varepsilon_{N+1} - \varepsilon_N), \varepsilon_{N+1} - \varepsilon_N)$ and $z \in \ri\R$, the operator $\mathds{1}_{\cH_1} - \widetilde{P^{0}_\sym}(z)$ is invertible.
\end{lemma}

This result is a direct consequence of the explicit formula~\eqref{eq:explicit_tildeP0} for $\widetilde{P^0}$, which ensures that $\widetilde{P^{0}_\sym}(z)$ is a bounded self-adjoint negative operator for the values of~$z$ under consideration. Let us deduce some extra properties of $\widetilde{\chi^0}$.

\begin{lemma}
  For any $\omega \in \R$, the operator $\widetilde{\chi^0_\sym}(\ri\omega)$ is a bounded, negative, self-adjoint operator on~$\cH_1$, satisfying $\widetilde{\chi^0_\sym}(- \ri \omega) = \widetilde{\chi^0_\sym}(\ri \omega)$, and such that
  \begin{equation}
    \label{eq:P_iomega_leq_chi_iomega}
    \widetilde{P_\sym^0}(\ri \omega) \le \widetilde{\chi_\sym^0}(\ri \omega) \le 0.
  \end{equation}
  The function $\omega \mapsto \widetilde{\chi^0_\sym}(\ri\omega)$ is analytic from $\R_\omega$ to $\cS(\cH_1)$ and is in $L^p(\R_\omega,\cS(\cH_1))$ for all $p > 1$. Finally, there exists a constant $C \in \R^+$ such that
  \begin{equation} \label{eq:estimate_chi0sym}
  	0 \le - \widetilde{\chi^0_\sym}(\ri \omega) \le \dfrac{C}{(\omega^2 + 1)^{1/2}} \left( \vc^{1/2} \rho_{0,N}^0 \vc^{1/2} \right).
  \end{equation}
\end{lemma}

This result is deduced from the definition~\eqref{eq:chi0}, the inequality $x \le (1 - x)^{-1} - 1 \le 0$ for $x \leq 0$, and Proposition~\ref{prop:propertiesP0}.

\paragraph{Sum-rule for $\widetilde{\chi^0}$.} From the sum-rule stated in Theorem~\ref{th:sumrule_P0}, we readily deduce the sum-rule for $\widetilde{\chi^0} := \vc^{-1/2} \widetilde{\chi^0_\sym} \vc^{-1/2}$, which is a bounded operator from $\cC'$ to $\cC$. Indeed, from the equality $(1 - x)^{-1} -1 = x + x^2 (1 - x)^{-1}$, we obtain
\[
	\forall \omega \in \R_\omega, \quad 
	\widetilde{\chi^0_\sym}(\ri \omega)  = \widetilde{P^0_\sym}(\ri \omega) + \left(\widetilde{P^0_\sym}(\ri \omega) \right)^2 \left( \mathds{1}_{\cH_1} - \widetilde{P^0_\sym}(\ri \omega) \right)^{-1}.
\]
In particular,
\[
	\forall \omega \in \R_\omega, \quad
	\omega^2 \widetilde{\chi^0}(\ri \omega) = \omega^2 \widetilde{P^0}(\ri \omega) + \dfrac{1}{\omega^2} \left( \omega^2 \widetilde{P^0}(\ri \omega) \right) \left( v_c^{1/2} \left( \mathds{1}_{\cH_1} - \widetilde{P^0_\sym}(\ri \omega) \right)^{-1} v_c^{1/2} \right)  \left( \omega^2 \widetilde{P^0}(\ri \omega) \right).
\]
This shows that the asymptotic behavior of $\widetilde{\chi^0}(\ri \omega)$ is, at dominant order, the same as for $\widetilde{P^0}(\ri \omega)$.
Taking the limit $\omega \to \pm \infty$ leads to a theorem similar to Theorem~\ref{th:sumrule_P0}, whose proof is skipped here for the sake of brevity.
\begin{theorem} \label{th:sumrule_chi0}
 The following weak-convergence holds:
	\[
		\forall (f,g) \in \cC' \times \cC', \quad \lim_{\omega \to \pm \infty} \left\bra \overline{f}, - \omega^2 \widetilde{\chi^0}(\ri \omega) g \right\ket_{\cC', \cC} = \left\bra \overline{f}, - \div \left( \rho_{0,N}^0 \nabla g \right) \right\ket_{\cC', \cC} = \int_{\R^3} \rho_{0, N}^0 \overline{\nabla f} \cdot \nabla g.
	\]
For all $g \in \cC'$ such that $\Delta g \in L^2(\R^3)$, the following strong convergence holds:
	\[
		\lim_{\omega \to \pm \infty} \omega^2 \widetilde{\chi^0}(\ri \omega) g = \div \left( \rho_{0,N}^0 \nabla g \right) \quad \text{in } \cC.
	\]
\end{theorem}

By comparing Theorems~\ref{th:sumrule_chi0} and~\ref{lem:sumrule_chi}, we see why using~\eqref{eq:h1KS} instead of~\eqref{eq:h1} for the definition of $h_1$ may lead to better approximations, since $\rho_{0,N}^0 = \rho_N^0$ in this case, so that the $\GW$ approximation $\chi^{\GW}$ of $\chi$ becomes exact in the high imaginary frequency domain. \\

Theorem~\ref{th:sumrule_chi0} is useful for the design of the so-called \textit{Plasmon-Pole models} (PPM) \cite{Hybertsen1986, VonderLinden1988, Godby1989, Engel1993}. Since the definition~\eqref{eq:chi0} requires the computation of a resolvent, the calculation of $\widetilde{\chi^0}(z)$ is numerically very expensive in practice. Some authors suggested to approximate $\widetilde{\chi^0}$ by an operator $\widetilde{\chi^\PPM}$ which is computationally less expensive. In practice, $\widetilde{\chi^\PPM}$ has a prescribed functional form, with adjustable parameters. Different approaches are taken in order to tune these parameters, and the previous sum-rule provides a standard way to fit some of them. This is done for instance in the PPM by Hybersten and Louie~\cite{Hybertsen1986} and in the PPM by Engel and Farid~\cite{Engel1993}. In the later article, the authors extensively comment on the fact that this sum-rule is an important requirement to be satisfied for a PPM.

\subsubsection{The RPA dynamically screened operator $W^0$}
\label{ssec:W0}

From the approximation $\chi^0$ of $\chi$, we directly deduce the approximation $W^0$ of $W$. Following the path taken in Section~\ref{subsec:W}, we define
\begin{equation} \label{eq:decomposition_W0}
	\widetilde{W^0}(z) := \vc + \widetilde{W_c^0}(z) 
	\quad \text{with} \quad
	\widetilde{W^0_c}(z) := \vc^{1/2} \widetilde{\chi_\sym^0}(z) \vc^{1/2}.
\end{equation}
This operator, when well-defined (say on the gap $( - (\varepsilon_{N+1} - \varepsilon_N), \varepsilon_{N+1} - \varepsilon_N)$ or on the imaginary axis $\ri \R$) is a bounded operator from $\cC$ to $\cC'$. The properties of $\widetilde{W^0}$ are directly deduced from the ones of $\widetilde{\chi^0_\sym}$, so we do not repeat them here for brevity.


\subsection{A mathematical study of the $\GW^0$ approximation}
\label{sec:GW0}

\subsubsection{The G$_0$W$^0$ approximation of the self-energy}
\label{ssec:self_energy}

In this section, we study the G$_0$W$^0$ approximation as a preliminary step to the study of the self-consistent $\GW^0$ approximation. This will help us understand some technical points to address in the analysis of the $\GW^0$ method. \\

The G$_0$W$^0$ approximation of the self-energy operator is formally defined as
\begin{equation} \label{eq:phys_Sigma}
  \Sigma^{00}(\br, \br', \tau) := \ri G_0(\br, \br', \tau) W^0(\br, \br', -\tau^+).
\end{equation}
Here, $G_0$ represents the Green's function of the non-interacting system introduced in Definition~\ref{def:G0}, and $W^0$ is the random phase approximation of the dynamically screened operator defined in Section~\ref{ssec:W0}. Already one difficulty arises: in Section~\ref{ssec:W0}, we only defined the function $\widetilde{W^0}(z)$ on the complex frequency domain, but we did not define some operator-valued function on the time-domain.
In this section, we assume that the function $\widetilde{W^0}(z)$ is indeed the Laplace transform of some operator $W^0(\tau)$. This will allow us to transform~\eqref{eq:phys_Sigma} into a formally equivalent definition that only involves $\widetilde{W^0}$. The resulting definition will be our starting point for the $\GW^0$ approximation.\\

With the kernel-product defined in Section~\ref{sec:kernel_products}, the definition~\eqref{eq:phys_Sigma} can be recast as
\[
	\Sigma^{00}(\tau) = \ri G_0(\tau^-) \odot W^0(-\tau).
\]
In view of the decomposition provided in~\eqref{eq:decomposition_W0}, it is natural to split $\Sigma^{00}$ into an exchange part $\Sigma_x^{00}$ and a correlation part $\Sigma_c^{00}$ (the terminology is motivated below):
\[
\Sigma ^{00}= \Sigma_x^{00} + \Sigma_c^{00} \quad \text{with} \quad  \Sigma_x^{00}(\tau)  = \ri G_{0,\rh}(0^-) \odot \vc \delta_0(\tau)
\quad \text{and} \quad \Sigma_c^{00}(\tau) = \ri G_0(\tau) \odot W_c (-\tau).
\]
Let us first consider the exchange part. As $\ri G_{0,\rh}(0^-)=-\gamma_{0,N}^0$, we obtain
\begin{equation} \label{eq:Sigma_x}
  \Sigma_x^{00}(\tau) =  K_x \delta_0(\tau),
\end{equation}
where $K_x$ is the integral operator on $\cH_1$ with kernel
\begin{equation} \label{eq:Kx}
K_x(\br,\br'):= - \dfrac{\gamma_{0,N}^0(\br, \br')}{| \br - \br'|}.
\end{equation}
We recover the usual Fock exchange operator associated with $\gamma_{0,N}^0$, which justifies the terminology ``exchange part'' for $\Sigma_x^{00}$. 
Let us now consider the correlation part. Observing that
\begin{itemize}
	\item $\widetilde{G_{0}}$ is analytic on $\UU \cup \LL \cup (\varepsilon_N,  \varepsilon_{N+1})$ (hence has a gap around $\mu_0$) ;
	\item $\widetilde{W^{0}}$ is analytic on $\UU \cup \LL \cup ( - (\varepsilon_{N+1} - \varepsilon_N), \varepsilon_{N+1} - \varepsilon_N)$ (hence has a gap around $0$),
\end{itemize}
we can use the same ideas as in Section~\ref{ssec:analytical_continuation}. \textit{By analogy} with Remark~\ref{rem:transformation_contour_deformation}, we recast the physical definition of $\widetilde{\Sigma^{00}}$ in~\eqref{eq:phys_Sigma} in a \textit{formally equivalent} definition in the complex frequency plane. This reformulation was first given by Rojas, Godby and Needs \cite{Rojas1995} (see also \cite{RSWRG99}), and is now known as the ``contour deformation'' technique. 

\begin{definition}[G$_0$W$^0$ approximation of the self-energy]
The exchange part of the self-energy in the G$_0$W$^0$ approximation is defined in the complex frequency domain by
\[
	\forall z \in \C, \qquad \widetilde{\Sigma_x^{00}}(z) = K_x,
\]
while the correlation part is defined, for $\nu' \in (-(\varepsilon_{N+1} - \varepsilon_N), \varepsilon_{N+1} - \varepsilon_N)$ and $\nu + \nu' \in (\varepsilon_N , \varepsilon_{N+1} )$ by
\[
	\forall \omega \in \R_\omega, \quad \widetilde{\Sigma_c^{00}}(\nu + \ri \omega) = - \dfrac{1}{2 \pi} \int_{-\infty}^{+\infty} \widetilde{G_0} \left( \nu + \nu' + \ri ( \omega + \omega') \right) \odot \widetilde{W_c^0}(\nu' + \ri \omega') \ \rd \omega.
\]
\end{definition}

The fact that the above quantity is independent of the choice of $\nu'$ comes from the analyticity of the integrand on the region of interest. 
In practice, we will focus on the case $\nu' = 0$ and $\nu = \mu_0$, and therefore consider the function $\R_\omega \ni \omega \mapsto \widetilde{\Sigma_0^{00}}(\mu_0 + \ri \omega)$ defined by
\begin{equation} \label{eq:widetilde_Sigma00}
	\forall \omega \in \R_\omega, \quad \widetilde{\Sigma_c^{00}}(\mu_0 + \ri \omega) 
	= - \dfrac{1}{2 \pi} \int_{-\infty}^{+\infty} \widetilde{G_0} \left( \mu_0 + \ri (\omega + \omega') \right) \odot \widetilde{W^0_c}(\ri \omega') \ \rd \omega.
\end{equation}

The next proposition shows that the above definition makes sense.
\begin{proposition} \label{prop:Sigma00}
	The operator $K_x$ arising in the exchange part $\Sigma_c^{00}$ of the self-energy is a negative Hilbert-Schmidt operator on $\cH_1$. Furthermore, for any $\omega \in \R_\omega$, the operator $\widetilde{\Sigma_c^{00}}(\mu_0 + \ri \omega)$ is a bounded operator on $\cH_1$, and satisfies $\widetilde{\Sigma_c^{00}}(\mu_0 - \ri \omega) = \widetilde{\Sigma_c^{00}}(\mu_0 + \ri \omega)^\ast$.
	The function $\omega \mapsto \widetilde{\Sigma_c^{00}}(\mu_0 + \ri \omega)$ is analytic from $\R_\omega$ to $\cB(\cH_1)$ and is in $L^p(\R_\omega, \cB(\cH_1))$ for all $p > 1$.
\end{proposition}

The first statements of Proposition~\ref{prop:Sigma00} can be seen as a special case of Proposition~\ref{prop:fg_fs}, while the symmetry property for the adjoint and the $L^p$ integrability follow from the properties of $\widetilde{G^0}$ and $\widetilde{W^0_c}$.
%

\subsubsection{Well-posedness of the $\GW^0$ approximation in the perturbative regime.}

We finally study the $\GW^0$ approximation. Following our definition~\eqref{eq:widetilde_Sigma00} of the G$_0$W$^0$ approximation of the self-energy, we recast the $\GW^0$ equation as follows.

\begin{definition} [The $\GW^0$ problem on the imaginary axis in the frequency domain] ~\\
	Find $G^{\GW^0} \in L^\infty(\R_\omega, \cB(\cH_1))$ solution to the system
	\begin{equation*}
	(\GW^0) \quad 
	\left\{ \begin{aligned}
	\widetilde{\Sigma^{\GW^0}} (\mu_0 + \ri \omega) & = K_x - \dfrac{1}{2 \pi} \int_{-\infty}^{+\infty} \widetilde{G^{\GW^0}} \big( \mu_0 + \ri (\omega + \omega') \big) \odot \widetilde{W_c^0}(\ri \omega') \, \rd \omega',  \\
	\widetilde{G^{\GW^0}}(\mu_0 + \ri \omega) & = \left[ \mu_0 + \ri \omega - \left( h_1 + \widetilde{\Sigma^{\GW^0}}(\mu_0 + \ri \omega) \right) \right]^{-1},
		\end{aligned} \right.
	\end{equation*}
	where $h_1$ is the one-body mean-field Hamiltonian defined in~\eqref{eq:h1} and $K_x$ is the exchange operator defined by~\eqref{eq:gamma0N0}-\eqref{eq:Kx}.
\end{definition}

\begin{remark}
	We are looking for a solution in $L^\infty(\R_\omega, \cB(\cH_1))$. Note that the true Green's function~$\widetilde{G}(\mu + \ri \cdot)$ is in $L^p(\R_\omega, \cB(\cH_1))$ for all $p > 1$ (in particular for $p = \infty$). We chose to work with $L^\infty(\R_\omega, \cB(\cH_1)$ for simplicity, but it is possible to work with other spaces $L^p(\R_\omega, \cB(\cH_1))$ with~$p > 1$.
\end{remark}

Since this problem seems quite difficult to study mathematically, we will only study it in a perturbative regime. More specifically, seeing $\Sigma^{\GW}$ as a correction term (see the discussion after Definition~\ref{def:self_energy}), we propose to study the following problem.
\begin{definition} [The $\GW^0_\lambda$ problem on the imaginary axis on the frequency domain] ~\\
	Find $G^{\GW^0_\lambda} \in L^\infty(\R_\omega, \cB(\cH_1))$ solution of the system
	\begin{equation} \label{eq:GW0_lambda}
	(\GW^0_\lambda) \quad 
	\left\{ \begin{aligned}
	\widetilde{\Sigma^{\GW^0_\lambda}} (\mu_0 + \ri \omega) & = K_x - \dfrac{1}{2 \pi} \int_{-\infty}^{+\infty} \widetilde{G^{\GW^0_\lambda}} \big( \mu_0 + \ri (\omega + \omega') \big) \odot \widetilde{W_c^0}(\ri \omega') \, \rd \omega'  \\
	\widetilde{G^{\GW^0_\lambda}}(\mu_0 + \ri \omega) & = \left[ \mu_0 + \ri \omega - \left( h_1 + \lambda \widetilde{\Sigma^{\GW^0_\lambda}}(\mu_0 + \ri \omega) \right) \right]^{-1}.
		\end{aligned} \right.
	\end{equation}
\end{definition}

According to~\eqref{eq:G0h1}, the unique solution for $\lambda = 0$ is the Green's function for the non interacting system $\widetilde{G^{\GW^0_{\lambda = 0}}} = \widetilde{G_0}$. This fact will allow us to treat the equation perturbatively. The exact $\GW^0$ equations correspond to the case $\lambda = 1$. Of course, several other choices of perturbation can be used. For instance, we can put the parameter $\lambda$ in front of the correlation part of the self-energy only. This amounts to considering the Hartree-Fock Hamiltonian as the reference Hamiltonian (instead of the Hartree Hamiltonian). The theory that we develop here can be straightforwardly generalized to such other cases. \\

It is convenient for the mathematical analysis to introduce the functionals $\fs$ and $\ffg$ respectively defined as
\[
	 \begin{array}{rrcl}
    \fs : & L^2(\R_\omega, \cB(\cH_1)) & \to & L^\infty(\R_\omega, \cB(\cH_1)) \\
    & \widetilde{G^\app}(\mu_0 + \ri \cdot)  & \mapsto & \dps \fs \left[ \widetilde{G^\app} \right](\mu_0 + \ri \cdot) := K_x - \dfrac{1}{2 \pi} \int_{-\infty}^{+\infty} \widetilde{G^{\app}}(\mu_0 + \ri (\cdot + \omega')) \odot \widetilde{W_c^0}(\ri \omega') \, \rd \omega',
  \end{array}
\]
and
\[
	\begin{array}{rrcl}
    \ffg_\lambda : & L^\infty(\R_\omega, \cB(\cH_1)) & \to & L^2(\R_\omega, \cB(\cH_1)) \\
    & \widetilde{\Sigma^\app}(\mu_0 + \ri \cdot)  & \mapsto & \dps \ffg \left[ \widetilde{\Sigma^\app} \right](\mu_0 + \ri \cdot) := \left[ \mu_0 + \ri \cdot - \left( h_1 + \lambda \widetilde{\Sigma^{\app}}(\mu_0 + \ri \cdot) \right) \right]^{-1}.
  \end{array}
\]
With this notation, $\widetilde{G^{\GW^0_\lambda}}$ is a solution of the $\GW^0_\lambda$ equations~\eqref{eq:GW0_lambda} if and only if it is a fixed-point of $\ffg_\lambda \circ \fs$. The fact that these maps are indeed well-defined is proved in the following proposition (see Section~\ref{proof:fg_fs} for the proof).

\begin{proposition} \label{prop:fg_fs}
	The operator $\fs$ is a bounded linear operator from $L^2(\R_\omega, \cB(\cH_1))$ to $L^\infty(\R_\omega, \cB(\cH_1))$.
	On the other hand, for all $M > 0$, there exists $\lambda_M > 0$ and $C_M \in \R^+$ such that for all $0 \le \lambda < \lambda_M$, and all $\widetilde{\Sigma^\app}$ such that
	$\left\|  \widetilde{\Sigma^\app}(\mu_0 + \ri \cdot) \right\|_{L^\infty(\R_\omega, \cB(\cH_1))} \le M$, the function $\ffg_\lambda[\Sigma^\app](\mu_0 + \ri \cdot)$ is well-defined as an element of $L^2(\R_\omega, \cB(\cH_1)) \cap L^\infty(\R_\omega, \cB(\cH_1))$, with
	\[
		\left\| \ffg_\lambda \left[\widetilde{\Sigma^\app} \right](\mu_0 + \ri \cdot) \right\|_{L^2(\R_\omega, \cB(\cH_1))} + \left\| \ffg_\lambda \left[ \widetilde{\Sigma^\app} \right](\mu_0 + \ri \cdot) \right\|_{L^\infty(\R_\omega, \cB(\cH_1))}  \le C_M.
	\]
	Moreover, for all $\widetilde{\Sigma^\app_1}, \widetilde{\Sigma^\app_2} \in L^\infty(\R_\omega, \cB(\cH_1))$ such that $\left\| \widetilde{\Sigma^\app_{j}}(\mu_0 + \ri \cdot) \right\|_{L^\infty(\R_\omega, \cB(\cH_1))} \le M$ for $1 \le j \le 2$,
	\begin{equation} \label{eq:resolvent_sigma}
		\ffg_\lambda \left[ \widetilde{\Sigma^\app_1} \right] - \ffg_\lambda \left[ \widetilde{\Sigma^\app_2} \right] = \lambda \ffg_\lambda\left[ \widetilde{\Sigma^\app_1} \right] \left( \widetilde{\Sigma^\app_2} - \widetilde{\Sigma^\app_1} \right) \ffg_\lambda \left[ \widetilde{\Sigma^\app_2} \right].
	\end{equation}
\end{proposition}

To prove the existence of a fixed-point for $\ffg_\lambda \circ \fs$, we rely on Picard's fixed-point theorem. Since the solution of the $\GW^0_{\lambda = 0}$ equations~\eqref{eq:GW0_lambda} for $\lambda = 0$ is $\widetilde{G_0}$, we are lead to introduce, for $r >0$, the (closed) ball
\[
	\fB \left( \widetilde{G_0}, r \right) = \left\{ \widetilde{G^\app}(\mu_0 + \ri \cdot) \in L^2(\R_\omega, \cB(\cH_1))\, , \ \left\| \widetilde{G^\app}(\mu_0 + \ri \cdot) - \widetilde{G_0}(\mu_0 + \ri \cdot) \right\|_{L^2(\R_\omega, \cB(\cH_1))} \le r \right\}.
\]
The existence of a fixed-point is given by the following theorem (see Section~\ref{proof:contraction} for the proof).
\begin{theorem} \label{th:contraction}
	There exists $\lambda_\ast  > 0$ and $r > 0$ such that, 
	for all $0 \le \lambda \le \lambda_\ast$, there exists a unique element $\widetilde{G^{\GW^0_\lambda}} \in \fB \left( \widetilde{G_0}, r \right)$ solution to the $\GW^0_\lambda$ equations~\eqref{eq:GW0_lambda}, or equivalently to the fixed point equation
	\[
		\widetilde{G^{\GW^0_\lambda}} = \ffg_\lambda \circ \fs \left( \widetilde{G^{\GW^0_\lambda}} \right).
	\]
	In addition, for all $\omega \in \R_\omega$, $\widetilde{G^{\GW^0_\lambda}}(\mu_0 + \ri \omega)$ is an invertible operator, and
	\begin{equation} \label{eq:G_GW_vs_G0}
		\left\| \left(\widetilde{G^{\GW^0_\lambda}}(\mu_0 + \ri \cdot) \right)^{-1} - \left( \widetilde{G_0}(\mu_0 + \ri \cdot)\right)^{-1} \right\|_{L^\infty(\R_\omega, \cB(\cH_1))} < \infty.
	\end{equation}
	Finally, the iterative sequence $(\ffg_\lambda \circ \fs)^k \left[ \widetilde{G_0} \right]$ converges to $\widetilde{G^{\GW^0_\lambda}}$, and there exists $0 \le \alpha < 1$ and~$C \in \R^+$ such that
	\[
		\left\| \left( \widetilde{G^{\GW^0_\lambda}} - (\ffg_\lambda \circ \fs)^k \left[ \widetilde{G_0}\right] \right)(\mu_0 + \ri \cdot) \right\|_{L^\infty(\R_\omega, \cB(\cH_1))} \le C\alpha^k.
	\]
\end{theorem}

\begin{remark}
	It is not difficult to deduce from~\eqref{eq:G_GW_vs_G0} that the function $\omega \mapsto \widetilde{G^{\GW^0_\lambda}}(\mu + \ri \omega)$ is actually in $L^p(\R_\omega, \cB(\cH_1))$, for all $p > 1$.
\end{remark}


\section{Conclusion}

This article is, to our knowledge, the first attempt to formalize with full mathematical rigor the GW theory for finite molecular systems derived by Hedin in his seminal work published in 1965~\cite{Hedin1965}. In Section~\ref{sec:operators}, we have provided a mathematical definition of some one-body operators arriving in many-body perturbation theory for electronic systems, namely the one-body Green's function $G$, the spectral function $\cA$, the reducible polarizability operator $\chi$, the dynamically screened interaction operator $W$, and the self-energy operator $\Sigma$. \\

In Section~\ref{sec:GW_finite}, we have worked out a mathematically consistent formulation of the $\GW^0$ approximation of the GW equations, and we have proved that the $\GW^0$  model has a solution in a perturbation regime. As a by-product, we have also shown that the widely used G$_0$W$^0$ approximation of the self-energy makes perfect mathematical sense.

\newpage


\section{Proofs} 
\label{sec:proofs}

\subsection{Proof of Lemma \ref{lem:image_Fourier_L_infty}}
\label{sec:image_Fourier_L_infty}

Let $s > 1/2$. For $f\in L^\infty(\R_\tau)$ and $\widehat \varphi \in \sS(\R_\omega)$,
\[
\begin{aligned}
  \left|\left\langle \mathcal{F}_T f, \widehat\varphi \right\rangle_{\sS',\sS}\right| 
  & = \left|\left\langle f, {\mathcal{F}_T{\widehat\varphi}} \right\rangle_{\sS',\sS}\right| 
  = 2\pi \left|\int_\R f(-\tau) {(\mathcal{F}_T^{-1}{\widehat\varphi})(\tau)}  \, \rd \tau\right| \\
  & = 2\pi  \left|\int_\R \frac{f(-\tau)}{(1+\tau^2)^{s/2}} \, (1+\tau^2)^{s/2} {(\mathcal{F}^{-1}_T{\widehat\varphi})(\tau)} \, \rd \tau\right| 
   \leq C_s \, \| f \|_{L^\infty} \| \widehat\varphi \|_{H^s},
\end{aligned}
\]
where we have used the Cauchy-Schwarz inequality in the last step. By density, $\mathcal{F}_T f$ can be extended to a linear form on~$H^s(\R)$. The equality case $\| \mathcal{F}_Tf \|_{H^{-s}} = C_s \| f\|_{L^\infty}$ is obtained for constant functions.


\subsection{Proof of Theorem \ref{th:Tit_Linfty}}
\label{ssec:proof_Tit_Linfty}

\paragraph{Proof of~(i).} The analyticity  of $\widetilde g$ directly follows from the results of~\cite[Chapter~VIII]{Schwartz}.

\paragraph{Proof of~(ii).}
Let $s > 1/2$, and consider $\varphi \in \sS(\R)$. Relying on the fact that $\widetilde g(\cdot + \ri \eta)$ can be seen as the Fourier transform of~$\tau \mapsto g(\tau)\re^{-\eta \tau}$, we obtain
\begin{equation}
\label{eq:reformulation_duality_widetilde_g}
\begin{aligned}
\langle \widetilde g(\cdot + \ri \eta), \varphi \rangle_{H^{-s},H^s} - \langle \widehat g, \varphi \rangle_{H^{-s},H^s} 
	& =\left\langle g \, \re^{-\eta \tau}, \widehat{\varphi} \right\rangle_{\sS',\sS} - \left\langle g, \widehat{\varphi} \right\rangle_{\sS',\sS} \\
	& = \int_0^\infty \left( \dfrac{g(\tau)}{(1 + \tau^{2})^{s/2}} \right) (1 + \tau^{2})^{s/2} \widehat{\varphi}(\tau) \left( \re^{ - \eta \tau} - 1 \right) \rd \tau,
\end{aligned}
\end{equation}
where the integral makes sense since $\tau \mapsto g(\tau) (1 + \tau^2)^{-s/2}$ and $\tau \mapsto (1 + \tau^2)^{s/2} \widehat{\varphi}(\tau)$ are in $L^2(\R)$. It is then possible to extend the above formula to any $\varphi \in H^s(\R)$. Moreover, by the Cauchy-Schwarz inequality,
\[
\left| \langle \widetilde g(\cdot + \ri \eta), \varphi \rangle_{H^{-s},H^s} - \langle \widehat g, \varphi \rangle_{H^{-s},H^s} \right| \leq I_{\eta,s} \| \varphi \|_{H^s} \| g\|_{L^\infty},
\]
where 
\[
I_{\eta,s} = \left(2 \pi \int_0^{+\infty} \frac{(1-\re^{-\eta \tau})^2}{(1+\tau^2)^s}\, \rd \tau \right)^{1/2} < \infty.
\]
Therefore, $\left\| \widetilde g(\cdot + \ri \eta) - \widehat{g} \right\|_{H^{-s}} \leq \| g\|_{L^\infty} I_{\eta, s}$. By dominated convergence, $I_{\eta, s} \to 0$ as $\eta \to 0^+$, which allows us to conclude to the strong convergence of $\widetilde g(\cdot + \ri \eta)$ to $\widehat{g}$ in $H^{-s}(\R_\omega)$.

A similar computation shows that, for $0 < \eta_1 \leq \eta_2$ and $s \in \RR$,
\[
\left\| \widetilde{g}(\cdot + \ri \eta_1)-\widetilde{g}(\cdot + \ri \eta_2)\right\|_{H^{s}} \leq \| g\|_{L^\infty} \left(2 \pi \int_0^{+\infty} \re^{-2\eta_1 \tau} \, \left(1-\re^{-(\eta_2-\eta_1) \tau}\right)^2 (1+\tau^2)^s\, \rd \tau\right)^{1/2}, 
\]
where we crucially use that $\eta_1 > 0$ to ensure the convergence of the time integral when $s > -1/2$. The right-hand side goes to $0$ as $\eta_2$ goes to $\eta_1$ by dominated convergence. This allows one to conclude to the continuity of $\eta \mapsto \widetilde{g}(\cdot + \ri \eta)$ from $(0,+\infty)$ to $H^{s}(\RR)$. When $s < -1/2$, it is possible to pass to the limit $\eta_1 \to 0$ and obtain the uniform continuity from $[0,+\infty)$ to $H^{s}(\RR)$. 

\paragraph{Proof of~(iii).} We follow the approach used in~\cite{Taylor58} for instance. Fix $z_0 \in \UU$, and consider, for $R > 0$ and $0 < \alpha \leq \Im(z_0)/2$, the oriented contour~$\sC$ in the complex plane composed of the semi-circle $\ri\alpha + R\re^{\ri \theta}$ for $0 \leq \theta \leq \pi$ and the line $\ri\alpha + \omega$ for $-R \leq \omega \leq R$.
The value~$R$ is taken sufficiently large for $z_0$ to be inside the domain encircled by the contour (see Figure~\ref{fig:half_circle}).

\begin{figure}
\begin{center}
\begin{tikzpicture}[scale =1]

	\draw (-5,0) -- (5,0);
	\draw (0, -1) -- (0, 5.5);

	\draw (1, 0.8) -- (1.2, 1);
	\draw (1.2, 0.8) -- (1, 1);
	\node at (1.5, 0.9) {$z_0$};
	
	\draw (-0.1, 0.4) -- (0.1, 0.4);
	\node at (0.2, 0.2) {$\ri \alpha$};
	
	\draw (4.5,-0.1) -- (4.5, 0.1);
	\node at (4.6, 0.2) {$R$};
	
	
	\draw[blue, very thick] (-4.5, 0.4) --(4.5, 0.4);
	\draw[blue, very thick] (3.8, 2.3) --(3.8, 2.8);
	\draw[blue, very thick] (4.2, 2.6) --(3.8, 2.8);
	
	\draw[blue, very thick] (4.5, 0.4) arc(0:180:4.5);
	\node[blue, very thick] at (3, 4.2) {$\sC$};
	
\end{tikzpicture}
	\caption{The contour $\sC$ used in the proof of~(iii).}
	\label{fig:half_circle}
\end{center}
\end{figure}
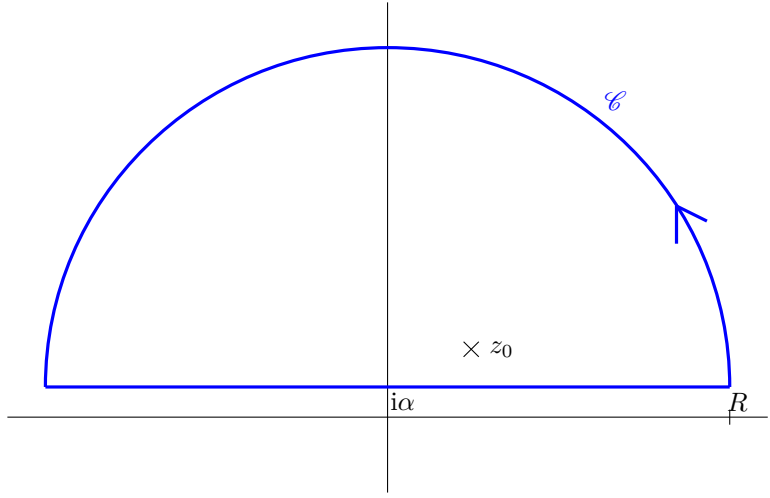

By Cauchy's residue theorem,
\begin{equation}
\label{eq:Cauchy_g_Linfty}
\widetilde{g}(z_0) = \frac{1}{2\ri\pi} \oint_\sC \frac{\widetilde{g}(z)}{z-z_0} \, \rd z = 
\frac{1}{2\ri\pi} \int_{-R}^R \frac{\widetilde{g}(\omega + \ri\alpha)}{\omega+\ri\alpha-z_0} \, \rd \omega + \frac{1}{2\pi}\int_0^\pi \widetilde{g}(\ri\alpha +R\re^{\ri\theta}) \frac{R\re^{\ri\theta}}{R\re^{\ri\theta}+\ri\alpha - z_0} \, \rd \theta .
\end{equation}
Now, for $z \in \UU$,
\[
\left| \widetilde{g}(z) \right| \leq \frac{\| g\|_{L^\infty}}{\Im(z)},
\]
so that 
\[
\left| \int_0^\pi \widetilde{g}(\ri\alpha +R\re^{\ri\theta}) \frac{R\re^{\ri\theta}}{R\re^{\ri\theta}+\ri\alpha - z_0} \, \rd \theta \right| \leq \| g\|_{L^\infty} \int_0^\pi \frac{R}{|\alpha +R\sin\theta|\,|R\re^{\ri\theta} + \ri\alpha -z_0|} \, \rd \theta,
\]
which, by dominated convergence, converges to~0 as $R\to +\infty$ when $\alpha$ is fixed. On the other hand, $\widetilde{g}(\cdot + \ri \alpha)$ belongs to $L^2(\R_\omega)$, while $(\cdot + \ri \alpha -z_0)^{-1}$ is in $H^1(\R_\omega)$, since $\ri\alpha-z_0$ has a non-zero imaginary part. Therefore, the limit $R \to +\infty$ can be taken in the first integral on the right-hand side of~\eqref{eq:Cauchy_g_Linfty}, which leads to
\[
\widetilde{g}(z_0) = \frac{1}{2\ri\pi} \int_{-\infty}^{+\infty} \frac{\widetilde{g}(\omega + \ri\alpha)}{\omega+\ri\alpha-z_0} \, \rd\omega  = \frac{1}{2\ri\pi} \left\langle \widetilde{g}(\cdot + \ri\alpha), (\cdot+\ri\alpha-z_0)^{-1} \right\rangle_{H^{-1},H^1}. 
\]
The conclusion now follows from the strong convergences of $(\cdot + \ri \alpha -z_0)^{-1}$ to $(\cdot -z_0)^{-1}$ in $H^1(\R_\omega)$ and of $\widetilde{g}(\cdot + \ri \alpha)$ to $\widehat{g}$ in $H^{-1}(\R_\omega)$ as $\alpha \to 0$.

\paragraph{Proof of~(iv).} Let $\varphi$ be a real-valued function in $\sS(\R_\omega)$. From (\ref{eq:gtildez}), we get
\[
\begin{aligned}
\int_\RR \widetilde{g}(\omega + \ri\eta) \varphi(\omega)\, \rd \omega & = \frac{1}{2\ri\pi} \int_\R \left\langle \widehat{g}, (\cdot-\omega-\ri\eta)^{-1} \right\rangle_{H^{-1},H^1} \varphi(\omega) \, \rd \omega .
\end{aligned}
\]
Taking the real parts of both sides, we obtain
\begin{equation}
\label{eq:reformulation_finite_eta_dispersion_relation}
\begin{aligned}
&\int_\RR \Re\left(\widetilde{g}(\omega + \ri\eta) \right) \varphi(\omega)\, \rd \omega \\
& \qquad = \frac{1}{2\pi}\int_\R \left(\left\langle \Im \widehat{g}, \frac{\cdot - \omega}{(\cdot-\omega)^2 + \eta^2} \right\rangle_{H^{-1},H^1} + \left\langle \Re \widehat{g}, \frac{\eta}{(\cdot-\omega)^2 + \eta^2} \right\rangle_{H^{-1},H^1}\right)\varphi(\omega) \, \rd \omega. 
\end{aligned}
\end{equation}
Consider now $\phi \in C^\infty(\RR^2)$, with support contained in $[-R,R] \times \RR$ for some finite $R > 0$. Then, Fubini's theorem for distributions (see~\cite[Chapter~IV, Theorem~IV]{Schwartz}) asserts that, for a given distribution $T \in \sS'(\RR)$ and $\varphi \in \sS(\R)$,
\[
\int_\RR \left\langle T,\phi(\cdot,\omega) \right\rangle_{\sS',\sS} \varphi(\omega) \, \rd\omega = \left\langle T,\int_\RR\phi(\cdot,\omega)\varphi(\omega) \, \rd\omega \right\rangle_{\sS',\sS}.
\]
When $T \in H^{-1}(\RR)$, the above linear form can be extended to functions in $H^1(\RR)$. Therefore, \eqref{eq:reformulation_finite_eta_dispersion_relation} can be rewritten as
\begin{equation}
\label{eq:reformulation_finite_eta_dispersion_relation_bis}
\begin{aligned}
\int_\RR \Re\left(\widetilde{g}(\omega + \ri\eta) \right) \varphi(\omega)\, \rd \omega & = \left\langle \Im \widehat{g},\frac{1}{2\pi}\int_\R \frac{\cdot - \omega}{(\cdot-\omega)^2 + \eta^2} \varphi(\omega) \, \rd\omega\right\rangle_{H^{-1},H^1} \\
& \ \ + \left\langle \Re \widehat{g},\frac{1}{2\pi}\int_\R \frac{\eta}{(\cdot-\omega)^2 + \eta^2} \varphi(\omega) \, \rd\omega\right\rangle_{H^{-1},H^1}.
\end{aligned}
\end{equation}
In view of the following strong convergences in $H^1(\RR)$,
\[
\frac{1}{2\pi}\int_\R \frac{\xi - \omega}{(\xi-\omega)^2 + \eta^2} \, \varphi(\omega) \, \rd\omega \xrightarrow[\eta\to 0]{} \frac12 (\fH \varphi)(\xi),
\qquad
\frac{1}{2\pi}\int_\R \frac{\eta}{(\xi-\omega)^2 + \eta^2} \, \varphi(\omega) \, \rd\omega \xrightarrow[\eta\to 0]{} \frac12 \varphi(\xi),
\]
the equality~\eqref{eq:reformulation_finite_eta_dispersion_relation_bis} leads, in the limit $\eta \to 0^+$, to
\[
\left\langle \Re \widehat{g}, \varphi\right\rangle_{H^{-1},H^1}  = \frac12 \left\langle \Im \widehat{g}, \fH(\varphi)\right\rangle_{H^{-1},H^1} + \frac12 \left\langle \Re \widehat{g}, \varphi\right\rangle_{H^{-1},H^1}. 
\]
The first equality in the statement of item~(iv) is finally obtained with the following lemma (recall that, according to Lemma~\ref{lem:hilbert_hm1}, $H^s(\RR)$ is stable by the Hilbert transform). The second equality follows by applying $\fH$ to both sides and remembering that $\fH^2 = -\mathrm{Id}$.

\begin{lemma}
\label{lem:fH_duality}
Let~$s \geq 0$. For any $T \in H^{-s}(\RR)$ and $\varphi \in H^s(\RR)$,
\[
\left\langle \fH(T), \varphi\right\rangle_{H^{-s},H^s} = -\left\langle T, \fH(\varphi)\right\rangle_{H^{-s},H^s}.
\]
\end{lemma}

\begin{proof}
  Consider first the case when $T,\varphi \in \sS(\RR)$. Then, using Plancherel's formula, the duality product can be rewritten using a $L^2$-scalar product 
  \[
  \begin{aligned}
    \langle \fH T, \varphi \rangle_{\sS',\sS} & = ( \overline{\fH T}, \phi)_{L^2} = 2\pi\left(\overline{\cF^{-1}(\fH T)} , \cF^{-1} \varphi\right)_{L^2} = 2\pi\left({-\ri \, \sgn(\cdot) \cF^{-1} \overline{T}} , \cF^{-1} \varphi\right)_{L^2} \\ 
    & = 2\pi\left({\cF^{-1} \overline{T}} , \ri \, \sgn(\cdot)\cF^{-1} \varphi\right)_{L^2} = -\left( \overline{T}, \fH \varphi\right)_{L^2} = -\langle T, \fH \varphi \rangle_{\sS',\sS}.
  \end{aligned}
  \]
  The conclusion is obtained by a density argument.
\end{proof}

\subsection{Proof of Proposition~\ref{lem:particle-type}}
\label{sec:proof_lem:particle-type}

The proof presented in Section~\ref{ssec:proof_Tit_Linfty} can be followed \emph{mutatis mutandis} upon introducing, for given elements $f,g \in \cH$, the bounded causal function
\[
a_{f,g}(\tau) = \left \langle f\left|T_{\rm c}(\tau)\right|g \right\rangle,
\]
and noting that $\| a_{f,g} \|_{L^\infty} \leq \| T_{\rm c}\|_{L^\infty(\cB(\cH))} \| f \| \| g\|$.

The only additional technical point is the strong analyticity property, which is however easily obtained from the following bound: for $z = \omega + \ri \eta \in \UU$, 
\[
\begin{aligned}
\left\| \frac{\rd \widetilde{T_{\rm c}}(z)}{\rd z} \right\|_{\cB(\cH)} & = \left\| \int_0^\infty T_{\rm c}(\tau) (\ri \tau) \re^{-\eta \tau} \re^{\ri \omega \tau} \, \rd \tau \right\|_{\cB(\cH)} \\
& \leq \| T_{\rm c} \|_{L^\infty(\R_\tau, \cB(\cH))} \int_0^\infty  \tau \re^{-\eta \tau} \, \rd \tau = \frac{\| T_{\rm c} \|_{L^\infty(\R_\tau,\cB(\cH))}}{\eta^2} < +\infty .
\end{aligned}
\]

\subsection{Proof of Proposition~\ref{prop:analytic_extension_causal_time_evolution}}
\label{sec:proof:prop:analytic_extension_causal_time_evolution}

For $z \in \UU$, we have
\[
\widetilde{A}_{\rm c}(z) = \int_\R A_{\rm c}(\tau) \, \re^{\ri z \tau} \, \rd \tau = -\ri \int_0^{+\infty} \re^{-\ri \tau H} \re^{\ri z \tau} \, \rd\tau.
\]
A simple computation shows that 
\[
\widetilde{A}_{\rm c}(\omega + \ri\eta) = -\ri \int_0^{+\infty}  \int_\R \re^{\ri \tau (\omega + \ri \eta - \lambda)}\,\rd P_\lambda^H \, \rd\tau = \int_\R \frac{1}{\omega + \ri \eta - \lambda} \, \rd P_\lambda^H = (\omega+\ri\eta-H)^{-1}. 
\]
The series of equalities can be made rigorous by testing them against functions $f,g \in \cH$, and using Fubini's theorem to justify the exchange in the order of integration.

The strong convergence of $\widetilde{A}_{\rm c}(\cdot + \ri \eta)$ to $\widehat{A_{\rm c}}$ in $H^1(\R_\tau,\cB(\cH))$ is ensured by Proposition~\ref{lem:particle-type}. The Fourier transform can therefore be deduced from this limiting procedure. We consider the limit of $\Im \widetilde{A}_{\rm c}(\cdot + \ri \eta)$, the real part of $\widetilde{A}_{\rm c}(\cdot + \ri \eta)$ being obtained from~\eqref{eq:Plemelj_formulae} and Definition~\ref{def:HT-PV}.

Let $f\in \cH$ and $\varphi \in \sS(\R_\omega)$. Then, using Fubini's theorem,
\begin{equation}
\label{eq:expression_A_c}
\left \langle \left\langle f\left| \Im \widetilde{A_{\rm c}}(\cdot+\ri\eta)\right|f\right\rangle, \varphi \right\rangle_{\sS',\sS} = -\int_\R \int_\R \frac{\eta}{(\omega-\lambda)^2  + \eta^2} \, \varphi(\omega) \, \mu_{f}^H(\rd \lambda) \, \rd\omega
= -\int_\R t_\eta(\lambda) \, \mu_{f}^H(\rd \lambda),
\end{equation}
where the measure $\mu_f^H$ is defined by $\mu_f^H(b) = \langle f \left|P_b^H\right|f\rangle$ for any $b \in \mathscr{B}(\R)$, and 
\[
t_\eta(\lambda) = \int_\R \frac{\eta}{(\omega- \lambda)^2 + \eta^2} \, \varphi(\omega) \, \rd \omega = \int_\R \frac{1}{\xi^2 + 1} \, \varphi(\lambda + \eta\xi) \, \rd \xi.
\]
Note that 
\[
\left|t_\eta(\lambda) - \pi \varphi(\lambda)\right| \leq \int_\RR \frac{1}{\xi^2 + 1} \left|\varphi(\lambda + \eta\xi) - \varphi(\lambda) \right| \rd \xi \leq \sqrt{\eta} \| \varphi'\|_{L^2} \int_\RR \frac{\sqrt{\xi}}{1+\xi^2} \, \rd\xi,
\]
where the last bound is obtained by rewriting $\varphi(\lambda + \eta\xi) - \varphi(\lambda)$ as the integral of its derivative and using a Cauchy-Schwarz inequality. This also shows that $t_\eta$ is uniformly bounded as $\eta \to 0^+$. Since the measure $\mu_f^H$ is finite, \eqref{eq:expression_A_c} leads by dominated convergence to
\[
\left \langle \left\langle f\left| \Im \widetilde{A_{\rm c}}(\cdot+\ri\eta)\right|f\right\rangle, \varphi \right\rangle_{\sS',\sS} \xrightarrow[\eta\to 0]{} -\pi \int_\RR \varphi(\lambda) \, \mu_{f}^H(\rd \lambda),
\]
which shows that $\Im \widehat{A_{\rm c}} = -\pi \, P^H$.

\subsection{Proof of Lemma~\ref{lemma:Im_positive_particle}}
\label{sec:proof:lemma:Im_positive_particle}

Let us first assume that $\Im \widehat{T}_{\rm c} \geq 0$. The aim is to prove that $\Re \widehat{T}_{\rm c} \geq 0$ on~$(-\infty,\omega_0]$. Consider to this end $\varphi \in \sS(\R)$ with $\mathrm{Supp}(\varphi) \subset (-\infty,\omega_0]$ and $\varphi \geq 0$. Then, for any $\omega \geq \omega_0$ and $\omega' \le 0$, it holds $\varphi(\omega-\omega') = 0$, so that
  \begin{equation}
    \label{eq:positivity_hilbert_transform}
    \forall \omega \geq \omega_0, \qquad (\fH\varphi)(\omega) = \lim_{\varepsilon \to 0^+} \int_{\R \backslash [-\varepsilon,\varepsilon]} \frac{\varphi(\omega-\omega')}{\omega'} \, \rd \omega' = \lim_{\varepsilon \to 0^+} \int_\varepsilon^{+\infty} \frac{\varphi(\omega-\omega')}{\omega'} \, \rd \omega' \geq 0.
  \end{equation}
Let $f \in \cH$. In view of~\eqref{eq:Plemelj_formulae} and Lemma~\ref{lem:fH_duality},
\[
\begin{aligned}
\left\langle \left\langle f \left| \Re \widehat{T}_{\rm c}\right| f \right\rangle, \varphi \right\rangle_{H^{-1},H^1} & =
- \left\langle \left\langle f \left| \fH\left(\Im \widehat{T}_{\rm c}\right) \right| f \right\rangle, \varphi \right\rangle_{H^{-1},H^1} =
-\left\langle \fH \left( \left\langle f \left| \Im  \widehat{T}_{\rm c} \right| f \right\rangle\right), \varphi \right\rangle_{H^{-1},H^1} \\
& = \left\langle \left\langle f \left| \Im  \widehat{T}_{\rm c} \right| f \right\rangle, \fH\varphi \right\rangle_{H^{-1},H^1}.
\end{aligned}
\]
The latter quantity is non-negative since $\Im \widehat{T}_{\rm c} \geq 0$ and $\fH \varphi \geq 0$ on~$\mbox{\rm Supp}(\Im \widehat{T_{\rm c}}) \subset [\omega_0,+\infty)$ (by~\eqref{eq:positivity_hilbert_transform}).

Let us now assume that $\Re \widehat{T}_{\rm c} \geq 0$ on~$(-\infty,\omega_0]$. The aim is to prove that $\Im \widehat{T}_{\rm c} \geq 0$ on the support of this distribution, which is included in $[\omega_0,+\infty)$. Consider therefore $\varphi \in \sS(\R)$ with $\mathrm{Supp}(\varphi) \subset [\omega_0,+\infty)$ and $\varphi \geq 0$. Note that
\[
\forall \omega \leq \omega_0, \qquad (\fH\varphi)(\omega) = \lim_{\varepsilon \to 0^-} \int_{-\infty}^\varepsilon \frac{\varphi(\omega-\omega')}{\omega'} \, \rd \omega' \leq 0,
\]
and, for any $f \in \cH$,
\[
\left\langle \left\langle f \left| \Im \widehat{T}_{\rm c}\right| f \right\rangle, \varphi \right\rangle_{H^{-1},H^1} = -\left\langle \left\langle f \left| \Re  \widehat{T}_{\rm c} \right| f \right\rangle, \fH\varphi \right\rangle_{H^{-1},H^1} \geq 0.
\]
This gives the desired conclusion.
 
\subsection{Proof of Lemma \ref{lem:L2Linfty_kernel}} 
\label{ssec:proof_L2Linfty_kernel}

The fact that $B_1 \in \cB(L^1(\R^2),L^2(\R^3))$ is a simple consequence of the following inequality: for $\varphi \in L^1(\R^3)$, it holds, for almost all $\br \in \R^3$, 
\[
|B_1 \varphi(\br)| = \left| \int_{\R^3} B_1 (\br, \br') \varphi(\br') \rd \br' \right| \leq \| B_1 (\br,\cdot)\|_{L^\infty(\R^3)} \, \| \varphi\|_{L^1(\R^3)}. 
\]
This shows that $B_1 \varphi \in L^2(\R^3)$ with 
\[
\|B_1 \varphi\|_{L^2(\R^3)} \leq \left( \int_{\R^3} \| B_1(\br,\cdot)\|^2_{L^\infty(\R^3)} \,\rd\br\right)^{1/2} \| \varphi\|_{L^1(\R^3)}.
\]

Now, for $f \in L^2(\R^3)$, it is easy to see that $B_1 f$ is an integral operator with kernel $B_1(\br, \br')f(\br')$. In addition,
\[
\begin{aligned} 
  \| B_1 f\|_{\fS_2(L^2(\R^3))}^2 & 
  = \int_{\R^3}\int_{\R^3} \left| B_1(\br, \br') f(\br')\right|^2 \rd \br \, \rd \br' \leq \int_{\R^3} \int_{\R^3} \| B_1(\br, \cdot)\|_{L^\infty(\R^3)}^2 |f(\br')|^2 \, \rd \br \, \rd\br'\\ 
  & = \left(\int_{\R^3} \| B_1(\br, \cdot) \|^2_{L^2(\R^3)} \rd\br \right) \| f \|_{L^2(\R^3)}^2.
\end{aligned}
\]
This gives the claimed result.

\subsection{Proof of Theorem~\ref{thm:analytic_continuation_kernel_product}} 
\label{sec:thm:analytic_continuation_kernel_product}
 
Fix $0 < \eta < \omega$ and $f, g \in \cH$. We start from~\eqref{eq:laplace_transform_C+}, which we rewrite as
\begin{equation} 
  \label{eq:initial_laplace_tranform}
  \widetilde{C^+}(\nu + \ri\omega) = \frac{\ri}{2\pi}\int_{-\infty}^{+\infty} \widetilde{A^+}\big(\nu+\nu'-\omega'+\ri(\omega-\eta)\big)\odot \widetilde{B^-}(\nu'-\omega'-\ri\eta) \, \rd \omega'.
\end{equation}
By Proposition~\ref{prop:resolvent} and Proposition~\ref{prop:resolvent_anti-causal},
\begin{equation}
  \label{eq:widetilde_A+_B-}
  \widetilde{A^+}(z) = A_1^* (z-A_2)^{-1}A_1, \qquad \widetilde{B^-}(z) = B_1^* (z+B_2)^{-1}B_1.
\end{equation}
The poles of $z \mapsto \widetilde{A^+}\big(\nu+\nu'+\ri (\omega-\eta)-z\big)$ are located on the half-line $\ri(\omega-\eta) + (-\infty,\nu+\nu'-a)$, while those of $z \mapsto \widetilde{B^-}(\nu' - \ri\eta - z)$ are located on the half-line $-\ri \eta + (b + \nu',+\infty)$. For any closed contour not enclosing any point of those two half-lines, the integral of $\widetilde{A^+}\big(\nu+\nu'+\ri (\omega-\eta)-z\big)\odot \widetilde{B^-}(\nu' - \ri\eta - z)$ on this contour vanishes. Let us choose the contour $\mathscr{C}_L$ plotted in Figure~\ref{fig:contour_analytic_SigmaGW} and evaluate the contributions of the left-hand side of
\begin{equation}
\label{eq:integral_to_compute_analytic_SigmaGW}
\oint_{\mathscr{C}_L} \Tr_{\cH} \left[ \widetilde{A^+}\big(\nu+\nu'+\ri(\omega-\eta)-z\big)g \widetilde{B^-}(\nu'-\ri\eta-z) \overline{f} \right]\rd z = 0,
\end{equation}
on the various segments. Recall that we choose $\nu < a+b$ and $\nu' \in (-b,a-\nu)$, so that $\nu+\nu'-a<0<\nu'+b$. Let us also emphasize that the operators appearing in the integrand do not have singularities.

\begin{figure}[h]
	\begin{center}
  \begin{tikzpicture}[scale =1]
    \draw (-7,0) -- (7,0);
    \draw (0, -5) -- (0, 5);
    \draw[red, very thick] (2, -0.2) -- (7, -0.2);
    \node[red] at (4, -0.5) {poles of $\widetilde{B^-}(\nu' - \ri\eta - \cdot)$};
    \draw[red] (2, -0.4) -- (2, 0.);
    \draw[red] (1.2, -0.4) -- (0.8, 0.);
    \draw[red] (0.8, -0.4) -- (1.2, 0.);
    \draw[red] (1.5, -0.4) -- (1.9, 0.);
    \draw[red] (1.9, -0.4) -- (1.5, 0.);
    \draw (1, -0.2) -- (1, 0.2);
    \node at (1.05, -0.7) {$\nu'+b$};
    \draw[red, very thick] (-2, 2) -- (-7, 2);
    \node[red] at (-4, 2.5) {poles of $\widetilde{A^+}(\nu+\nu' + \ri (\omega-\eta) - \cdot)$};
    \draw[red] (-2, 1.8) -- (-2, 2.2);
    \draw[red] (-1, 1.8) -- (-1.4, 2.2);
    \draw[red] (-1.4, 1.8) -- (-1, 2.2);
    \draw (-1.2, -0.2) -- (-1.2, 0.2);
    \node at (-1.6, -0.5) {$\nu+\nu'-a$};	
    \draw(-0.2, 2) -- (0.2, 2) node[right] {$\ri (\omega-\eta)$};
    \draw[blue, very thick] (-6, 0.) --(6, 0.) -- (6, 4.2) -- (0, 4.2) -- (0, -4.2) -- (-6, -4.2) -- (-6, 0.);
    \draw[<->]  (6.5, -0.2) -- (6.5, 0) node[right] {$\eta$};
    \draw (6,0.1) -- (6, -0.3) node[below] {$L$};
    \draw (-6,0.1) -- (-6, -0.1);
    \node at (-6, 0.6) {$-L$};
    \node at (0.4, 4.6) {$L$};
    \node at (0.4, -4.2) {$-L$};
    \node[red, very thick] at (6.4, 2) {$\mathscr{C}_L$};
    
    \draw[blue] (2.8, 0.2) -- (3,0)  -- (2.8, -0.2);
    \draw[blue] (5.8, 2.8) -- (6, 3) -- (6.2, 2.8);
    \draw[blue] (3, 4.4) -- (2.8, 4.2) -- (3, 4);
    \draw[blue] (-0.2, 1.2) -- (0, 1) -- (0.2, 1.2);
    \draw[blue] (-3, -4.4) -- (-3.2, -4.2) -- (-3, -4);
    \draw[blue] (-6.2, -2) -- (-6, -1.8) -- (-5.8, -2);
    \draw[blue] (-3, 0.2) -- (-2.8, 0) -- (-3, -0.2);
  \end{tikzpicture}
  
  \caption{Contour $\mathscr{C}_L$ used to compute the integral~\eqref{eq:integral_to_compute_analytic_SigmaGW}, with $\eta>0$ small compared to~$\omega$. Note that the condition $\nu + \nu' - a < 0 < \nu' + b$ ensures that the central vertical part of the contour does not intersect the poles of the functions in the integrand.}
  \label{fig:contour_analytic_SigmaGW}
  
  \end{center}
\end{figure}
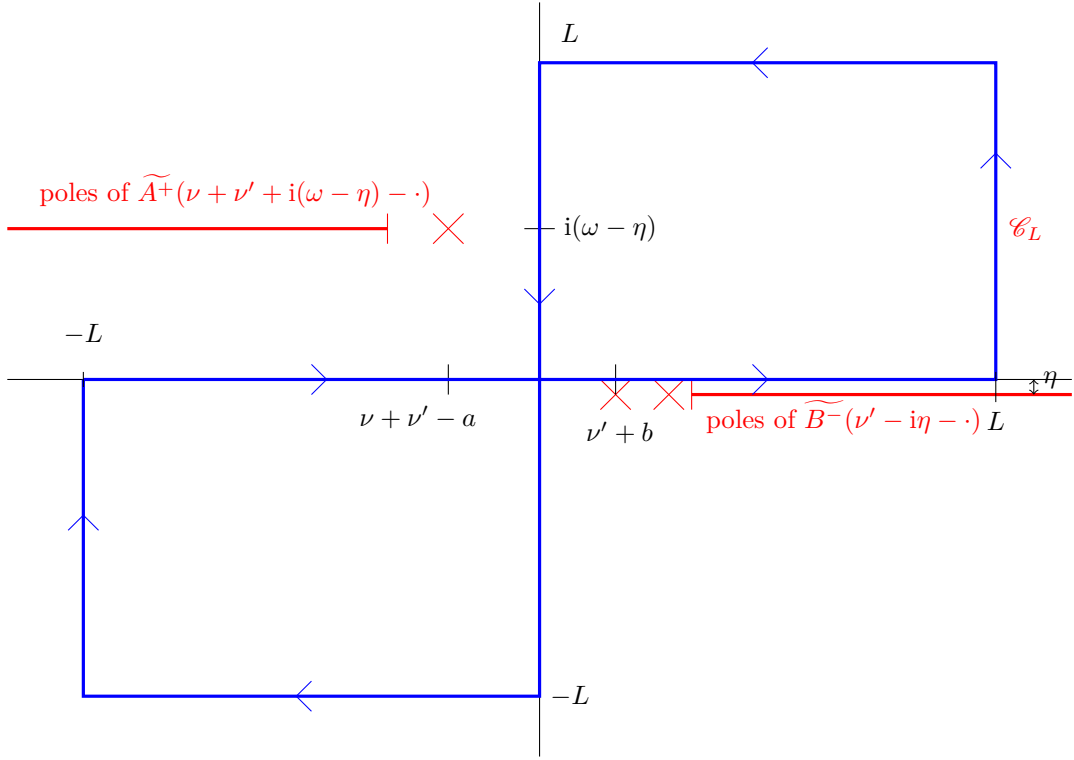

Let us first consider the part of the integral corresponding to the right side of the contour. Using~\eqref{eq:widetilde_A+_B-}, we obtain that for all $\omega' \in [0,L]$,
\[
\begin{aligned}
  & \left| \int_0^L \Tr_{\cH} \left[ \widetilde{A^+}\big(\nu+\nu' - L + \ri (\omega-\eta-\omega')\big) g \widetilde{B^-}\big(\nu'-L - \ri (\omega'+\eta)\big)\overline{f} \right] \rd \omega' \right| \\
  & \le \int_0^L \left|  \Tr_{\cH} \left[ A_1^* \dfrac{1}{\nu+\nu' - L + \ri (\omega-\eta-\omega') - A_2} A_1 g B_1^* \dfrac{1}{\nu'-L - \ri (\omega'+\eta) + B_2} B_1 \overline{f} \right] \right| \rd \omega' \\
  & \leq \left\| B_1 \overline{f} \right\|_{\fS_2(\cH,\cH_b)} \left\| B_1 \overline{g} \right\|_{\fS_2(\cH,\cH_b)} \| A_1 \|_{\cB(\cH,\cH_a)}^2  \\
	& \qquad \times \int_0^L \left\| \dfrac{1}{\nu+\nu' - L + \ri (\omega-\eta-\omega') - A_2}  \right\|_{\cB(\cH_a)} \left\| \dfrac{1}{\nu'-L - \ri (\omega'+\eta) +B_2 } \right\|_{\cB(\cH_b)}  \rd \omega' \\
  & \le C \| f \|_{\cH} \| g \|_{\cH} \dfrac{1}{L} \int_0^L \dfrac{\rd \omega'}{\omega' + \eta} = C \| f \|_\cH \| g \|_\cH \dfrac{1}{L} \log \left( \dfrac{L + \eta}{\eta} \right),
\end{aligned}
\]
where we have used $A_2 - (\nu+\nu') + L \geq a - (\nu + \nu') + L \geq L$. Similar estimates can be stated for the upper, lower and left parts of the contour. For the upper part for instance, for which the integration is performed from $z = L+\ri L$ to $z = \ri L$, we get
\begin{align*}
  & \left| \int_{L}^0 \Tr_{\cH} \left[ \widetilde{A^+}\big(\nu+\nu'-\omega'+\ri(\omega-\eta-L)\big)g \widetilde{B^-}\big(\nu'-\omega'-\ri(\eta+L)\big) \overline{f} \right] \rd \omega' \right| \\
  & \le  C \| f \|_{\cH} \| g \|_{\cH} \int_0^L \left\| \dfrac{1}{\nu+\nu' -\omega' + \ri (\omega-\eta-L) - A_2}  \right\|_{\cB(\cH_a)} \left\| \dfrac{1}{\nu'-\omega' + \ri (L+\eta) - B_2} \right\|_{\cB(\cH_b)}  \rd \omega' \\
  & \le C \| f \|_{\cH} \| g \|_{\cH} \int_0^L \left( \dfrac{1}{L+\eta} \right) \dfrac{1}{\omega' + a - (\nu+\nu')} \rd \omega' \\
  & = C \| f \|_{\cH} \| g \|_{\cH} \dfrac{\log \big(L+a - (\nu+\nu')\big) - \log \big(a - (\nu+\nu')\big)}{L+\eta},
\end{align*}
where we recall that $a - (\nu+\nu') > 0$. We then take the limit $L \to +\infty$, so that the contributions to the integral which are not on the imaginary axis or on the real axis vanish. We deduce that
\[
\begin{aligned}
& \int_{-\infty}^{\infty} \Tr_{\cH} \left[ \widetilde{A^+}\big(\nu+\nu'-\omega'+\ri(\omega-\eta)\big)g \widetilde{B^-}(\nu'-\omega'-\ri\eta) \overline{f} \right] \rd \omega' \\
& \qquad = \ri \int_{-\infty}^{+\infty}  \Tr_{\cH} \left[ \widetilde{A^+}\big(\nu+\nu'+\ri(\omega-\eta+\omega')\big)g \widetilde{B^-}\big(\nu'+\ri(\omega'-\eta)\big)  \overline{f} \right] \rd \omega' \\
& \qquad = \ri \int_{-\infty}^{+\infty}  \Tr_{\cH} \left[ \widetilde{A^+}\big(\nu+\nu'+\ri(\omega+\omega')\big)g \widetilde{B^-}(\nu'+\ri \omega')  \overline{f} \right] \rd \omega'. \\
\end{aligned}
\]
In view of~\eqref{eq:initial_laplace_tranform}, we finally obtain that
\begin{equation}
\label{eq:final_analytic_continuation_C+}
\widetilde{C^+}(\nu + \ri \omega) = -\dfrac{1}{2 \pi} \int_{-\infty}^{+\infty} \widetilde{A^+}\big(\nu+\nu'+\ri(\omega+\omega')\big) \odot \widetilde{B^-}(\nu'+\ri \omega') \, \rd \omega'.
\end{equation}
We next note that our choices for $\nu,\nu'$ ensure that the expressions on both sides are analytic for all $\omega > 0$, and can be extended analytically to all $\omega \in \R$. Therefore, the above equality also holds true for $\omega \le 0$.

In a similar fashion, we prove that, for all $\omega \in \R$,
\[
\widetilde{C^-}(\nu + \ri \omega) = -\dfrac{1}{2 \pi} \int_{-\infty}^{+\infty} \widetilde{A^-}\big(\nu+\nu' + \ri (\omega +\omega')\big) \odot \widetilde{B^+}(\nu'+\ri \omega') \rd \omega'.
\]
This equality is established as for~\eqref{eq:final_analytic_continuation_C+} by considering $\widetilde{C^-}(\nu - \ri \omega)$ for $\omega > 0$ and evaluating the various parts of the left-hand side of
\[
\oint_{\mathscr{C}_L} \Tr_{\cH} \left[ \widetilde{A^-}\big(\nu+\nu'-\ri(\omega-\eta)-z\big)g \widetilde{B^+}(\nu'+\ri\eta-z) \overline{f} \right]\rd z = 0.
\]
The poles of the integrand are on the half-lines $-\ri(\omega-\eta) + (\nu+\nu'+a,+\infty)$ and $\ri\eta + (-\infty,-b+\nu')$. The conditions $\nu > -(a+b)$ and $-a-\nu<\nu'<b$ ensure that $-b+\nu'<0<\nu+\nu'+a$, so that the integrand has no singularity on the imaginary axis.

Finally, since $A^+(\tau) \odot B^+(-\tau) = A^-(\tau) \odot B^-(-\tau) = 0$ for $\tau \neq 0$, we can concatenate $\widetilde{C^+}(\nu + \ri \omega)$ and $\widetilde{C^-}(\nu + \ri \omega)$, and obtain
\[
\widetilde{C}(\nu + \ri \omega) = -\dfrac{1}{2 \pi} \int_{-\infty}^{+\infty} \widetilde{A}\big(\nu+\nu' + \ri (\omega +\omega') \big) \odot \widetilde{B}(\nu' + \ri \omega') \, \rd \omega'.
\]

\subsection{Proof of Corollary~\ref{corr:analytic_C}}
\label{sec:corr:analytic_C}

The proof is based on the representation~\eqref{eq:general_formula_analytic_continuation_C+} with the choice $\omega = 0$. Consider $\nu < a+b$. It holds
\[
\begin{aligned}
\widehat{C^+}(\nu) = \widetilde{C^+}(\nu) 
& = -\dfrac{1}{2 \pi} \int_{-\infty}^{+\infty} \widetilde{A^+}(\nu+\nu' + \ri \omega') \odot \widetilde{B^-}(\nu' + \ri \omega') \, \rd \omega' \\
& = - \dfrac{1}{2\pi} \int_{-\infty}^{+\infty} \left[ A_1^*\frac{\nu+\nu'-A_2 - \ri\omega'}{(\nu+\nu'-A_2)^2+(\omega')^2}A_1 \right] \odot \left[ B_1^* \frac{\nu'+B_2-\ri\omega'}{(\nu'+B_2)^2+(\omega')^2}B_1 \right] \rd \omega'.
\end{aligned}
\]
The odd terms in $\omega'$ cancel out by symmetry, so that
\begin{equation}
\label{eq:widehat_C+_is_real}
\begin{aligned}
\widehat{C^+}(\nu) & 
= \dfrac{1}{2\pi} \int_{-\infty}^{+\infty} \left[ A_1^*\frac{A_2-(\nu+\nu')}{(\nu+\nu'-A_2)^2+(\omega')^2}A_1 \right] \odot \left[ B_1^* \frac{\nu'+B_2}{(\nu'+B_2)^2+(\omega')^2}B_1 \right] \rd \omega' \\
& \ \ + \dfrac{1}{2\pi} \int_{-\infty}^{+\infty} \left[ A_1^*\frac{\omega'}{(\nu+\nu'-A_2)^2+(\omega')^2}A_1 \right] \odot \left[ B_1^* \frac{\omega'}{(\nu'+B_2)^2+(\omega')^2}B_1 \right] \rd \omega'.
\end{aligned}
\end{equation}
This shows that this operator is positive and self-adjoint in view of Lemmas~\ref{lem:positivity_kernel_product} and~\ref{lem:adjoint_kernel_product}. As a result, $\Im \widehat{C^+}(\nu) = 0$ for $\nu < a+b$. This proves the first assertion in~\eqref{eq:Sigma_limit_real}. Also, we get from~\eqref{eq:widehat_C+_is_real} that $\widehat{C^+}(\nu) = \Re \widehat{C^+} \ge 0$ for $\nu < a+b$. Together with Lemma~\ref{lemma:Im_positive_particle}, this shows the first assertion of~\eqref{eq:imaginary_widehat_C}. The results concerning $\Im \widehat{C^-}$ are proved in a similar way.

\subsection{Proof of Proposition~\ref{prop:Psi_N^0}}
\label{sec:proof:prop:Psi_N^0}

The first assertion follows from the fact that the domain of $H_N$ is $H^2(\R^{3N})$ and that $H_N \Re (\Psi_N^0) = E_N^0 \Re (\Psi_N^0)$ and $H_N \Im (\Psi_N^0) = E_N^0 \Im ( \Psi_N^0)$.

The density $\rho_N^0$ is bounded since it decreases exponentially fast away from the nuclei and is continuous~\cite{Fournais2002}.

\medskip

In order to prove~\eqref{eq:recover_gamma_with_creation_annihilation}, we rely on~\eqref{eq:annihilation_op} and~\eqref{eq:def-DM} in order to write (recall that $\Psi_N^0$ is real valued)
\[
\begin{aligned}
& \left\langle \Psi_N^0 | a^\dagger(g) a(f) |\Psi_N^0 \right\rangle_{\cH_N} 
= \left\langle a(g)\Psi_N^0 \big| a(f) \Psi_N^0 \right\rangle_{\cH_N} \\
& \quad = N \int_{\R^{3(N-1)}} \left(\int_{\R^3} g(\br)\Psi_N^0(\br,\br_1,\dots,\br_{N-1}) \, \rd\br\right) \left(\int_{\R^3} \overline{f(\br')} \Psi_N^0(\br',\br_1,\dots,\br_{N-1}) \, \rd\br' \right) \rd \br_1\dots \rd \br_{N-1} \\
& \quad = \int_{\R^3} \int_{\R^3} {g(\br)} \gamma_N^0(\br,\br') \overline{f(\br')} \, \rd\br \, \rd\br' = \langle f | \gamma_N^0| g \rangle.
\end{aligned}
\]
To bound the kernel $| \gamma_N^0(\br, \br')|^2$, we write
\begin{align*}
  | \gamma_N^0(\br, \br')|^2 & = N^2 \left| \int_{\R^3} \ldots \int_{\R^3} {{\Psi_N^0(\br, \br_2, \ldots \br_N)}\Psi_N^0(\br', \br_2, \ldots \br_N) } \rd \br_2 \ldots \rd \br_N \right|^2 \\
  & \le  N^2 \left( \int_{\R^3} \ldots \int_{\R^3} \left| \Psi_N^0(\br, \br_2, \ldots \br_N) \right|^2  \rd \br_2 \ldots \rd \br_N  \right)\left( \int_{\R^3} \ldots \int_{\R^3} \left| \Psi_N^0(\br', \br_2, \ldots \br_N) \right|^2  \rd \br_2 \ldots \rd \br_N  \right) \\
  & \le \rho_N^0(\br) \rho_N^0(\br'),
\end{align*}
where we used the Cauchy-Schwarz inequality.

\medskip

Let us finally recall why $\rho_{N,2}^0$ defines a bounded integral operator. Note first that $\rho_{N,2}^0(\br,\br') \geq 0$. For $f,g \in \cH_1$, the Cauchy-Schwarz inequality then leads to
\[
\begin{aligned}
\left| \langle f | \rho_{N,2}^0 | g \rangle \right| & = \left| \int_{\R^3} \int_{\R^3} \overline{f(\br)} \rho_{N,2}^0(\br,\br') g(\br') \, \rd\br \, \rd\br' \right| \\
& \leq \left(\int_{\R^3} \int_{\R^3} |f(\br)|^2 \rho_{N,2}^0(\br,\br') \, \rd\br \, \rd\br' \right)^{1/2}\left(\int_{\R^3} \int_{\R^3} |g(\br')|^2 \rho_{N,2}^0(\br,\br') \, \rd\br \, \rd\br' \right)^{1/2} \\
& = \frac{(N-1)}{2} \left( \int_{\R^3} |f|^2 \, \rho_N^0 \right)^{1/2} \left( \int_{\R^3} |g|^2 \, \rho_N^0 \right)^{1/2} \\
& \leq \frac{(N-1)}{2} \left\| \rho_N^0 \right\|_{L^\infty} \| f \|_{\cH_1} \| g \|_{\cH_1}.
\end{aligned}
\]
This shows that $\rho_{N,2}^0$ defines a bounded operator on~$\cH_1$, with operator norm lower or equal to $(N-1)\left\| \rho_N^0 \right\|_{L^\infty}/2$.


\subsection{Proof of Theorem~\ref{thm:GM_formula}} 
\label{sec:proof:thm:GM_formula}

Since
\[
A_- f(\br_1,\dots,\br_{N-1}) = \sqrt{N} \int_{\R^3} f(\br) \Psi_N^0(\br,\br_1,\dots,\br_{N-1})\, \rd\br,
\] 
and introducing
\[
\Delta_{N-1} = \sum_{i=1}^{N-1} \Delta_{\br_i},
\]
it is easily seen that $(1-\Delta_{N-1}) A_-$ is an integral operator with kernel $\left[(1-\Delta_{N-1}) \Psi_N^0\right](\br_1,\dots,\br_{N-1}; \br)$. As $\Psi_N^0 \in H^2(\R^{3N})$, it follows that $(1-\Delta_{N-1}) A_- \in \fS_2(\cH_1,\cH_{N-1})$. Therefore, any operator of the form $A_-^\ast B A_-$, where the operator $B$ on~$\cH_{N-1}$ is such that $(1-\Delta_{N-1})^{-1/2} B (1-\Delta_{N-1})^{-1/2} \in \cB(\cH_{1})$, is trace-class. In particular, the operator
\[
\partial_\tau G_{\rm h}(\tau) \big|_{\tau=0^-} = - A_-^\ast (H_{N-1} - E_N^0 ) A_-,
\]
is trace-class.

Let us now compute more explicitly the action of this operator. Let 
\[
	h_1 := - \dfrac12 \Delta + v_\ext.
\]
We use the definition~\eqref{eq:annihilation_op} of $a(\overline{f})$, and obtain
\begin{align}
  \left( \sum_{i=1}^{N-1}h_1 (\br_i) \right) \left[ a(\overline{f})\Psi_{N}^0 \right](\br_1, \ldots, \br_{N-1})
  & = \sqrt{N} \int_{\mathbb{R}^{3}} {f}(\br_{N}) \left( \sum_{i=1}^{N-1} h_1(\br_{i}) \right)\Psi_{N}^0( \br_{1}, \ldots, \br_{N}) \, \rd \br_{N} \nonumber \\
  & = \sqrt{N} \int_{\mathbb{R}^{3}} {f}(\br_{N}) \left(H_{0,N}\Psi_{N}^0\right)( \br_{1}, \ldots, \br_{N}) \, \rd \br_{N} \label{eq:commutation_A_h} \\
  & \ - \sqrt{N} \int_{\mathbb{R}^{3}} (h_1f)(\br_{N}) \Psi_{N}^0( \br_{1}, \ldots, \br_{N}) \, \rd \br_{N}, \nonumber 
\end{align}
so that 
\begin{align*}
	\left( H_{N-1} A_- f\right) (\br_1, \ldots, \br_{N-1}) & = E_N^0 \left(A_- f\right)(\br_1, \ldots, \br_{N-1}) - \left(A_- h_1 f\right)(\br_1, \ldots, \br_{N-1})  \\
		& \quad - \sqrt{N} \sum_{i=1}^{N-1} \int_{\R^3} f(\br)\frac{\Psi_N^0(\br,\br_1,\dots,\br_{N-1})}{|\br-\br_i|} \, \rd\br. 
\end{align*}
Moreover, it is easily seen that, for any $\Phi_{N-1} \in \cH_{N-1}$,
\[
\left(A_-^\ast \Phi_{N-1}\right)(\br) = \sqrt{N} \int_{\R^{3(N-1)}} {\Psi_N^0(\br,\br_1,\dots,\br_{N-1})} \Phi_{N-1}(\br_1,\dots,\br_{N-1}) \, \rd\br_1\dots\rd\br_{N-1}.
\]
Therefore, 
\[
\left(A_-^\ast \left(H_{N-1}-E_N^0\right) A_- f\right)(\br) = - \left(\gamma_N^0 h_1 f\right)(\br) -  \int_{\R^{3}} K_N(\br,\br') f(\br') \, \rd\br',
\]
with
\[
K_N(\br,\br') = N \sum_{i=1}^{N-1} \int_{\R^{3(N-1)}} \frac{{\Psi_N^0(\br,\br_1,\dots,\br_{N-1})}\Psi_N^0(\br',\br_1,\dots,\br_{N-1})}{|\br'-\br_i|} \, \rd\br_1\dots\rd\br_{N-1}.
\]
Since we already know that the integral operator $K_N$ on $\cH_1$, with kernel $K_N(\br, \br')$ is trace-class, and that $\Psi_N^0$ is continuous and decays exponentially fast (see \textit{e.g.}~\cite{Fournais2002}), we have~\cite[Theorem A.2]{Simon2005}
\[
\begin{aligned}
\Tr_{\cH_1} (K_N) = \int_{\R^3} K_N(\br,\br) \, \rd\br & = N(N-1) \int_{\R^{3(N-1)}} \frac{\left|\Psi_N^0(\br,\br_1,\dots,\br_{N-1})\right|^2}{|\br-\br_1|} \, \rd\br_1\dots\rd\br_{N-1} \\
& = 2 \int_{\R^3}\int_{\R^3} \dfrac{\rho_{N,2}^0(\br, \br')}{| \br - \br' |} \ \rd \br \ \rd \br',
\end{aligned}
\]
where we recall that $\rho_{N,2}^0$ is the two-body density matrix defined in~\eqref{def:2body}. 

Finally, 
\[
\begin{aligned}
  \Tr_{\cH_1} \left( \partial_\tau G_{\rm h}(\tau) \big|_{\tau=0^-} \right) & = - \Tr_{\cH_1} \left( A_-^\ast(H_{N-1} - E_N^0)A_- \right) \\
  & =  \Tr_{\cH_1} \left( \gamma_N^0 h_1 \right) + 2\int_{\R^3}\int_{\R^3} \dfrac{\rho_{N,2}^0(\br, \br')}{| \br - \br' |} \ \rd \br \ \rd \br',
\end{aligned}
\]
which gives the claimed result in view of the following representation of the ground state energy:
\begin{equation} 
  \label{eq:EN0}
  E_N^0 = \bra \Psi_N^0 | H_N | \Psi_N^0 \ket = \Tr_{\cH_1} \left( h_1 \gamma_N^0 \right) + \int_{\R^3}\int_{\R^3} \dfrac{\rho_{N,2}^0(\br, \br')}{| \br - \br' |} \ \rd \br \ \rd \br'.
\end{equation}

\subsection{Proof of Lemma~\ref{lem:B0}} 
\label{sec:proof:lem:B0}

Proposition~\ref{prop:Psi_N^0} shows that $\rho_N^0 \in L^p(\R^3)$ for $1 \le p \le +\infty$. This implies that $\left( \sum_{i=1}^N v(\br_i) \right) \Psi_N^0$ belongs to~$\cH_N$ for $v \in \cC'$ since, by H\"older's inequality, and the inequality $\left( \sum_{i=1}^N v(\br_i) \right)^2 \le N \sum_{i=1}^N v(\br_i)^2$, it holds that
\begin{align*}
\int_{\R^{3N}} \left( \sum_{i=1}^N v(\br_i) \right)^2 |\Psi_N^0(\br_1,\dots,\br_N)|^2 \, \rd \br_1\dots\rd\br_N 
	& \le N \int_{\R^3} v^2 \rho_N^0 \leq  N \| v\|_{L^6}^2 \| \rho_N^0 \|_{L^{3/2}} \\
	& \leq N C_{\cC'} \| v\|_{\cC'}^2 \| \rho_N^0 \|_{L^{3/2}},
\end{align*}
where we have used the embedding $\cC' \hookrightarrow L^6(\R^3)$. Moreover, $\rho_N^0 \in L^{6/5}(\R^3) \hookrightarrow \cC$, so that $\left|\bra v, \rho_N^0 \ket_{\cC', \cC}\right| \le \| v \|_{\cC'} \| \rho_N^0 \|_{\cC}$. We therefore deduce that $B$ is a bounded operator from $\cC'$ to $\cH_N$, whose norm satisfies
\[
\left\| B \right\|_{\mathcal{B}(\cC',\cH_N)} \leq \sqrt{N C_{\cC'} \| \rho_N^0 \|_{L^{3/2}}} + \| \rho_N^0 \|_{\cC}.
\]
We finally have, for $v \in \cC'$,
\[
\left\bra \Psi_N^0 \left|\sum_{i=1}^N v(\br_i) \right| \Psi_N^0 \right\ket_{\cH_N} = \int_{\R^3} v(\br)  \rho_N^0(\br)\, \rd \br = \bra v, \rho_N^0 \ket_{\cC', \cC}\,
\]
from which we deduce that $\bra \Psi_N^0 | B v \ket_{\cH_N} = 0$. Since $v$ was arbitrary, we conclude that $B^* \Psi_N^0 = 0$. 

\subsection{Proof of Theorem~\ref{lem:sumrule_chi}} 
\label{sec:proof:lem:sumrule_chi}

Consider $f,g \in C^\infty_c(\R^3, \C)$ (that is $C^\infty$ with compact supports). Using the fact that $(H_N - E_N^0) | \Psi_N^0 \ket = 0$, we obtain
\begin{align*}
  \bra \overline{f}, \vc^{-1} B^* (H_N - E_N^0) B g \ket_{\cC',\cC} & = \bra f \left|B^* (H_N - E_N^0) B \right|g \ket_{\cC'} = \left\langle Bf \left| H_N - E_N^0 \right| Bg \right\rangle_{\mathcal{H}_N} \\
&  = \left\langle \Psi_N^0 \left| \left(\sum_{j=1}^N \overline{f(\br_j)}\right) (H_N - E_N^0) \left(\sum_{i=1}^N g(\br_i)\right) \right| \Psi_N^0 \right\rangle_{\mathcal{H}_N}.
\end{align*}
We next observe that
\begin{align*}
& (H_N - E_N^0) \left( \sum_{i=1}^N g(\br_i) \Psi_N^0(\br_1,\dots,\br_N) \right) = \\
& \quad =  \left( - \frac12 \sum_{i=1}^N \Delta g(\br_i) \right) \Psi_N^0(\br_1,\dots,\br_N) - \sum_{i=1}^N \nabla g(\br_i) \cdot \nabla_{\br_i}\Psi_N^0(\br_1,\dots,\br_N), 
\end{align*}
so that, using the fact that $\Psi_N^0$ is real-valued (see Proposition~\ref{prop:Psi_N^0}),
\begin{align*}
  \left\bra f \left| B^* (H_N - E_N^0) B  \right| g \right\ket_{\cC'} 
  & = \sum_{i,j=1}^N \int_{\mathbb{R}^{3N}} \overline{f(\br_j)} \left( - \dfrac12 \Delta g(\br_i) \right) | \Psi_N^0(\br_1, \ldots \br_N) |^2 \, \rd \br_1 \ldots \rd \br_N \\
  & \quad - \sum_{i,j=1}^N \int_{\mathbb{R}^{3N}} \overline{f(\br_j)} \Psi_N^0(\br_1, \ldots, \br_N) \nabla g (\br_i)  \cdot \nabla_{\br_i} \Psi_N^0(\br_1, \ldots, \br_N) \, \rd \br_1 \ldots \rd \br_N \\
  & =  \dfrac12 \sum_{i=1}^N \int_{\mathbb{R}^{3N}} \overline{\nabla f(\br_j)} \cdot \nabla g(\br_i)  | \Psi_N^0(\br_1, \ldots \br_N) |^2 \rd \br_1 \ldots \rd \br_N \\
  & \quad + \dfrac12 \sum_{i,j=1}^N \int_{\mathbb{R}^{3N}} \overline{f(\br_j)} \nabla g (\br_i) \cdot \nabla_{\br_i} \left( \left| \Psi_N^0(\br_1, \ldots, \br_N)\right|^2 \right) \rd \br_1 \ldots \rd \br_N \\
  & \quad - \sum_{i,j=1}^N \int_{\mathbb{R}^{3N}} \overline{f(\br_j)} \Psi_N^0(\br_1, \ldots, \br_N) \nabla g (\br_i) \cdot \nabla_{\br_i} \Psi_N^0(\br_1, \ldots, \br_N) \rd \br_1 \ldots \rd \br_N \\
  & =  \dfrac12 \int_{\mathbb{R}^{3}} \overline{\nabla f(\br)} \cdot \nabla g(\br) \rho_N^0(\br) \rd \br.
\end{align*}
We deduce that $2 \vc^{-1} B^* (H_N - E_N^0) B = 2 \vc^{-1} B^* (H_N^\sharp - E_N^0) B = -\div \left( \rho_N^0 \nabla \cdot \right)$ as operators on the core~$C^\infty_c(\R^3, \C)$. We next observe that $\div \left( \rho_N^0 \nabla \cdot \right)$ can be extended as a bounded operator from~$\cC'$ to~$\cC$. Indeed, for $f,g \in C^\infty_c(\R^3, \C)$,
\[
	\left| \left\bra g,\div \left( \rho_N^0 \nabla f\right) \right\ket_{\cC',\cC} \right| \leq \left\| \rho_N^0 \right\|_{L^\infty} \| \nabla f\|_{L^2} \| \nabla g\|_{L^2} = 4\pi \| \rho_N^0 \|_{L^\infty} \| f\|_{\cC'} \| g\|_{\cC'},
\]
which shows that 
\[
\left\| \div \left( \rho_N^0 \nabla \cdot \right) \right\|_{\mathcal{B}(\cC',\cC)} \leq 4\pi \left\| \rho_N^0 \right\|_{L^\infty}.
\]
Therefore, $2 \vc^{-1} B^* (H_N^\sharp - E_N^0) B = -\div \left( \rho_N^0 \nabla \cdot \right)$ as bounded operators from~$\cC'$ to~$\cC$.

For the second part of the proof, we first note that it is sufficient to check the convergence in the case when $f = g \in C^\infty_c(\R^3, \R)$. It holds:
\[
\begin{aligned}
  & \left\bra \overline{f}, \vc^{-1}B^* \dfrac{H_N^\sharp - E_N^0}{(H_N^\sharp - E_N^0)^2 + \omega^2} \omega^2 Bf \right\ket_{\cC',\cC} - \bra \overline{f}, \vc^{-1} B^* (H_N^\sharp - E_N^0) Bf \ket_{\cC',\cC}  \\
  & \qquad \qquad = -\left\bra \overline{f}, \vc^{-1} B^*  \left( \dfrac{ (H_N^\sharp -E_N^0)^3 }{(H_N^\sharp  - E_N^0)^2 + \omega^2} \right) Bf \right\ket_{\cC',\cC} = -\int_0^{+\infty} \left( \dfrac{\lambda^3}{\lambda^2 + \omega^2} \right) \rd\left\|  P_\lambda^{H_N^\sharp - E_N^0} (Bf) \right\|_{\cH_N}^2, 
\end{aligned}
\]
where we used the spectral decomposition of $H_N^\sharp - E_N^0$ in the last equality. The integrand of the last integral converges pointwise to $0$ when $|\omega|\to +\infty$. It is also non-negative and uniformly bounded by $\lambda \1_{[0,+\infty)}(\lambda)$, which is integrable since 
\[
\int_0^{+\infty} \lambda \, \rd \left\| P_\lambda^{(H_N^\sharp - E_N^0)} (Bf) \right\|_{\cH_N}^2 = \left\langle \overline{f}, \vc^{-1} B^* \left(H_N^\sharp - E_N^0\right) Bf \right\rangle_{\cC',\cC} = \dfrac12 \int_{\mathbb{R}^{3}} \rho_N^0 | \nabla f |^2 < \infty.
\]
The weak convergence therefore follows from the dominated convergence theorem.

To prove the strong convergence, we use the following rewriting for $g \in \cC'$:
\[
\vc^{-1} B^*  \left( \dfrac{ (H_N^\sharp -E_N^0)^3 }{(H_N^\sharp - E_N^0)^2 + \omega^2} \right) Bg = \vc^{-1} B^* A_\omega Mg,
\]
where
\[
A_\omega = \dfrac{ (H_N^\sharp -E_N^0)^2 }{(H_N^\sharp - E_N^0)^2 + \omega^2} 
\]
strongly converge to~0 on~$\cH_N$, and 
\[
Mg := (H_N^\sharp - E_N^0) B g =  \sum_{i=1}^N \left( -\frac12 \Delta g(\br_i) \Psi_N^0 - \nabla g(\br_i) \cdot \nabla_{\br_i} \Psi_N^0 \right).
\]
When $g \in \cC'$ is such that $\Delta g \in L^2$, it holds that $Mg \in \cH_N$ (by a proof similar to the one in Section~\ref{sec:proof:lem:B0}), which allows us to conclude.

\subsection{Proof of Proposition \ref{lemma:time-ordered_G0F}}
\label{sec:proof:lemma:time-ordered_G0F} 

We first prove~\eqref{eq:G0+h1} and~\eqref{eq:G0h1}. We start from an expression similar to the one provided by Lemma~\ref{lemma:time-ordered_GF}:
\[
\widetilde{G_0}(z)  := A_{0,+} \left( z - (H_{0,N+1} - E_{0,N}^0) \right)^{-1} A_{0,+}^\ast  + A_{0,-}^\ast \left( z + H_{0,N-1} - E_{0,N}^0 \right)^{-1} A_{0,-} .
\]
Then, we notice that, for $f \in \cH_1$, it holds
\begin{equation}
  \label{eq:commutator_A0dagger_h1}
  \sum_{i=1}^{N+1} h_1(\br_i) \left( A_{0,+}^\ast f \right) = \sum_{i=1}^{N+1} h_1(\br_i) \left( a^\dagger(f)  | \Phi_N^0 \ket \right) = a^\dagger (h_1 f)  \ | \Phi_N^0 \ket + a^\dagger(f)  \left| \left( \sum_{i=1}^{N} h_1(\br_i) \right) \Phi_N^0 \right\ket,
\end{equation}
or, equivalently,
\[
	H_{0,N+1} A_{0,+}^\ast (f) = A_{0,+}^\ast \left( h_1 f \right) + E_{0,N}^0 A_{0,+}^\ast (f);
\]
so that
\begin{equation} \label{eq:resolvent_particle}
  \left( z - (H_{0,N+1} - E_{0,N}^0) \right) A_{0,+}^\ast = A_{0,+}^\ast (z - h_1).
\end{equation}
Hence, the particle part of~\eqref{eq:G0+h1} is a consequence of the equality $A_{0, +} A_{0,+}^\ast = \mathds{1}_\cH - \gamma_{0,N}^0$ (similarly to~(\ref{eq:A+bounded})). To handle the hole part, we use computations similar to~\eqref{eq:commutation_A_h}, and find
\[
H_{0,N-1} A_{0,-} (f) = E_{0,N}^0 A_{0,-} (f) - A_{0,-} \left( h_1 f \right).
\]
We deduce that
\begin{equation} \label{eq:resolvent_hole}
  (z + H_{0,N-1} - E_{0,N}^0) A_{0,-} = A_{0,-} (z - h_1),
\end{equation}
and we conclude using the fact that $A_{0,-}^\ast A_{0,-}  = \gamma_{0,N}^0$.
Combining (\ref{eq:resolvent_particle}) and (\ref{eq:resolvent_hole}) leads to $\widetilde{G}_0(z) (z - h_1) = A_{0,-}^\ast A_{0,-} + A_{0,+} A_{0,+}^\ast = \mathds{1}_{\cH_1}$. Upon replacing $z$ by $\overline{z}$ and passing to adjoints, it also follows $(z - h_1) \widetilde{G}_0(z) = \mathds{1}_{\cH_1}$. This shows that $\widetilde{G_0}(z) = (z - h_1)^{-1}$. \\

To prove the first assertion of Proposition~\ref{lemma:time-ordered_G0F}, we notice that the operator-valued functions $\tau \mapsto - \ri \Theta(\tau) \left(\mathds{1}_{\cH_1} - \gamma_{0,N}^0 \right) \re^{- \ri \tau h_1}$ and $\tau \mapsto G_{0,\rp}(\tau)$ have the same Laplace transforms (see Proposition~\ref{prop:resolvent}). We conclude that the two operators coincide since the Laplace transform is one-to-one. The proof for $G_{0,\rh}$ is similar.
	
\subsection{Proof of Proposition~\ref{lemma:Hz}}
\label{sec:proof:lemma:Hz}

Let $z \in \C \setminus \R$. Let us first prove that $\widetilde G(z)$ is a one-to-one operator. Let $f \in \cH_1$ be such that $\widetilde G(z) f = 0$. From
\begin{align} \label{eq:ineq_Hz}
  \Im \left( \left\bra f \left| \widetilde G(z) \right| f \right\ket \right) = & \ \Im \left( \left\bra A_+^\ast f \left| \left( z - (H_{N+1} - E_{N}^0) \right)^{-1} \right| A_+^\ast f \right\ket_{\cH_{N+1}} \right. \\
  	& \quad +  \left.
  \left\bra A_- f \left| \left( z - (E_N^0 - H_{N-1}) \right)^{-1} \right| A_-f \right\ket_{\cH_{N-1}} \right) \nonumber \\
  = & - \Im(z) \int_\R \left( (\Re(z) + E_N^0 - \lambda)^2 + \Im(z)^2 \right)^{-1} \rd \left \| P_\lambda^{H_{N+1}} \left(A_+^\ast f \right)\right\|_{\cH_{N+1}}^2 \nonumber \\
  & - \Im(z) \int_\R \left( (\Re(z) - E_N^0 + \lambda)^2 + \Im(z)^2 \right)^{-1} \rd \left\| P_\lambda^{H_{N-1}} \left(A_- f \right)\right\|_{\cH_{N-1}}^2,
\end{align}
we deduce that both terms in the right-hand side must vanish. In particular, since $\Im(z) \neq 0$, it must hold $A_+^\ast f = A_-f = 0$. In view of the identity $A_-^\ast A_- + A_+ A_+^\ast = \mathds{1}_{\cH_1}$, this implies $f = 0$. Hence, $\widetilde G(z)$ is one-to-one. 

As a consequence, $\widetilde G(z)$ is an invertible operator from $\cH_1$ to its image $\widetilde D(z)$. Since $\left( \widetilde G(z) \right)^* = \widetilde G(\overline{z})$ is also one-to-one, $\widetilde D(z)$ is dense in $\cH_1$. 

Let us finally prove that $\widetilde D(z) \subset H^2(\R^3)$. We use to this end the equality~\eqref{eq:widetilde_G}. Let us consider the first term in this equality. A simple computation shows that, for any $\Phi_{N+1} \in \cH_{N+1}$,
\[
\left(A_+\Phi_{N+1}\right)(\br) = \sqrt{N+1} \int_{\R^{3N}} \Phi_{N+1}(\br,\br_1,\dots,\br_N) \Psi_N^0(\br_1,\dots,\br_N) \, \rd\br_1\dots\rd\br_N,
\]
so that $A_+$ is a bounded operator from $H^2(\R^{3N})\cap\cH_N$ to $H^2(\R^3)$. Since $A_+^\ast$ is a bounded operator from $\cH_1$ to $\cH_{N+1}$ and $(z-H_{N+1}+E_N^0)^{-1}$ is a bounded operator from $\cH_{N+1}$ to $H^2(\R^{3N})\cap\cH_N$, we deduce that $A_+(z-H_{N+1}+E_N^0)^{-1}A_+^\ast$ is a bounded operator from~$\cH_1$ to $H^2(\R^3)$. Similarly, for any $\Phi_{N-1} \in \cH_{N-1}$,
\[
\left(A_-^\ast \Phi_{N-1}\right)(\br) = \sqrt{N} \int_{\R^{3(N-1)}} \Phi_{N-1}(\br_2,\dots,\br_N) \Psi_N^0(\br,\br_2,\dots,\br_N) \, \rd\br_2\dots\rd\br_N,
\]
so that $A_-^\ast$ is a bounded operator from $\cH_{N-1}$ to $H^2(\R^3)$. This allows us to prove that $A_-^\ast (z+H_{N-1}-E_N^0)^{-1}A_-$ is a bounded operator from~$\cH_1$ to $H^2(\R^3)$. Finally, $\widetilde G(z)$ is a bounded operator from~$\cH_1$ to $H^2(\R^3)$, which proves that $\widetilde D(z) \subset H^2(\R^3)$.

\subsection{Proof of Lemma~\ref{lem:P0}}
\label{proof:lemP0}

Let us first prove that $\left( P^{0,+}_\sym(\tau)\right)_{\tau \in \R}$ defines a bounded causal operator. The proof is similar for $\left( P^{0,-}_\sym(\tau)\right)_{\tau \in \R}$. We rely on the following result.

\begin{lemma} \label{lem:Avcf}
  For all $h \in L^6(\R^3)$, the operator $\gamma_{0,N}^0 h$ is a Hilbert-Schmidt operator on $\cH_1$, and there exists $K \in \R^+$ such that
  \[
  	\forall h \in L^6(\R^3), \quad \left\| \gamma_{0,N}^0 h \right\|_{\fS_2(\cH_1)} \le K \| h \|_{L^6}.
  \]
\end{lemma}

\begin{proof}[Proof of Lemma~\ref{lem:Avcf}]
Since $\gamma_{0,N}^0$ is a projector, for $h \in L^6(\R^3)$, the operator
\[
\overline{h} \gamma_{0,N}^0 \gamma_{0,N}^0 h = \overline{h}(1-\Delta)^{-1/2}  (1-\Delta)^{1/2}\gamma_{0,N}^0(1-\Delta)^{1/2}  (1-\Delta)^{-1/2}h
\]
is the composition of $(1-\Delta)^{1/2}\gamma_{0,N}^0(1-\Delta)^{1/2} \in \fS_1(\cH_1)$ with the bounded operators $(1-\Delta)^{-1/2}h$ and $\overline{h}(1-\Delta)^{-1/2}$. In fact $(1-\Delta)^{-1/2}h \in \fS_6(\cH_1)$ with
\[
\left\| (1-\Delta)^{-1/2}h \right\|_{\fS_6(\cH_1)} \leq K \| h\|_{L^6}
\]
by the Kato-Seiler-Simon inequality~\cite{SS75,Simon2005}. Therefore, $\gamma_{0,N}^0 h \in \fS_2(\cH_1)$ with
\[
\left\| \gamma_{0,N}^0 h \right\|_{\fS_2(\cH_1)} \leq K \| h\|_{L^6},
\]
which concludes the proof.
\end{proof}

We now proceed to the proof of Lemma~\ref{lem:P0}.
We first note that, for $f, g \in C^\infty_c(\R^3, \C)$, it holds
\begin{equation} \label{eq:vc1/2dual}
	\bra f | \vc^{1/2} g \ket_{\cH_1} = \bra \vc^{-1/2} f | g \ket_{\cC} = \bra \vc^{1/2} \overline{f},  g \ket_{\cC', \cC}.
\end{equation}
In particular, for $\tau \in \R_\tau^+$, and for $f,g \in \cH_1$, we get 
\begin{align*}
	\bra f | P^{0,+}_\sym(\tau)  g \ket_{\cH_1} 
		& = \left\bra f \Big| \vc^{1/2} P^{0,+}(\tau) \vc^{1/2} g \right\ket_{\cH_1} 
		= - \ri \Theta(\tau) \left\bra \vc^{1/2} \overline{f} , G_{0, \rp}(\tau) \odot G_{0, \rh}(-\tau) \vc^{1/2} g \right\ket_{\cC', \cC} \\
		& = - \ri \Theta(\tau) \Tr_{\cH_1} \left[G_{0, \rp}(\tau) \left( \vc^{1/2} g \right) G_{0, \rh}(-\tau) \left( \vc^{1/2} \overline{f} \right) \right].
\end{align*}
Let us prove that $G_{0, \rp}(\tau) \left( \vc^{1/2} g \right) G_{0, \rh}(-\tau) \left( \vc^{1/2} \overline{f} \right)$ is indeed a trace-class operator. Replacing $G_{0, \rp}$ and $G_{0, \rh}$ by their expressions found in Proposition~\ref{lemma:time-ordered_G0F}, and owing to the fact that $\gamma_{0,N}^0$ is a projector commuting with $h_1$, we obtain
\begin{align*}
	\left| \bra f | P_\sym^{0,+}(\tau)  | g \ket \right| 
		& = \left| \Tr_{\cH_1} \left( (\mathds{1}_{\cH_1} - \gamma_{0,N}^0 ) \re^{- \ri \tau h_1} \left( \vc^{1/2}  g\right) \gamma_{0,N}^0 \re^{\ri \tau h_1} \gamma_{0,N}^0 \gamma_{0,N}^0 \left( \vc^{1/2} \overline{f} \right) \right) \right| \\
		& \le \left\| (\mathds{1}_{\cH_1} - \gamma_{0,N}^0 ) \re^{- \ri \tau h_1} \right\|_{\cB(\cH_1)} \left\| \re^{\ri \tau h_1} \right\|_{\cB(\cH_1)} \\
			& \qquad \times \left\| \gamma_{0,N} \left( \vc^{1/2} \overline{f} \right) \right\|_{\cB(\cH_1)} \left\| \gamma_{0,N} \left( \vc^{1/2} g \right) \right\|_{\cB(\cH_1)}  \| \gamma_{0,N}^0 \|_{\fS_1(\cH_1)}.
\end{align*}
According to Lemma \ref{lem:vc}, $\vc^{1/2} \overline{f} \in \cC' \hookrightarrow L^6(\R^3)$. Therefore, $\gamma_{0,N}^0 {\left( \vc^{1/2} f \right)} \in \fS_2(\cH_1)$ by Lemma~\ref{lem:Avcf}, hence is bounded, with
\[
\left\| \gamma_{0,N}^0 {\left( \vc^{1/2} f \right)} \right\|_{\cB(\cH_1)} \le
\left\| \gamma_{0,N}^0 {\left( \vc^{1/2} f \right)} \right\|_{\fS_2(\cH_1)} \le C \left\| \vc^{1/2} \overline{f} \right\|_{L^6} \leq \widetilde{C} \| f \|_{\cH_1}.
\]
Similarly, $\left\| \left( \vc^{1/2} g \right) \gamma_{0,N}^0 \right\|_{\cB(\cH_1)} \leq \widetilde{C} \| g \|_{\cH_1}$. Altogether, we found a constant $C \in \R^+$ independent of~$\tau$ such that
\[
	\forall f, g \in \cH_1, \quad \left| \bra f | P_\sym^{0,+}(\tau) | g \ket \right| \le C \| f \|_{\cH_1} \| g\|_{\cH_1},
\]
which proves that $\left( P_\sym^{0,+}(\tau) \right)_{\tau \in \R}$ is a bounded causal operator on $\cH_1$.\\

Let us now prove that $\phi_k$ is a bounded operator from $\cH_1$ to $\cC$. Recall that $\phi_k$ is real-valued and $\phi_k \in H^2(\R^3) \subset L^2(\R^3, \R) \cap L^\infty(\R^3, \R)$ for $1 \le k \le N$. For $f \in \cH_1$, we obtain
\[
	\left\| \phi_k f \right\|_{\cC} \le C \left\| \phi_k f \right\|_{L^{6/5}} \le C \left\| \phi_k \right\|_{L^{3}} \| f \|_{\cH_1},
\]
where $C$ is a constant independent of $f$. The proof that $\phi_k$ is also a bounded operator from $\cC'$ to $\cH_1$ is similar, noticing that $\cC' \hookrightarrow L^6(\R^3)$. We now use~\eqref{eq:gamma0_projector}, and find that, for $f,g \in \cH_1$ and for~$\tau \in \R^+_\tau$,
\begin{align}
	\bra f | P^{0,+}_\sym(\tau) | g \ket 
		& = - \ri \sum_{k=1}^N \Tr_{\cH_1} \left( (\mathds{1}_{\cH_1} - \gamma_{0,N}^0 ) \re^{- \ri \tau h_1} \left( \vc^{1/2}  g\right) | \phi_k \ket  \re^{\ri \tau \varepsilon_k} \bra \phi_k |  \left( \vc^{1/2} \overline{f} \right) \right) \nonumber \\
		& = - \ri \sum_{k=1}^N \Tr_{\cH_1} \left( (\mathds{1}_{\cH_1} - \gamma_{0,N}^0 ) \re^{- \ri \tau (h_1 - \varepsilon_k)} \left| \left( \vc^{1/2}  g\right) \phi_k \right\ket \left\bra \left( \vc^{1/2} f \right)\phi_k \right| \right) \nonumber \\
		& = - \ri \sum_{k=1}^N \left\bra \left( \vc^{1/2} f \right)\phi_k \right| (\mathds{1}_{\cH_1} - \gamma_{0,N}^0 ) \re^{- \ri \tau (h_1 - \varepsilon_k)} (\mathds{1}_{\cH_1} - \gamma_{0,N}^0 ) \left| \left( \vc^{1/2}  g\right) \phi_k \right\ket, \label{eq:fP0+g}
\end{align}
which gives~\eqref{eq:explicitP0}. \\

We finally prove that $P_\sym^{0,-}(\tau) = P_\sym^{0,+}(-\tau)$. Performing similar calculation as for $P_\sym^{0,+}$, we find that, for $f,g \in \cH_1$ and for $\tau < 0$,
\begin{align*}
	\bra f | P^{0,-}_\sym(\tau) | g \ket & = - \ri \sum_{k=1}^N \left\bra \left( \vc^{1/2} \overline{g} \right) \phi_k \right|(\mathds{1}_{\cH_1} - \gamma_{0,N}^0 ) \re^{\ri \tau (h_1 - \varepsilon_k)} \left| \left( \vc^{1/2} \overline{f} \right) \phi_k \right\ket \\
		& =  \bra \overline{g} | P^{0,+}_\sym(-\tau) | \overline{f} \ket.
\end{align*}
For a bounded operator $A \in \cB(\cH_1)$ and for $f,g \in \cH_1$, it holds $
	\bra f | A g \ket = \overline{\bra Ag | f \ket} = \bra \overline{g} | \overline{A^\ast f} \ket
$, so that, since the functions $\phi_k$ are real-valued for $1 \le k \le N$,
\begin{equation} \label{eq:P0-_P0+}
	\bra f | P^{0,-}_\sym(\tau) | g \ket = \left\bra f \Big| \overline{\left( P^{0,+}_\sym (- \tau) \right)^\ast \overline{g}} \right\ket.
\end{equation}
Since $h_1$ is real-valued, in the sense that $h_1 f$ is real-valued whenever $f$ is real-valued, we easily get that
\[
	\forall f \in \cH_1, \quad \overline{(\mathds{1}_{\cH_1} - \gamma_{0,N}^0 ) \re^{ - \ri \tau h_1 } f} = (\mathds{1}_{\cH_1} - \gamma_{0,N}^0 ) \re^{\ri \tau h_1} \overline{f},
\]
so that, 
\[
	\overline{\left( P^{0,+}_\sym (- \tau) \right)^\ast \overline{g}} = P^{0,+}_\sym (- \tau) {g}.
\]
Together with~\eqref{eq:P0-_P0+}, this proves $P^{0,-}(\tau) = P^{0,+}( - \tau)^\ast$.

\subsection{Proof of Proposition~\ref{prop:propertiesP0}}
\label{proof:propertiesP0}

The expression for $\widetilde{P^0_\sym}$ in~\eqref{eq:explicit_widetildeP} comes from the expression for $\widetilde{P^0_\sym}$ in~\eqref{eq:explicit_tildeP+}.
Since for $k \le N$, it holds $\varepsilon_k \le \varepsilon_N$, we obtain
\begin{equation} \label{eq:ineq:h1_epsk}
	\forall 1 \le k \le N, \quad (\mathds{1}_{\cH_1} - \gamma_{0,N}^0 ) (h_1 - \varepsilon_k) \ge \varepsilon_{N+1} - \varepsilon_N > 0.
\end{equation}
From~\eqref{eq:explicit_widetildeP}, we deduce that $\widetilde{P^0_\sym}(\ri \omega)$ is a negative bounded operator for all $\omega \in \R_\omega$. The self-adjointness comes from Remark~\ref{rem:phik}. \\

The bound~\eqref{eq:int_fP0f} is proved similarly as in Lemma~\ref{lem:ReGp}. Let us now prove~\eqref{eq:ineq_P0}. From~\eqref{eq:ineq:h1_epsk}, it holds that, for all $1 \le k \le N$,
\[
	0 \le  (\mathds{1}_{\cH_1} - \gamma_{0,N}^0 ) \dfrac{h_1 - \varepsilon_k}{\omega^2 + (h_1 - \varepsilon_k)^2}
	\le
	\sup_{E \ge \varepsilon_{N+1} - \varepsilon_N} \left( \dfrac{E}{\omega^2 + E^2}  \right) = 
	\left\{ \begin{array}{l}
		\dfrac{1}{2 \omega} \quad  \text{if} \quad \omega \ge \varepsilon_{N+1} - \varepsilon_N \\
		\dfrac{\varepsilon_{N+1} - \varepsilon_N}{\omega^2 + (\varepsilon_{N+1} - \varepsilon_N)^2} \quad \text{otherwise}.
	\end{array} \right.
\]
In particular, there exists a constant $C \in \R^+$ such that
\[
	\forall 1 \le k \le N, \quad \forall \omega \in \R_\omega, \quad 
	0 \le  (\mathds{1}_{\cH_1} - \gamma_{0,N}^0 ) \dfrac{h_1 - \varepsilon_k}{\omega^2 + (h_1 - \varepsilon_k)^2} 
	\le \dfrac{C}{\sqrt{\omega^2 + 1}}.
\]
Using the fact that for $1 \le k \le N$, $\phi_k$ is real-valued, with $\sum_{k=1}^N \phi_k^2 = \rho_{0,N}^0$, we obtain
\[
	\forall \omega \in \R_\omega, \quad  0 \le - \widetilde{P_\sym^0}(\ri \omega) 
	\le \dfrac{2C}{\sqrt{\omega^2 + 1}} \sum_{k=1}^N \vc^{1/2}  \phi_k^2 \vc^{1/2}
	= \dfrac{2C}{\sqrt{\omega^2 + 1}} \vc^{1/2} \rho_{0,N}^0 \vc^{1/2},
\]
which proves~\eqref{eq:ineq_P0}. The fact that $\vc^{1/2} \rho_{0,N}^0 \vc^{1/2}$ is indeed a bounded self-adjoint operator on $\cH_1$ comes from the fact that $v_c^{1/2}$ is a bounded operator from $\cC$ to $\cH_1$ and from $\cH_1$ to $\cC'$, and that the operator of multiplication by $\phi_k$ is a bounded operator from $\cC'$ to $\cH_1$ and from $\cH_1$ to $\cC$. Together with the fact that the function $\omega \mapsto (\omega^2 + 1)^{-1/2}$ is in $L^p(\R_\omega)$ for all $p> 1$, this implies that $\widetilde{P^0_\sym} (\ri \cdot) \in L^p(\R_\omega, \cS(\cH_1))$ for all $p > 1$. The analyticity of this map is straightforward.



\subsection{Proof of Theorem~\ref{th:sumrule_P0}}
\label{proof:sumrule_P0}

Let us prove the equality $2 \sum_{k=1}^N \phi_k (\mathds{1}_{\cH_1} - \gamma_{0,N}^0) (h_1 - \varepsilon_k) \phi_k = - \div \left( \rho_{0,N}^0 \nabla \cdot \right)$, as operators from $\cC'$ to $\cC$. We first note that, since $\phi_k \in L^4(\R^3)$ for $1 \le k \le N$, it holds $\phi_k \phi_l \in \cH_1$ for $1 \le k,l \le N$. In particular,
\begin{align*}
	\sum_{k=1}^N \phi_k \gamma_{0,N}^0 (h_1 - \varepsilon_k) \phi_k = \sum_{k=1}^N \sum_{l=1}^N | \phi_k \phi_l \ket (\varepsilon_l - \varepsilon_k) \bra \phi_l \phi_k | = 0,
\end{align*}
so that
\begin{equation} \label{eq:without_gamma0N}
	2 \sum_{k=1}^N \phi_k (\mathds{1}_{\cH_1} - \gamma_{0,N}^0) (h_1 - \varepsilon_k) \phi_k = 2 \sum_{k=1}^N \phi_k ( h_1 - \varepsilon_k) \phi_k.
\end{equation}
Consider now $f,g \in C^\infty_c(\R^3, \C)$. In view of the equality
\begin{align*}
	(h_1 - \varepsilon_k) \left( \phi_k g\right) = \phi_k \left( - \frac12 \Delta g \right) - \nabla \phi_k \cdot \nabla g,
\end{align*}
it follows that
\begin{align*}
	2 \left\bra f \Bigg| \sum_{k=1}^N \phi_k (h_1 - \varepsilon_k) \phi_k \Bigg| g \right\ket_{\cH_1} 
	& =  \left\bra f \Bigg| \sum_{k=1}^N \phi_k^2 \left( -  \Delta g \right) \right\ket_{\cH_1} - 2\left\bra f \Bigg| \sum_{k=1}^N \phi_k \nabla \phi_k \cdot \nabla g \right\ket_{\cH_1}  \\
		& = \bra f | \rho_{0,N}^0 ( - \Delta g) \ket_{\cH_1} - \bra f |  \nabla \rho_{0,N}^0  \cdot \nabla g \ket_{\cH_1} 
		= \int_{\R^3} \rho_{0,N}^0 \overline{\nabla f} \cdot  \nabla g,
\end{align*}
where we used an integration by part to obtain the last equality. Together with~\eqref{eq:without_gamma0N}, we obtain that the operators $2 \sum_{k=1}^N \phi_k (\mathds{1}_{\cH_1} - \gamma_{0,N}^0) (h_1 - \varepsilon_k) \phi_k$ and  $- \div \left( \rho_{0,N}^0 \nabla \cdot \right)$ are equal on the core~$C^\infty_c(\R^3, \C)$. 
The end of the proof is similar to the one of Theorem~\ref{lem:sumrule_chi}. \\


\subsection{Proof of Lemma~\ref{lem:widetilde_odot_odot}}
\label{proof:widetilde_odot_odot}

It is sufficient to prove that, for any $\nu' < \varepsilon_{N+1}$,  $\nu > \varepsilon_N - \nu'$,   $\omega, \omega' \in \R_\omega$, and~$f, g \in \cH_1$, it holds
\begin{align*} \label{eq:with_widetilde_odot}
	&\Tr_{\cH_1} \left(  \left(v_c^{1/2} \overline{f}\right) \widetilde{G_{0,\rh}}(\nu + \nu' + \ri (\omega + \omega')) \left( v_c^{1/2} g\right) \widetilde{G_{0, \rp}}(\nu' + \ri \omega') \right) \\
	&\qquad 
	=
	\Tr_{\cH_1} \left( \widetilde{G_{0,\rh}}(\nu + \nu' + \ri (\omega + \omega')) \left( v_c^{1/2} g \right) \widetilde{G_{0, \rp}}(\nu' + \ri \omega') \left( v_c^{1/2} \overline{f} \right) \right)
\end{align*}

We will only consider the case $\nu = 0$, $\nu' = \mu_0$, $\omega = \omega' = 0$ for simplicity, the other cases being similar. We rely on the fact that if $A, B \in \cB(\cH_1)$ are such that $AB$ and $BA$ are trace-class operators, then $\Tr(AB) = \Tr(BA)$~\cite{Simon2005}. In our case, we consider $f, g \in C^\infty_c \subset \cH_1$, so that $f_1 := v_c^{1/2}f$ and $g_1 := v_c^{1/2}g$ are in $\cC' \cap L^\infty$, and we set $A = \overline{f_1}$ and $B = \widetilde{G_{0,\rh}}(\mu_0) g_1 \widetilde{G_{0, \rp}}(\mu_0)$. The operators $A$ and $B$ are bounded operators on $\cH_1$. Moreover, from Proposition~\eqref{lemma:time-ordered_G0F}, we get 
\[
	BA = \widetilde{G_{0,\rh}}(\mu_0) g_1 \widetilde{G_{0, \rp}}(\mu_0) \overline{f_1} = \gamma_{0,N}^0 \left( \dfrac{1}{\mu_0 - h_1} g_1 \right) \left( \dfrac{ \mathds{1}_{\cH_1} - \gamma_{0,N}^0 }{\mu_0 - h_1}\overline{f_1} \right).
\]
It holds $\gamma_{0,N}^0 \in \fS_1(\cH_1)$. Also, from the definition of $\mu_0$, it is easy to see that there exists $0 < c \le C < \infty$ such that
\[
	c(1 - \Delta) \le \left| \mu_0  - h_1 \right| \le C (1 - \Delta).
\]
In particular, the operator $(\mu_0 - h_1)^{-1} g_1 = \left[ (\mu_0-h_1)^{-1} (1 - \Delta) \right] \left[(1 - \Delta)^{-1}g_1 \right]$ is the composition of two bounded operators. Actually, the operator $(1 - \Delta)^{-1}g_1$ is in the Schatten class $\fS_6(\cH_1)$, thanks to the Kato-Seiler-Simon inequality~\cite{SS75,Simon2005}, and it holds
\begin{equation} \label{eq:KSS_g}
	\left\| (1 - \Delta)^{-1} g_1 \right\|_{\cB(\cH_1)} \le \left\| (1 - \Delta)^{-1} g_1 \right\|_{\fS_6(\cH_1)} \le C \| g_1 \|_{L^6} \le C' \| g \|_{\cH_1},
\end{equation}
where $C' \in \R^+$ is a constant independent of $g$. Similarly, $(\mathds{1} - \gamma_{0,N})(\mu_0 - h_1)^{-1} \overline{f_1}$ is a bounded operator, satisfying estimates similar to~\eqref{eq:KSS_g}.
Altogether, we deduce that, for all $f,g \in C^\infty_c(\R^3)$,
\[
	\widetilde{G_{0,\rh}}(\mu_0) \left( v_c^{1/2} g \right) \widetilde{G_{0, \rp}}(\mu_0) \left( v_c^{1/2} \overline{f} \right) \in \fS_1(\cH_1),
\]
with
\begin{equation} \label{eq:BA_continuous}
	\left\| \widetilde{G_{0,\rh}}(\mu_0) \left( v_c^{1/2} g \right) \widetilde{G_{0, \rp}}(\mu_0) \left( v_c^{1/2} \overline{f}  \right) \right\|_{\fS_1(\cH_1)} \le C\| g \|_{\cH_1}  \| f \|_{\cH_1},
\end{equation}
where $C \in \R^+$ is a constant independent of $f$ and $g$. The proof that $AB = \widetilde{G_{0,\rp}}(\mu_0) g_1 \widetilde{G_{0, \rh}}(\mu_0) \overline{f}_1$ is trace-class is similar. As a result, we deduce that for any $f, g \in C^\infty_c$, 
\begin{equation*}
	\Tr_{\cH_1} \left(  \left(v_c^{1/2} \overline{f}\right) \widetilde{G_{0,\rh}}(\mu_0) \left( v_c^{1/2} g\right) \widetilde{G_{0, \rp}}(\mu_0) \right)
	=
	\Tr_{\cH_1} \left( \widetilde{G_{0,\rh}}(\mu_0) \left( v_c^{1/2} g \right) \widetilde{G_{0, \rp}}(\mu_0) \left( v_c^{1/2} \overline{f} \right) \right).
\end{equation*}
The proof for the general case $f, g \in \cH_1$ is deduced by density from the estimate~\eqref{eq:BA_continuous}.


\subsection{Proof of Proposition~\ref{prop:fg_fs}}
\label{proof:fg_fs}

Let us first prove that $\fs$ is a bounded linear operator. Let $\widetilde{G^\app}(\mu_0 + \ri \cdot) \in L^2(\R_\omega, \cB(\cH_1))$. For $f,g \in \cH_1$ and $\omega \in \R_\omega$,
\begin{align*}
	\left\bra f \Big| \fs \left[ \widetilde{G^\app}\right](\mu + \ri \omega) \Big| g \right\ket = \bra f | K_x | g \ket - \dfrac{1}{2 \pi} \int_{-\infty}^{+\infty} \Tr_{\cH_1} \left( \widetilde{G^\app} \big( \mu_0 + \ri (\omega + \omega') \big) g \widetilde{W_c^0}( \ri \omega') \overline{f} \right) \rd \omega'.
\end{align*}
Let us first treat the exchange part $K_x$, and prove that it is a Hilbert-Schmidt operator. From the definition~\eqref{eq:Kx}, $K_x$ is an integral operator, and, from Proposition~\ref{prop:Psi_N^0}, its kernel satisfies
\[
\int_{\R^3} \int_{\R^3} \left| K_x(\br, \br')\right|^2 \rd \br \, \rd \br' \le \int_{\R^3} \int_{\R^3} \dfrac{\rho_{0,N}^0(\br) \rho_{0,N}^0(\br')}{| \br - \br'|^2} \rd \br \, \rd \br' < \infty,
\]
where the last inequality comes from the fact that $\rho_{0,N}^0 \ast |\cdot|^{-2} \in L^\infty(\R^3)$ (since $\rho_{0,N}^0 \in L^1(\R^3) \cap L^\infty(\R^3)$, while $|\cdot|^{-2} \in L^1(\R^3) + L^{\infty}(\R^3)$) and $\rho_{0,N}^0 \in L^1(\R^3)$. We conclude that $K_x$ is a Hilbert-Schmidt operator, hence is bounded. \\

For the correlation part, we use the following lemma.

\begin{lemma} \label{lem:gW0f}
	For all $f,g \in \cH_1$ and all $\omega \in \R_\omega$, the operator $g \widetilde{W^0_c}(\ri \omega) \overline{f}$ is trace-class, and
	\begin{equation} \label{eq:ineq_fWg}
		\exists C \in \R^+, \quad \forall f, g \in \cH_1, \quad \left\| g  \widetilde{W^0_c}(\ri \omega) \overline{f} \right\|_{\fS_1(\cH_1)} \le \dfrac{C}{(\omega^2 + 1)^{1/2}} \| f \|_{\cH_1} \| g \|_{\cH_1}.
	\end{equation}
\end{lemma}

\begin{proof}[Proof of Lemma~\ref{lem:gW0f}]
We first prove~\eqref{eq:ineq_fWg} for $f = g \in \cH_1$. Let $\psi \in C^\infty_c(\R^3, \C)$, we have
\[
	\left\bra \psi \left| {f} \widetilde{W^0_c}(\ri \omega) \overline{f} \right|  \psi  \right\ket 
	= \left\bra  \psi  \left| \overline{f} v_c^{1/2} \widetilde{\chi^0_\sym}(\ri \omega) v_c^{1/2}  f \right| \psi  \right\ket
	,
\]
and we infer from~\eqref{eq:estimate_chi0sym} that
\[
	\forall f \in \cH_1, \quad \forall \omega \in \R_\omega, \quad 
	\left\| f \widetilde{W^0_c}(\ri \omega) \overline{f} \right\|_{\fS_1(\cH_1)} \le \dfrac{C}{(\omega^2 + 1)^{1/2}} \left\| f \vc \rho_{0,N}^0 \vc \overline{f} \right\|_{\fS_1(\cH_1)}.
\]
Since
	\[
		\left\| f \vc \rho_{0,N}^0 \vc \overline{f} \right\|_{\fS_1(\cH_1)} = \left\| \sqrt{\rho_{0,N}^0} \vc \overline{f} \right\|_{\fS_2(\cH_1)}^2 = \int_{\R^3} \int_{\R^3} \dfrac{\rho_{0,N}^0 (\br) \left| f(\br') \right|^2}{| \br - \br' |^2} \, \rd \br \, \rd \br' \le C \| f \|_{\cH_1}^2,
	\]
	where we used again the fact that $\rho_{0,N}^0 \ast |\cdot|^{-2}$ is bounded, we obtain that~\eqref{eq:ineq_fWg} holds true for $f = g$. For $f \neq g$, we deduce from the fact that $\widetilde{\chi^0_\sym}(\ri \omega)$ is a bounded self-adjoint negative operator, that
\begin{align*}
	\left\| g \widetilde{W^0_c}(\ri \omega) \overline{f} \right\|_{\fS_1(\cH_1)} 
	& = \left\| g v_c^{1/2} \widetilde{\chi^0_\sym}(\ri \omega) v_c^{1/2} \overline{f} \right\|_{\fS_1(\cH_1)} 
	= \left\| g v_c^{1/2} \sqrt{- \widetilde{\chi^0_\sym}(\ri \omega)} \sqrt{- \widetilde{\chi^0_\sym}(\ri \omega)} v_c^{1/2} \overline{f} \right\|_{\fS_1(\cH_1)} \\
	& \le \left\| g v_c^{1/2} \sqrt{- \widetilde{\chi^0_\sym}(\ri \omega)} \right\|_{\fS_2(\cH_1)} \left\| \sqrt{- \widetilde{\chi^0_\sym}(\ri \omega)} v_c^{1/2} \overline{f} \right\|_{\fS_2(\cH_1)} \\
	& \le \left\| g v_c^{1/2} \widetilde{\chi^0_\sym}(\ri \omega) v_c^{1/2} \overline{g} \right\|_{\fS_1(\cH_1)}^{1/2} 
	\left\| f v_c^{1/2} \widetilde{\chi^0_\sym}(\ri \omega) v_c^{1/2} \overline{f} \right\|_{\fS_1(\cH_1)}^{1/2}\\
	& \le \dfrac{C}{(\omega^2 + 1)^{1/2}} \| f \|_{\cH_1} \| g \|_{\cH_1}.
\end{align*}	
\end{proof}

We now proceed to the proof of Proposition~\ref{prop:fg_fs}. From Lemma~\ref{lem:gW0f}, we get, for $f, g \in \cH_1$,
\begin{align*}
	\left| \left\bra f \Big| \fs_c \left[ \widetilde{G^\app}\right](\mu + \ri \omega) \Big| g \right\ket \right| 
	& =  \dfrac{1}{2 \pi} \left| \int_{-\infty}^{+\infty} \Tr_{\cH_1} \left( \widetilde{G^\app} \big( \mu_0 + \ri (\omega + \omega') \big) g \widetilde{W^0}( \ri \omega') \overline{f} \right) \rd \omega' \right| \nonumber \\
	& \le \dfrac{1}{2 \pi} \int_{-\infty}^{+\infty} \left\|  \widetilde{G^\app}(\mu_0 + \ri (\omega + \omega')) \right\|_{\cB(\cH_1)} \left\| g \widetilde{W^0}( \ri \omega') \overline{f} \right\|_{\fS_1(\cH_1)} \rd \omega' \\
	& \le \int_{-\infty}^{+\infty} \left\|  \widetilde{G^\app}(\mu_0 + \ri (\omega + \omega')) \right\|_{\cB(\cH_1)} \dfrac{C}{(\omega'^2 + 1)^{1/2}} \| f \|_{\cH_1} \| g \|_{\cH_1} \rd \omega' \\
	& \le C' \left\| \widetilde{G^\app}(\mu_0 + \ri \cdot) \right\|_{L^2(\R_\omega, \cB(\cH_1))}        \| f \|_{\cH_1} \| g \|_{\cH_1},
\end{align*}
where we used Cauchy-Schwarz inequality for the last inequality and the fact that $\omega \mapsto (\omega^2 + 1)^{-1/2} \in L^2(\R_\omega)$. Here, $C'$ does not depend on $\omega \in \R_\omega$ nor on $f, g \in \cH_1$. Altogether, we deduce that $\fs$ is a bounded linear operator from $L^2(\R_\omega, \cB(\cH_1))$ to $L^\infty(\R_\omega, \cB(\cH_1))$.\\

We now prove the claimed properties of $\ffg_\lambda$. Consider $M > 0$. By definition of $\mu_0$, the real number $d := \max(\varepsilon_{N+1} - \mu_0, \mu_0 - \varepsilon_N) $ is positive, and $\left| \mu_0 - h_1 \right| \ge d$. Let us choose $0 < \lambda_M < d/ M$. For $0 \le \lambda \le \lambda_M$ and $\widetilde{\Sigma^\app}(\mu_0 + \ri \cdot) \in L^\infty(\R_\omega, \cB(\cH_1))$ such that $\left\|  \widetilde{\Sigma^\app}(\mu_0 + \ri \cdot) \right\|_{L^\infty(\R_\omega, \cB(\cH_1))} \le M$, it holds
\begin{align*}
	\forall \omega \in \R_\omega, \quad  \mu_0 + \ri \omega - h_1 - \lambda \widetilde{\Sigma^\app}(\mu_0 + \ri \omega) 
		=  \left[ \mu_0 + \ri \omega - h_1\right] \left( 1 - \lambda \left[ \mu_0 + \ri \omega - h_1 \right]^{-1} \widetilde{\Sigma^\app}(\mu_0 + \ri \omega) \right).
\end{align*}
Since
\[
	\left\| \lambda \left[ \mu_0 + \ri \omega - h_1 \right]^{-1} \widetilde{\Sigma^\app}(\mu_0 + \ri \omega) \right\|_{\cB(\cH_1)} 
	\le  \dfrac{\lambda}{d}\left\|  \widetilde{\Sigma^\app}(\mu_0 + \ri \omega)\right\|_{\cB(\cH_1)} \le \dfrac{\lambda}{d} M \le \dfrac{\lambda_M}{d} M < 1,
\]
the operator $ 1 - \lambda \left[ \mu_0 + \ri \omega - h_1 \right]^{-1} \widetilde{\Sigma^\app}(\mu_0 + \ri \omega) $ is invertible, with
\[
	\left\| \left( 1 - \lambda \left[ \mu_0 + \ri \omega - h_1 \right]^{-1} \widetilde{\Sigma^\app}(\mu_0 + \ri \omega) \right)^{-1} \right\|_{\cB(\cH_1)} \le \dfrac{d}{d - \lambda_M M}.
\]
Since $\mu_0 + \ri \omega - h_1$ is an invertible operator with $\left\| (\mu_0 + \ri \omega - h_1)^{-1} \right\|_{\cB(\cH_1)} \le  ( \omega^2 + d^2)^{1/2}$, we obtain that $ \mu_0 + \ri \omega - h_1 - \lambda \widetilde{\Sigma^\app}(\mu_0 + \ri \omega) $ is invertible, with
\[
	\left\| \left[ \mu_0 + \ri \omega - h_1 - \lambda \widetilde{\Sigma^\app}(\mu_0 + \ri \omega) \right]^{-1} \right\|_{\cB(\cH_1)} \le \dfrac{d}{d - \lambda_M M} \left\| (\mu_0 + \ri \omega - h_1)^{-1} \right\|_{\cB(\cH_1)} \le \dfrac{K_M}{(\omega^2 + 1)^{1/2}}.
\]
We deduce from this inequality that
\[
	\left\| \ffg_\lambda \left( \widetilde{\Sigma^\app} \right) (\mu_0 + \ri \cdot) \right\|_{L^\infty(\R_\omega, \cB(\cH_1))} + \left\| \ffg_\lambda \left( \widetilde{\Sigma^\app} \right) (\mu_0 + \ri \cdot) \right\|_{L^2(\R_\omega, \cB(\cH_1))} \le C_M,
\]
where the constant $C_M \in \R^+$ does not depend on $ \lambda \in [0, \lambda_M]$. This gives the claimed result.
Finally,~\eqref{eq:resolvent_sigma} is a direct consequence of the resolvent formula.

\subsection{Proof of Theorem~\ref{th:contraction}}
\label{proof:contraction}

Let us denote for simplicity
\[
	\| \fs \| = \| \fs \|_{\cB(L^2(\R_\omega, \cB(\cH_1)), L^\infty(\R_\omega, \cB(\cH_1)))},
\]
and fix $M> \| \fs \| \left\| \widetilde{G^0}(\mu + \ri \cdot) \right\|_{L^2(\R_\omega, \cB(\cH_1))}$. Let $\lambda_M$ and $C_M$ be chosen as in Proposition~\ref{prop:fg_fs} for this choice of $M > 0$, and introduce 
\[
	r = \frac{M}{\| s \|} -  \left\| \widetilde{G^0}(\mu + \ri \cdot) \right\|_{L^2(\R_\omega, \cB(\cH_1))} > 0.
\]
For this choice of $r$, it holds that, for any $\widetilde{G^\app} \in \fB \left( \widetilde{G_0}, r \right)$, 
\[
	\left\| \fs \left[ \widetilde{G^\app} \right](\mu_0 + \ri \cdot) \right\|_{L^\infty(\R_\omega, \fB(\cH_1))} \le M.
\]
Therefore, from Proposition~\ref{prop:fg_fs},  $\ffg_\lambda \circ \fs \left[ \widetilde{G^\app} \right]$ is well-defined for all $\lambda \in [0, \lambda_M]$. \\

Let us prove that there exists $\lambda_\ast > 0$ sufficiently small such that for any $0 \le \lambda \le \lambda_\ast$, $\ffg_\lambda \circ \fs$ maps $\fB \left( \widetilde{G^0}, r \right)$ into itself. For $\widetilde{G^\app} \in \fB \left( \widetilde{G_0}, r \right)$, it holds
\begin{align}
	 &\left\| \left( \ffg_\lambda \circ \fs \left[ \widetilde{G^\app} \right] - \widetilde{G_0} \right)(\mu_0 + \ri \cdot) \right\|_{L^2(\R_\omega, \cB(\cH_1))}\nonumber \\ 
	& \qquad  \le \left\| \left( \ffg_\lambda \circ \fs \left[ \widetilde{G^\app} \right] - \ffg_\lambda \circ \fs \left[ \widetilde{G_0} \right] \right)(\mu_0 + \ri \cdot) \right\|_{L^2(\R_\omega, \cB(\cH_1))} \label{eq:Lipschitz} \\
	& \qquad  \qquad + \left\| \left( \ffg_\lambda \circ \fs \left[ \widetilde{G_0} \right] -  \widetilde{G^0} \right)(\mu_0 + \ri \cdot) \right\|_{L^2(\R_\omega, \cB(\cH_1))}. \label{eq:G1-G0}
\end{align}

To control the first term~\eqref{eq:Lipschitz}, we use the following result.
\begin{lemma} \label{lem:Lipschitz}
The map $\ffg_\lambda \circ \fs$ is $(\lambda C_M^2 \| \fs \|)$-Lipschitz on $\fB \left( \widetilde{G_0}, r \right)$.
\end{lemma}

\begin{proof}[Proof of Lemma~\ref{lem:Lipschitz}] 
Let $\widetilde{G^\app_1}, \widetilde{G^\app_2} \in \fB \left( \widetilde{G_0}, r \right)$. From~\eqref{eq:resolvent_sigma}, we obtain
\begin{equation} \label{eq:Lipschitz_resolvent}
	\ffg_\lambda \circ \fs \left[  \widetilde{G^\app_1} \right] - \ffg_\lambda \circ \fs \left[ \widetilde{G^\app_2} \right] =  \lambda \left( \ffg_\lambda \circ \fs  \left[ \widetilde{G^\app_1} \right] \right) \left( \fs \left[ \widetilde{G^\app_2} \right] - \fs \left[ \widetilde{G^\app_1} \right] \right) \left( \ffg_\lambda \circ \fs \left[ \widetilde{G^\app_2} \right] \right).
\end{equation}
From Proposition~\ref{prop:fg_fs},
\[
	\left\| \ffg_\lambda \circ \fs \left[ \widetilde{G^\app_j} \right] (\mu_0 + \ri \cdot ) \right\|_{L^\infty(\R_\omega, \cB(\cH_1))} \le C_M
	 \quad  \text{for} \quad
	1 \le j \le 2.
\]
Moreover,
\[
	\left\| \left( \fs \left[ \widetilde{G^\app_2} \right] - \fs \left[ \widetilde{G^\app_1} \right] \right)(\mu_0 + \ri \cdot) \right\|_{L^\infty(\R_\omega, \cB(\cH_1)} \le \| \fs \| \left\| \left( \widetilde{G^\app_2} - \widetilde{G^\app_1} \right) (\mu_0 + \ri \cdot) \right\|_{L^2(\R_\omega, \cB(\cH_1))}.
\]
Plugging these estimates into~\eqref{eq:Lipschitz_resolvent}, we obtain
\begin{align*}
	& \left\| \left( \ffg_\lambda \circ \fs \left[ \widetilde{G^\app_1} \right] - \ffg_\lambda \circ \fs \left[ \widetilde{G^\app_2} \right] \right)(\mu_0 + \ri \cdot) \right\|_{L^2(\R_\omega, \cB(\cH_1))}  \\
	& \qquad \le \lambda C_M^2 \| \fs \| \left\| \left( \widetilde{G^\app_2} - \widetilde{G^\app_1} \right)(\mu_0 + \ri \cdot) \right\|_{L^2(\R_\omega, \cB(\cH_1))},
\end{align*}
which proves that $\ffg_\lambda \circ \fs$ is $(\lambda C_M^2 \| \fs \|)$-Lipschitz on $\fB \left( \widetilde{G_0}, r \right)$.
\end{proof}

Let us now control~\eqref{eq:G1-G0}. By noting that $\ffg_{\lambda = 0} \circ \fs \left[ \widetilde{G_0} \right] = \widetilde{G_0}$, we get from the resolvent formula that
\[
	\ffg_\lambda \circ \fs \left[ \widetilde{G_0} \right] - \widetilde{G_0} 
	=
	 \left(\ffg_\lambda - \ffg_0 \right) \circ \fs \left( \widetilde{G_0} \right) = \lambda \left( \ffg_\lambda \circ \fs \left[ \widetilde{G_0} \right]\right) \left( \fs  \left[ \widetilde{G_0} \right]\right)  \widetilde{G_0}.
\]
Using estimates similar to the ones used in the proof of Lemma~\ref{lem:Lipschitz}, we deduce that
\begin{equation} \label{eq:control_G1-G0}
	\left\| \left( \ffg_\lambda \circ \fs \left[ \widetilde{G_0} \right] - \widetilde{G_0} \right)(\mu_0 + \ri \cdot) \right\|_{L^2(\R_\omega, \cB(\cH_1))} 
	\le \lambda C_M^2 \left\| \fs \left[ \widetilde{G_0} \right](\mu_0 + \ri \cdot) \right\|_{L^\infty(\R_\omega, \cB(\cH_1))}.
\end{equation}

From Lemma~\ref{lem:Lipschitz} and~\eqref{eq:control_G1-G0}, we arrive at the conclusion that for all $0 \le \lambda \le \lambda_\ast$, where
\[
	\lambda_\ast = \dfrac{1}{ C_M^2 \left( \| \fs \| r +  	\left\| \fs \left[ \widetilde{G_0} \right](\mu_0 + \ri \cdot) \right\|_{L^\infty(\R_\omega, \cB(\cH_1))} \right)},
\]
it holds $\ffg_\lambda \circ \fs \left( \fB \left( \widetilde{G_0}, r \right) \right) \subset \fB \left( \widetilde{G_0}, r \right)$. \\

Finally, without loss of generality, we can assume that $\lambda_\ast C_M^2 \| \fs \| < 1$, so that, from Lemma~\ref{lem:Lipschitz}, we get that for all $0 \le \lambda \le \lambda_\ast$, the map $\ffg_\lambda \circ \fs$ is a contraction. The end of the proof follows from Picard's fixed point theorem.

\subsection*{Acknowledgements}

We are grateful to the numerous physicists who patiently explained to us various aspects of the GW formalism: Xavier Blase, Christian Brouder, Fabien Bruneval, Ismaela Dabo, Claudia Draxl, Lucia Reining, Gian-Marco Rignanese, Patrick Rinke, Pina Romaniello, Francesco Sottile and Paulo Umari. Several useful discussions took place during the IPAM program on \textit{Materials for a sustainable energy future} and the IPAM hands-on summer school on \textit{Electronic structure theory for materials and biomolecules}.



\bibliographystyle{plain}
\bibliography{GW}

\end{document}